\documentclass[letterpaper,USenglish,thm-restate]{lipics-v2021} 

\usepackage{multicol}
\usepackage{amsthm}
\usepackage{amssymb}
\usepackage{algorithmicx}
\usepackage[utf8]{inputenc}
\usepackage{algorithm}
\usepackage{algpseudocode}
\usepackage{graphicx}
\usepackage{float}
\usepackage{color}
\usepackage[dvipsnames]{xcolor}
\usepackage{paralist}
\usepackage[title]{appendix}
\usepackage{systeme}
\usepackage{lmodern}
\usepackage{hyperref}
\usepackage{thmtools}
\usepackage{pifont}

 \newcommand{\rmv}[1]{}

\newcommand{\br}{test-or-set}

\newcommand{\mar}{verifiable register}
\newcommand{\marc}{Verifiable Register}

\newcommand{\tar}{authenticated register}

\newcommand{\sr}{sticky register}
\newcommand{\src}{Sticky Register}
\newcommand{\precedes}{precedes}

\newcommand{\Set}{\textsc{Write}}
\newcommand{\Test}{\textsc{Read}}
\newcommand{\rSet}{\textsc{Set}}
\newcommand{\rTest}{\textsc{Test}}
\newcommand{\valid}{\textsc{Verify}}
\newcommand{\sign}{\textsc{Sign}}
\newcommand{\validp}[1]{\textcolor{black}{{\valid}\text{$#1$~execution}}}
\newcommand{\validpg}[1]{\textcolor{black}{{\valid}\text{$#1$~execution}}}

\newcommand{\set}{\mathit{set}}
\newcommand{\setb}{\mathit{set_{\bot}}}
\newcommand{\setv}{\mathit{set_{val}}}

\newcommand{\fresh}{\textbf{Help}}

\newcommand{\AW}{\mathcal{I}}

\newcommand{\iln}[1]{in}

\newcommand{\ct}{\text{correct}}
\newcommand{\hct}{\textsc{correct}}

\newcommand{\wof}{\text{witness of}}
\newcommand{\wsof}{\text{witnesses of}}

\newcommand{\Done}{\textsc{done}}

\newcommand{\done}{\textsc{done}}
\newcommand{\true}{\textsc{true}}
\newcommand{\false}{\textsc{false}}
\newcommand{\success}{\textsc{success}}
\newcommand{\fail}{\textsc{fail}}

\algdef{SE}[SUBALG]{Indent}{EndIndent}{}{\algorithmicend\ }
\algtext*{Indent}
\algtext*{EndIndent}

 \title{\hspace{3.5cm}You can lie but not deny:\\\hspace{1.3cm}SWMR registers with signature properties\\\hspace{1.9cm}in systems with Byzantine processes}

\author{\hspace{4.2cm}Xing Hu \qquad Sam Toueg}{~\\\hspace{2.3cm}Department of Computer Science, University of Toronto, Canada}{}{}{}
\titlerunning{}
\authorrunning{ }
\begin{document}
\nolinenumbers


\maketitle 

\begin{abstract}

We define and show how to implement SWMR registers that provide properties of unforgeable digital signatures — without actually using such signatures — in systems with Byzantine processes.
Intuitively, processes can use these registers
    to write values as if they are ``signed'',
    such that these ``signed values'' can be ``verified'' by any process and ``relayed'' to any process.
All our register implementations are from SWMR registers, and they work in systems
	with $n > 3f$ processes, $f$~of which can be Byzantine.
We show that these implementations are optimal in the number of Byzantine processes they can tolerate:
	more precisely, we prove that if $3 \le n \le 3f$,
	the registers that we propose
	cannot be implemented from SWMR registers
	without using signatures.
    
 The registers that we introduce in this paper can also be implemented without signatures
 	in \emph{message-passing} systems with $n > 3f$ processes, $f$~of which can be Byzantine:
	this is because SWMR registers can be implemented in such systems
	(Mostéfaoui, Petrolia, Raynal, and Jard 2017).
	    
\end{abstract}

\section{Introduction}\label{Introduction}

We consider systems with Byzantine processes that communicate via single-writer multi-reader (SWMR) registers.
In such systems, a faulty process $p$ can first write a value $v$ that is read by some correct process $q$,
	and then it can erase $v$ and (intuitively) ``deny'' that it ever wrote $v$:
	now $q$ may be unable to prove to other processes that it actually read $v$ from $p$.
This difficulty in determining the authorship of a value $v$ greatly complicates
	the design of fault-tolerant algorithms in systems with Byzantine processes.
This is why several algorithms in the literature assume the existence of
	\emph{unforgeable digital signatures}\footnote{Unforgeable
	digital signatures are typically implemented using cryptographic techniques (such as Public-Key Cryptography).
	The assumption is that ``forging'' these signatures requires solving
	some computational problem that is known to be hard.} that processes can use to sign the values that they write (e.g., \cite{Aguilera2019, CohenKeidar2021}).
With this assumption, a process $q$ that reads a value $v$ signed by $p$
	can determine whether $p$ indeed wrote $v$ by \emph{verifying $p$’s signature}. 
Furthermore, $q$~can ``relay'' the value $v$ signed by $p$ to other processes so
	that they can also independently verify the signature and determine that $p$ wrote $v$.
So after a process $p$ properly signs a value $v$,
	and the signature is verified by some correct process, 
	$p$ can no longer deny that it has written this value.

In this paper, we introduce three
	types of SWMR registers that provide
	the above properties of unforgeable digital signatures, and show how to implement them from plain SWMR registers
	\emph{without using signatures}.
These implementations work in systems with $n > 3f$ processes,
	$f$ of which can be Byzantine,
	which we show is optimal in terms of fault-tolerance.
We now describe these registers. 

\smallskip
\textbf{Verifiable Registers.}
     A SWMR \emph{verifiable} register has $\Test$ and $\Set$ operations that behave exactly
    as with a ``normal'' SWMR register.
But in addition to $\Test$ and $\Set$, it also has $\sign$ and $\valid$ operations that
    can be used to emulate properties of unforgeable signatures. Intuitively,
    the writer $p$ of this register can sign any value that it previously wrote, and the readers can verify
    the signature of any value that may have been signed by $p$, as follows:

\rmv{ 
\begin{compactitem}
\item Process $p$ can execute $\sign(v)$ to ``sign'' any value $v$ that it has previously written. 

\item Any process $q$ can execute $\valid(v)$ on any value $v$ to determine whether
    $p$ indeed
        signed $v$: 
	$\valid(v)$ returns $\true$ if and only if $p$ previously signed $v$.
    
\item The $\valid$ operation also provides the ``relay'' property of signed values as follows:
	if any correct process executes $\valid(v)$ and this returns $\true$,
	then every correct process that executes $\valid(v)$ thereafter is guaranteed to also get $\true$.
\end{compactitem}
}

\begin{compactitem}
\item The writer $p$ can execute $\sign(v)$ to ``sign'' any value $v$ that it has previously written. 

\item Any reader $q$ can execute $\valid(v)$ on any value $v$ to determine whether
    $p$ indeed signed~$v$: 
	$\valid(v)$ returns $\true$ if and only if $p$ previously signed $v$.

\end{compactitem}

Note that the above property of the $\valid$ operation
	provides the ``relay'' property of signed values because
	if any correct reader executes $\valid(v)$ and this returns $\true$,
	then every correct reader that executes $\valid(v)$ thereafter is guaranteed to also get $\true$.
This holds even if the writer is Byzantine.
So if it writes a value $v$ into a {\mar} and signs it, the writer cannot later ``deny'' that it did so:
	all correct readers would be able to determine that the writer is lying
	by simply executing $\valid(v)$ and checking its response.
	
\rmv{
With a {\mar}, the writer is not required to sign every value that it writes,
    or to sign a value immediately after writing it:
    a writer may choose to sign only a subset of the values that it writes,
    and it is allowed to sign any of the values that it previously wrote, even older ones.}

We give the precise specification of {\mar}s in Section~\ref{marc}, and describe how to implement them from SWMR registers without using signatures in Section~\ref{marc-imp}.

 \smallskip
\textbf{Authenticated Registers.}
 With a {\mar}, $\Set$ and $\sign$ are \emph{separate} operations.
 So there is a ``gap'' between the writing of a value $v$ and its subsequent signing by the writer (if the writer decides to actually sign $v$).
Thus a {\mar} allows the following undesirable scenario:
	the writer writes a value $v$
	\emph{that it intends to sign immediately after}, then the writer slows down,
	a reader now reads $v$ from the register, but when it subsequently performs
	a $\valid(v)$ operation to check whether $v$ was \emph{signed},
	the reader gets $\false$ (because the writer did not sign $v$ yet).
 
 This delay between the writing of a value and its signing may be problematic for some algorithms that
 	use digital signatures to write \emph{only signed values} in their registers:
 	intuitively, in these algorithms
	the writing \emph{and} signing of a value
 	into a register is an \emph{atomic} operation (and
	this atomicity may be crucial to the algorithms correctness).
 For such algorithms, we propose a variation of the {\mar} that we call an \emph{\tar}.
 
 A SWMR \emph{authenticated} register has three operations, namely
 	$\Set$, $\Test$, and $\valid$.
 Intuitively:
 
 \begin{compactitem}
 \item When the writer executes a $\Set$ to write a value $v$ into the register,
		it is \emph{as if} $v$
		is automatically signed with the writer's signature
		at the moment of the writing.
		
		So, in contrast
		to a {\mar}, with an {\tar} the writting and ``signing'' of each value is atomic.
		Moreover, since \emph{every} value that is written is ``automatically signed'',
		now there is no need for a separate $\sign(-)$ operation.

 \item When a reader executes a $\Test$,
		it reads the latest value that was written. 

 \item To determine whether some value $v$ was previously written,
 	any reader
 	can execute $\valid(v)$: this operation 
	returns $\true$ if and only if $v$ was written 
	by the writer.
 \end{compactitem}

As with a {\mar}, this $\valid$ operation provides the ``relay'' property of signed values: 
	if any correct reader executes $\valid(v)$ and this returns $\true$,
	from this point on every correct reader that executes $\valid(v)$ is also guaranteed to~get~$\true$.
	
The precise specification of an {\tar} is given in Section~\ref{tar-sec},
	and its implementation from SWMR registers without using signatures is described in Section~\ref{tar-imp}.

\smallskip
\textbf{Sticky Registers.}
The use of signatures alone, however, is not sufficient to prevent another
	possible disruptive behaviour by Byzantine processes, as we now explain.
Consider a situation where each process is supposed to write a \emph{unique} value
	into one of the SWMR registers that it owns
	(e.g., this register holds the process proposal value in a consensus algorithm).
Even with unforgeable digital signatures, a Byzantine process can easily violate this ``uniqueness'' requirement
	by successively writing several \emph{properly signed values} into this register
	(e.g., it could successively propose several different values to try to foil consensus).

To prevent this malicious behaviour, processes could be required to use
	   SWMR \emph{sticky registers} for storing values that should be unique:
	these registers have {\Set} and {\Test} operations,
	but once a value is written into a sticky register, the register never changes its value again.\footnote{This is akin to the \emph{MWMR} sticky bit defined in~\cite{Plotkin89}.}
	So if any correct process reads a value $v \neq \bot$ from a sticky register $R$
	(where $\bot$ is the initial value of $R$),
	then every correct process that reads $R$ thereafter is guaranteed to read $v$.
Because $R$ is sticky, this holds even if the writer of $R$ is Byzantine.
Note that this ``stickiness'' property also gives the ``relay'' property and thus prevents ``deniability'':
	once any correct process reads a value $v \neq \bot$ from a SWMR \emph{sticky register} $R$,
	the writer of $R$ cannot deny that it wrote $v$ because all processes can now read $v$ directly from $R$.

Sticky registers can be easily used to broadcast a message
	with the following ``uniqueness'' property
	(also known as \emph{non-equivocation} in \cite{clement2012}):
	to broadcast a message $m$,
	a process $p$ simply writes $m$ into a SWMR sticky register $R$;
	to deliver $p$’s message, a process reads $R$ and
	if it reads some $m \neq \bot$, it delivers $m$.
Because $R$ is sticky, once any correct process delivers a message $m$ from $p$,
	every correct process that subsequently reads $R$ will also~deliver~$m$.\linebreak
So correct processes cannot deliver different messages from $p$, \emph{even if $p$ is Byzantine}.

We give the precise specification of {\sr}s in Section~\ref{srsec}, and describe how to implement them from SWMR registers without using signatures in Section~\ref{srsec-imp}.

\textbf{Fault-tolerance.} All our register implementations (which work
	in systems with $n > 3f$ processes, $f$ of which can be Byzantine)
	are optimal in terms of fault-tolerance.
More precisely, we prove that if $3 \le n \le 3f$,
	{\tar}s, {\mar}s and {\sr}s
	cannot be implemented from SWMR registers
	without using signatures.\footnote{When $n=2$ there is only one writer and one reader.
	In this special case, it is trivial to implement
	all the registers that we consider here using SWMR registers.}
To prove this, we consider a \emph{{\br}} object 
    that can be set by a single process and tested by multiple processes.
We show that in a system with Byzantine processes,
	this object cannot be implemented from SWMR registers without signatures if $3 \le n \le 3f$,
    	but it can be implemented from any one of the registers that we propose.

Using our register implementations, one can transform several algorithms
	that assume unforgeable digital signatures, into signature-free algorithms.
For example, they can be used to remove signatures from the implementations
	of the \emph{atomic snapshot} and \emph{reliable broadcast} objects
	described by Cohen and Keidar in \cite{CohenKeidar2021}. 
This gives the first known implementations of these objects in systems with Byzantine processes without signatures.
It should be noted, however,
    that while the implementations given in \cite{CohenKeidar2021}
    work for systems with $n>2f$ processes,
    their translated signature-free counterparts
    work only for systems with $n>3f$ processes.

Finally, we note that since SWMR registers can be implemented in \emph{message-passing} systems
    with $n > 3f$ processes, $f$~of which can be Byzantine~\cite{Mostefaoui2016},
    {\mar}s, {\tar}s, and {\sr}s can also
    be implemented in these message-passing systems without using signatures.
    
\smallskip
\textbf{Remark.}
A remark about the use of signatures is now in order.
Note that a process that reads a value $v$ from a SWMR~register~$R$
	does \emph{not} need signatures to identify who wrote $v$ into $R$:
	it knows that $v$ was written by the single process that can write $R$
	(namely, the ``owner'' of $R$).
This is because no process, even a Byzantine one, can access
	the ``write port'' of any SWMR register that it does not own.
So what are signatures used for? To the best of our knowledge,
	they are typically used to provide the following ``unforgeability'' and ``relay'' properties:
if a correct process $q$ reads {\emph{any} register (not necessarily one that is owned by $p$)
	and sees a value $v$ signed by $p$,
	then, after verifying $p$'s signature, $q$ is sure that $p$ indeed wrote $v$;
furthermore, $q$ can relay this signed value to other processes
	so they can also independently verify that $p$ wrote $v$.
As we explained above, {\mar}s, {\tar}s, and {\sr}s provide
	these two properties of signatures, but they can be implemented
	from SWMR registers
	without signatures.

\smallskip
\textbf{Roadmap}.
We briefly discuss related works in~Section \ref{relatedwork}, and our model in Section~\ref{model}.
We define {\mar}s in Section~\ref{marc}, and show how to implement them in Section~\ref{marc-imp}.
We specify {\tar}s in Section~\ref{tar-sec},
	and describe their implementation in Section~\ref{tar-imp}.
We define {\sr}s in Section~\ref{srsec}, and explain how to implement them in Section~\ref{srsec-imp}.
In Section~\ref{Impossibility-Result}, we define the {\br} object and use it to prove
    that our implementations of {\mar}s, {\tar}s, and {\sr}s are optimal
    in the number of Byzantine processes that they can tolerate.
We conclude the paper with some remarks in Section~\ref{conclusion}.
We prove that our implementations of {\mar}s, {\tar}s,  and {\sr}s are correct in Appendix~\ref{a-mar},~\ref{a-tar}, and~\ref{a-sar}, respectively.

\section{Related Work}\label{relatedwork}

Herlihy and Wing~\cite{linearizability} defined \emph{linearizability} --- a central concept of implementation correctness ---
    but their definition applies only to systems with crash failures.
To the best of our knowledge, 
Most\'efaoui, Petrolia, Raynal, and Jard~\cite{Mostefaoui2016} were the first to extend linearizability for the implementation of SWMR registers in systems with Byzantine processes.
Cohen and Keidar \cite{CohenKeidar2021} were the first to extend linearizability for \emph{arbitrary objects} in systems with Byzantine processes, and they called this extension \emph{Byzantine linearizability}.

In~\cite{CohenKeidar2021} Cohen and Keidar also gave Byzantine linearizable implementations of
    \emph{reliable broadcast}, \emph{atomic snapshot}, and \emph{asset transfer}
    objects from SWMR registers; their algorithms use signatures
    and work in systems with $n>2f$ processes.
The properties of signatures that the algorithms in~\cite{CohenKeidar2021}
    rely on are also provided
    by the properties of {\mar}s, {\tar}s, or {\sr}s,
    and so one can use our implementations of any of these registers
    to give signature-free counterparts of the algorithms in~\cite{CohenKeidar2021}
    that work for systems with $n>3f$ processes and SWMR registers.
In~\cite{Aguilera2019}, Aguilera, Ben-David, Guerraoui, Marathe and Zablotchi
    implemented \emph{non-equivocating} broadcast ---
    a weaker variant of reliable broadcast --- from SWMR \emph{regular} registers;
    their algorithm also uses signatures and works in systems with $n>2f$ processes, $f$ out of which can be Byzantine.

The results of this paper are, in some sense, the shared-memory counterparts
    of results by Srikanth and Toueg~\cite{SrikanthToueg} for message-passing systems with Byzantine processes.
In that work, they showed how to provide properties of
    signatures in message-passing systems without actually using signatures,
    and in particular, they showed how to implement \emph{authenticated broadcast}
    with a signature-free algorithm.
Since with SWMR registers one can simulate message-passing systems,
    one may think that it is possible to obtain the results of this paper
    by transforming the signature-free algorithm 
    that provides properties of signatures in message-passing systems (given in~\cite{SrikanthToueg})
    into a signature-free algorithm 
    that provides properties of signatures in shared-memory system with SWMR registers.
However, this does not work.
In a nutshell, with the asynchronous authenticated broadcast algorithm given in~\cite{SrikanthToueg}, message delivery is only \emph{eventual} and this does not
    give the linearizability property needed in shared-memory systems.

Most\'efaoui, Petrolia, Raynal, and Jard~\cite{Mostefaoui2016} provided an implementation of
    reliable broadcast in message-passing systems with Byzantine processes 
    and their implementation works for systems with $n>3f$ processes. 
However, simulating their message-passing algorithm for reliable broadcast 
    using SWMR registers does not yield a \emph{linearizable}
    implementation of a reliable broadcast \emph{object} (defined in~\cite{CohenKeidar2021}).
This is due to the eventual delivery in message-passing systems that we mentioned earlier.

Clement, Junqueira, Kate, and Rodrigues~\cite{clement2012} proposed the concept of \emph{non-equivocation} for message-passing systems with Byzantine processes.
Intuitively, {non-equivocation} prevents a process from ``sending different messages to different replicas in the same round while it was supposed to send the same message according to the protocol''. This property is similar to the  ``uniqueness'' property of a SWMR {\sr} in shared-memory systems: once this register is written, no two processes can read different values. In~\cite{clement2012} they also show that the combined capability of non-equivocation
    and \emph{transferable authentication} (such as signatures),
    can be used to transform algorithms designed to tolerate
    crash failures into algorithms that tolerate Byzantine failures in message-passing systems.

\section{Model Sketch}\label{model} 

We consider systems with $n$ asynchronous processes that communicate via SWMR registers and are subject to Byzantine failures.
A~\emph{correct} process takes infinitely many steps and executes its algorithm exactly as prescribed.
A process that is not correct is \emph{Byzantine}.
Note that a Byzantine process
    can behave arbitrarily, and in particular
    it can deviate in any way from the algorithm it is supposed to execute.

\subsection{Object Implementations}

Roughly speaking, the \emph{implementation} of an object from a set of registers
	is given by operation procedures that each process can execute
	to perform operations on the implemented object; these procedures can access the given registers.
So each operation on an implemented object \emph{spans an interval} that starts
	with an \emph{invocation}
	and completes with a corresponding \emph{response}.
An operation that has an invocation but no response is \emph{incomplete}.
In a history of an implemented object, each correct process invokes operations sequentially, where steps of different processes are interleaved.

\begin{definition}
Let $o$ and $o'$ be any two operations.
\begin{itemize}
\item $o$ \emph{precedes} $o'$ if the response of $o$ is 
	before the invocation~of~$o'$.

\item $o$ \emph{is concurrent with} $o'$ if neither precedes the other.

\end{itemize}
\end{definition}

\subsection{Linearizability of Object Implementations}

Roughly speaking, an implementation $I$ of an object $O$ (of type $T$) is linearizable,
    if processes that apply operations using the implementation $I$
    behave as if they are applying their operations to an \emph{atomic} 
    object $O$ (of type $T$).\footnote{For brevity, henceforth when we refer to an object $O$, we mean an object $O$ of type $T$, i.e., we omit the explicit reference to the \mbox{type $T$ of~$O$.}
    We sometimes we refer to $T$ as the ``sequential specification'' of~$O$.}
So with a linearizable implementation, every operation appears
	to take effect instantaneously at some point
	(the ``linearization point'') in its execution interval.

The formal definition of linearizability for systems with crash failures is given in~\cite{linearizability}.
    We informally describe it below.

\begin{definition}
    A history $H_c$ is a \emph{completion} of a history $H$ of an implementation
    if it can be obtained from $H$ as follows. For each incomplete operation $o$ in $H$, either remove $o$, or add a (corresponding) response of $o$ \mbox{(after the invocation of $o$).}
\end{definition}

\begin{definition}\label{LinearizableCrash}
A linearization of a history $H$ is a \emph{sequence} of operations $L$ such that
	an operation is in $L$ if and only if it is in $H$.
\end{definition}

\begin{definition}\label{LinearizableImplementationCrash}
    A history $H$ of an implementation is linearizable with respect to an object $O$
    if there is a completion $H_c$ of $H$ 
        with the following property. There is a linearization $L$ of $H_c$ such that:
    \begin{enumerate}
    \item  $L$ respects the precedence relation between the operations of $H_c$. That is: if $o$ precedes $o'$ in $H_c$ then $o$
    precedes
    $o'$ in $L$.
    \item $L$ conforms to the sequential specification of 
    \mbox{the object $O$.}
    \end{enumerate}
\end{definition}

\begin{definition}
    An implementation $\mathcal I$ of an object $O$
    is linearizable if every history $H$ of $\mathcal I$ is linearizable with respect to
    $O$.
\end{definition}

The above definition of linearizability for systems with crash failures does not work for systems with Byzantine failures.
For these systems, Cohen and Keidar define the concept of \emph{Byzantine linearizability} in~\cite{CohenKeidar2021}.
Intuitively, an \emph{implementation} of an object $O$ is Byzantine linearizable if it behaves as if it were an \emph{atomic} object $O$ (implemented in ``hardware'') that cannot be tampered with by any process, even Byzantine ones. With such an atomic object, all the operations are instantaneous, and while Byzantine processes can apply \emph{any} operation allowed by (the type of) the object, they cannot ``break'' it.\footnote{In sharp contrast, Byzantine processes may be able to break a ``software'' \emph{implementation} of an object because they are ``part'' of this implementation and so they can interfere with its functioning.}
So, roughly speaking, in a system with Byzantine failures,
    correct processes can use a \emph{Byzantine linearizable} implementation of an object as if it were an atomic object that can also be accessed by (but \emph{not} broken by) Byzantine processes.\footnote{This rough intuition is not quite accurate: it is known that 
    there are
    \emph{randomized} algorithms that work correctly when they use atomic objects
    but do not work if these objects are replaced with their linearizable implementations \cite{sl11,sl12,sl19,sl21}.}

The formal definition of Byzantine linearizability is given in~\cite{CohenKeidar2021}.
    We informally describe it below.

\begin{definition}
Let $H$ be a history of an implementation and $\hct$ be the set of processes that are correct in $H$.
$H|{\hct}$ is the history consisting of the steps of only the processes in $\hct$ (at the same times they occur in $H$).
\end{definition}

\begin{definition}\label{def-hbl}
    A history $H$ of an implementation is Byzantine linearizable with respect to an object~$O$
    if there exists a history $H'$ so that $H'|{\hct} = H|{\hct}$ and $H'$ is linearizable with respect to~$O$ (according to Definition~\ref{LinearizableImplementationCrash}).
\end{definition}

\begin{definition}\label{LinearizableByz-mar}\label{LinearizableImplementationByz}
An implementation $\mathcal I$ of an object $O$
	is \emph{Byzantine linearizable} if and only if every history $H$ of $\mathcal I$ is Byzantine linearizable with respect to $O$.

\end{definition}

\noindent
Since we consider systems with Byzantine processes, we define implementation correctness as follows:

\begin{definition}
An implementation $\mathcal I$ of an object $O$
	is \emph{correct} if and only for all histories $H$ of $\mathcal I$:
     \begin{itemize}
        \item \textsc{[Byzantine linearizability]} $H$ is Byzantine linearizable with respect to $O$.
        \item \textsc{[Termination]} All processes that are correct in $H$ complete all their operations.
    \end{itemize}
\end{definition}

\subsection{Implementation Steps}

When a process applies an operation to an \emph{implemented} object,
    it executes the implementation procedure that is associated with this operation.
So a process takes implementation steps ``inside'' the [invocation,response]
    intervals of its own operations.
We assume here that processes can also take
    implementation steps ``outside'' their own operation intervals.
Intuitively, they are allowed to take such steps to ``help'' other processes
    that are currently applying their own operations.

\section{\marc}\label{marc}

A SWMR \emph{verifiable} register has $\Test$ and $\Set$ operations that behave
    as with a ``normal'' SWMR register.
 In addition, this register also provides  $\sign$ and $\valid$ operations that work as follows.
 Intuitively, the writer of the register can apply a $\sign(v)$ operation to ``sign''
	any value $v$ that it has previously written.
Moreover, for any value $v$, readers can apply a $\valid(v)$ operation to determine 
    whether  $v$ was signed by the writer.
Note that the writer is not required to sign every value that it writes,
    or to sign a value immediately after writing it:
    a writer may choose to sign only a subset of the values that it writes,
    and it is allowed to sign any of the values that it previously wrote, even older ones.

More precisely, we define the sequential specification of a multivalued SWMR {\mar} where
    the writer can write any value from some domain $\mathcal V$, and the register is initialized to a value $v_o \in \mathcal V$, as follows:

\begin{definition}\label{def-mar}

A SWMR {\mar} has four operations.

\vspace*{1mm}

\textbf{Read and write operations:}

\begin{compactitem}
    \item $\Set(v)$ by the writer takes a value $v\in \mathcal{V}$ and returns $\done$.
    
    \item $\Test$ by any reader returns a value $v \in \mathcal{V}$ such that:
        \begin{itemize}
            \item either there is a $\Set(v)$ before it and this $\Set(v)$ is the last $\Set(-)$ before it,
            \item or $v = v_0$ (the initial value of the register) and there is no $\Set(-)$ before it.
        \end{itemize}
\end{compactitem}

\textbf{Sign and verify operations:}
\begin{compactitem}
      \item $\sign(v)$ by the writer takes a value $v\in \mathcal{V}$, and returns $\success$ or $\fail$.
      
      A $\sign(v)$ returns $\success$ if and only if there is a $\Set(v)$ before it. 
      
     \item $\valid(v)$ by any reader takes a value $v \in  \mathcal{V}$, and returns $\true$ or $\false$.
     
        A $\valid(v)$ returns $\true$ if and only if there is a $\sign(v)$ that returns $\success$~before~it. 
\end{compactitem}

\end{definition}

The above definition implies that a {\mar} has the following properties:

\begin{observation}\label{validity}
    [\textsc{validity}] If the writer executes a $\sign(v)$ that returns $\success$,
    then every subsequent $\valid(v)$ by any reader returns $\true$.
\end{observation}

\begin{observation}\label{unforgeability}
       [\textsc{unforgeability}] If a $\valid(v)$ returns $\true$ to some reader, then the writer previously executed
       a $\sign(v)$ that returned $\success$.
\end{observation}

\begin{observation}\label{relay}
       [\textsc{relay}] If a $\valid(v)$ returns $\true$ to some reader, then every subsequent $\valid(v)$ by any reader also returns $\true$.
\end{observation}

\section{Implementation of a SWMR {\marc}}\label{marc-imp}

We now explain how to implement a {\mar} in a system with SWMR registers and $n>3f$ processes, $f$ of which may be Byzantine.
To do so, we first describe some of the difficulties that must be overcome, and how the algorithm works.

\subsection{Some Challenges and Their Solutions}\label{idea-mar}
To illustrate the main ideas, consider the special case when $n=3f+1$.
When a writer $p$ of a SWMR register $R$ writes a value $v$ into $R$, a process $q$ that reads this value from $R$ cannot immediately accept this value as if it was ``signed'' by $p$: $q$ knows that $p$ wrote $v$ into $R$, but if $p$ erases $v$ and then denies that it wrote this value, $q$ is left with no proof (no ``signature'') that it can use to convince other processes that $p$ actually wrote $v$ into $R$.

So we first adopt an idea that Srikanth and Toueg used in \cite{SrikanthToueg} to ``simulate message authentication'' (i.e., to send messages \emph{as if} they were signed by the sender using digital signatures, but without actually using signatures).
Roughly speaking, when any process reads a value $v$ \emph{directly} from $p$'s register $R$, it becomes a ``witness'' that $p$ wrote $v$ (but it does not necessarily accept $v$ as being ``signed'' by $p$), and then it publicizes that it is now a {\wof} $v$ by writing so in a SWMR register.
When a process sees that a value $v$ has at least $f+1$ witnesses, it knows at least one of them is a \emph{correct} {\wof} $v$, and so it also becomes a {\wof} $v$ (if it was not already one).
When a process $q$ sees that a value $v$ has at least $2f+1$ witnesses, it accepts that $v$ was written and ``signed'' by $p$: if $q$ now invokes $\valid(v)$, this operation should return $\true$.

Implementing the above idea in shared-memory systems, however, runs into a serious difficulty. In a nutshell, this is because in contrast to a message-passing system where message deliveries are ``eventual'', in a share-memory system with registers every invocation of a $\Test$ or $\valid(-)$ operation must complete without waiting for an eventual event \emph{that may never occur}, and it must do so with a response that is \emph{linearizable} with respect to all the $\Set(-)$ and $\sign(-)$ operations.

To illustrate this problem, suppose that a process $q$ invokes $\valid(v)$ to determine whether $p$ previously wrote $v$ into its SWMR register $R$ and ``signed'' it. Using the witnesses idea sketched above, $q$ can ``ask'' all processes whether they are now willing to be witnesses of $v$, and then wait for $2f+1$ processes to reply (it cannot wait for more replies, since $f$ of the $3f+1$ processes can be Byzantine): if at least $2f+1$ processes reply ``Yes''
then $q$'s $\valid(v)$ returns $\true$; and if strictly less than $f+1$ processes reply ``Yes'' then $q$'s $\valid(v)$ returns $\false$.
Intuitively, so far this ``partial algorithm'' ensures that: (a) if $\valid(v)$ returns $\true$, then $p$ previously wrote $v$ into $R$ and ``signed'' it, (b) if $\valid(v)$ returns $\false$, then $p$ did not previously write and ``sign'' $v$, and (c) if $\valid(v)$ returns $\true$, then any subsequent $\valid(v)$ by a correct process can \emph{not} return $\false$.

But what can $q$ do if it invokes $\valid(v)$ and the number $k$ of processes that reply ``Yes'' to its inquiry 
    is less than $2f+1$ but more that $f$ (i.e., if $f <k <2f+1$)?
    If $q$'s $\valid(v)$ returns $\false$, then it is possible that another correct process $q'$ previously invoked $\valid(v)$ and got $\true$ (because some of the $2f+1$ processes that replied ``Yes'' to $q'$, were Byzantine and they later replied ``No'' to $q$);
    and if $q$'s $\valid(v)$ returns $\true$, then a correct process $q'$ may later invoke $\valid(v)$ and get $\false$ (because some of the Byzantine processes that replied ``Yes'' to $q$ later reply ``No'' to $q'$).
Since both outcomes may violate the ``relay'' property of a {\mar} (Observation~\ref{relay}), and $q$'s $\valid(v)$ operation must complete with a $\true$ or $\false$ response, $q$ is now in a bind. To overcome this, $q$ may ask again all processes to recheck whether they are \emph{now} willing to be a {\wof} $v$, hoping that some of their Yes/No ``vote'' changes, but it is easy to see that a Byzantine writer can collude with other Byzantine processes to ensure
    that this second attempt (of asking all processes and waiting for the first $2f+1$ replies) yields exactly the same number $k$ of ``Yes''.

Solving the above issue is not obvious. The key idea is based on the following mechanism:

\begin{itemize}
\item When $q$ invokes $\valid(v)$ it proceeds in rounds. In each round $q$ does \emph{not} ask all the processes (for their Yes/No votes) and then wait for the first $2f+1$ responses to decide what to do next. Instead, to determine which processes to ask and wait for, it keeps track of two sets:
\begin{itemize}
    \item $\set_1$ that consists of all the processes that replied ``Yes'' 
    in \emph{any previous round}, and
    \item $\set_0$ that consists of all the processes that replied ``No'' \emph{after the last round in which a process replied ``Yes''}.
\end{itemize}
These sets are maintained and used as follows.

\item During each round, $q$ asks all processes that are \emph{not} in $\set_0 \cup \set_1$, and
    waits for \emph{one} of them (anyone) to reply.

\begin{itemize}

    \item If this reply is ``No'',
    then $q$ first inserts the process that replied ``No'' into $\set_0$.
    If $|\set_0| > f$ then $\valid(v)$ returns $\false$.
    Otherwise, $q$ goes to the next round.
    
    \item If this reply is ``Yes'',
    then $q$ first inserts the process that replied ``Yes'' into $\set_1$.
    If $|\set_1| > n-f$ then $\valid(v)$ returns $\true$.
    Otherwise, \emph{$q$ resets $\set_0$ to $\emptyset$} and goes to the next round.
    
    Note that in this new round, processes that were in $\set_0$ (because they voted ``No'')
    will be asked for their vote again (because $\set_0$ is reset to $\emptyset$).
        So these processes are given the opportunity to re-check whether they can be {\wsof} $v$,
        and they can change their vote from ``No'' to ``Yes''.
    But processes that replied ``Yes'' in \emph{any} round so far
        will never be asked again
    (because $\set_1$ is \emph{not} reset, in fact it is non-decreasing).
        So these processes can never go back and change their ``Yes'' votes to ``No''.

\end{itemize}
\end{itemize}

But how do we know that this way
    of implementing $\valid(v)$ guarantees \emph{termination}?
What if $q$ waits for any process that is not in $\set_0 \cup \set_1$ to reply,
    but none of them ever reply because they are \emph{all} Byzantine?
For example, suppose that at the beginning of a round $|\set_1| = 2f$ and $|\set_0| = f$ and these sets are disjoint (this can actually happen).
Then there is only \emph{one} process that is not in $\set_0 \cup \set_1$, and so $q$ must wait for its reply.
But what if this process is Byzantine and never replies?

In our proof of the algorithm,
    we show that every $\valid(v)$ by a correct process indeed completes.
To prove this, we show that at the beginning of each round,
    the set of processes that are \emph{not} in $\set_0 \cup \set_1$ contains at least one correct process,
    and so $q$ cannot get stuck waiting forever~for~one~reply.

In the above intuitive description of how the algorithm works, we said that processes tell $p_k$ whether they are {\wsof} a value $v$ or not, by replying ``Yes'' or ``No'' to $p_k$. For simplicity however, in the algorithm that we actually give (see Algorithm~\ref{code-mar}) processes do not reply with ``Yes'' or ``No'', instead they reply with the set of \emph{all} the values they are currently {\wsof}: $p_k$ can see for itself whether $v$ is included in these sets or not.

We now describe the algorithm in detail.

    \begin{algorithm}[ht!]
	\caption{Implementation of a SWMR multivalued \mar~writable by process~$p_1$ (the \emph{writer}) and readable by every process $p_k \in \{p_2,\dots, p_n\}$ (the \emph{readers}), 
		for $n > 3f$.
	}\label{code-mar} 
	\ContinuedFloat
    
        \footnotesize
        \vspace{-4mm}
	\begin{multicols}{2}

	\emph{Every process $p_i \in\{p_1,\dots,p_n\}$ has the following:}

	\hspace{4mm} \ul{shared registers}
	\vspace{.7mm}
		   
        \hspace{1cm}
		$R_{i}$: $\textsf{SWMR}$ register; initially $\emptyset$

	\hspace{1cm}
		For each reader $p_j \in \{p_2,\dots,p_n\}$:
		
		\hspace{1.5cm} $R_{ij}$: $\textsf{SWSR}$ register readable by $p_j$; 

		\hspace{2.2cm} initially $ \langle \emptyset, 0 \rangle$
	
	\columnbreak
	 \footnotesize
        \emph{The writer $p_1$ also has the following:}
 
	\hspace{4mm} \ul{shared register}
	\vspace{.7mm}

	\hspace{1cm}
            $R^*$: $\textsf{SWMR}$ register; initially $v_0$
  
		\hspace{4mm} \ul{local variable}
	
\vspace{.7mm}
        
        \hspace{1,1cm}$r^*$: a set of values; initially $\emptyset$

 \vspace{1.8mm}
    \mbox{\emph{Each reader $p_k \in\{p_2,\dots,p_n\}$ also has the following:}}

	\hspace{4mm} \ul{shared register}
	\vspace{.7mm}

	\hspace{1cm}
		$C_{k}$: $\textsf{SWMR}$ register; initially $0$

	\end{multicols}

	{

        \vspace{-2mm}
	\hrule
	\vspace{-2mm}
	\begin{algorithmic}[1]
	\begin{multicols}{2}

        \Statex \hspace{-5mm}  \textit{$\triangleright$ Executed by the writer $p_1$}
        \vspace{2mm}
	\Statex \textsc{\Set($v$):} 
	
	\Indent
        
	\State \label{r-m}$R^* \gets v$ 
        \State \label{addtoset-m}$r^* \gets r^* \cup \{v\}$ 
	\State \Return $\Done$
	\EndIndent
    
        \vspace{5mm}
        \Statex
        \textsc{\sign($v$):}  
	
	\Indent
    
        \State \label{checksign-m}\textbf{if} $v \in r^*$ \textbf{then}
        \Indent
	\State \label{setter1-m}$R_1 \gets R_1 \cup \{v\}$ 
	\State \label{success-m} \Return $\success$
        \EndIndent
        \State \textbf{else} 
        \Indent
            \State \Return $\fail$
        \EndIndent
	\EndIndent
    
        \vspace{5mm}

	\columnbreak

        \Statex \hspace{-5mm}  \textit{$\triangleright$ Executed by any reader $p_k \in \{p_2, \dots, p_n\}$}
        \vspace{2mm}
        
        \Statex
        \textsc{\Test():} 
	
	\Indent
    
	\State \label{read-m}$v\gets R^*$
	\State \Return $v$
	\EndIndent

        \vspace{5mm}
	\noindent
	\textsc{$\valid$($v$):}
	\Indent
	\State $\set_0, \set_1 \gets \emptyset$
	\State \label{whileloop-m} \textbf{while} true \textbf{do}
	\Indent
	\State \label{ckplus-m} $C_k \gets C_k+1$
	\State \label{repeat} \textbf{repeat}
	\Indent 
	\State \label{findone-m} \textbf{for every} process $p_j \notin \set_1 \cup \set_0$ \textbf{do}
		\Indent
			\State \label{readri-m}$\langle r_j ,c_j \rangle \gets R_{jk}$
		\EndIndent
	\EndIndent
	\State \label{until} \textbf{until} $\exists~ p_j \notin \set_1 \cup \set_0$ s.t. $c_j \ge C_k$
	\State \label{check1-m} \textbf{if} $v \in r_j$ \textbf{then} 
	\Indent
	\State \label{set1-m} $\set_1 \gets \set_1 \cup \lbrace p_j \rbrace$
	\State \label{empty0-m} $\set_0 \gets \emptyset$
	\EndIndent
	\State  \label{notv-mar}\textbf{if} $v \not \in r_j$ \textbf{then}
	\Indent
	\State \label{set0-m} $\set_0 = \set_0 \cup \lbrace p_j \rbrace$
	\EndIndent
	\State \label{return1-m} \textbf{if} $|\set_1| \ge n-f$ \textbf{then} \textbf{return} $\true$
	\State \label{return0-m} \textbf{if} $|\set_0| > f$ \textbf{then} \textbf{return} $\false$

	\EndIndent
	
	\EndIndent
	\end{multicols}
	\vspace{-1mm}
	\hrule
	\vspace{4mm}
	\Statex

        \Statex \hspace{-5mm}  \textit{$\triangleright$ Executed by every process $p_j\in\{p_1,\dots,p_n\}$}
        \vspace{2mm}
        
	\noindent
	\fresh():
        \Indent
        \State \label{collectck-init}\textbf{for every} reader $p_k\in \{p_2, \dots, p_n\}$ \textbf{do} $prev\_c_k \gets 0$ 
	\State \textbf{while} true \textbf{do}
	\Indent
	
	\State \label{collectck-m}\textbf{for every} reader $p_k\in \{p_2, \dots, p_n\}$ \textbf{do} $c_k \gets C_k$  
	\State \label{askers-m}$askers \gets \{ p_k\in \{p_2, \dots, p_n\} ~|~ c_k > prev\_{c_{k}} \}$
	
	\State\label{replyasker-m}\textbf{if} $askers \neq \emptyset$ \textbf{then} 
    
	\Indent 
	
	\State \label{sets-m} \textbf{for} every process $p_i \in\{p_1,\dots,p_n\}$ \textbf{do} $ r_i \gets R_i$
	
	\State \label{followcondition-m} \textbf{for} each value $v$ such that $v \in r_1$ or $|\{ r_i|~ v \in r_i\}| \ge f+1$ \textbf{do}
        
		\Indent
	\State \label{follow-m}$R_j \gets R_j \cup \{v\}$ 
        \State \label{readrj-m} $r_j \gets R_j$
       
	\EndIndent

	\State \label{tellasker-m}\textbf{for every} $p_k \in askers$ \textbf{do}
	
	\Indent 
	
	\State \label{fresh1-m} $R_{jk} \gets \langle r_j, c_k \rangle$
	\State \label{setprev-m}$prev\_c_k \gets c_k$
	\EndIndent
	\EndIndent
	
	\EndIndent
	\EndIndent
	
	\end{algorithmic}
	}
	\vspace{2mm}
	\end{algorithm}

\subsection{The algorithm}\label{algo-mar}

Algorithm~\ref{code-mar} gives a linearizable implementation of a SWMR {\mar} that is writable by process $p_1$ (the writer) and readable by the set $\{p_2,...,p_n\}$ of processes (the readers).
It works for systems 
    with $n >3f$ processes, $f$ of which can be Byzantine, as follows.

\begin{itemize}
\item
To \emph{write} a value $v$, 
    the writer $p_1$ calls the $\Set(v)$ procedure.
    In this procedure, $p_1$ writes $v$ in the shared register $R^*$ and inserts $v$ into the local set $r^*$.
\item
To \emph{read} a value, 
    a reader $p_k\in \{p_2,...,p_n\}$ calls the $\Test()$ procedure.
    In this procedure, $p_k$ reads the value stored in $R^*$ and returns it.
\item
To \emph{sign} a value $v$, 
    the writer $p_1$ calls the $\sign(v)$ procedure. 
    In this procedure, $p_1$ first checks whether $v$ is in the local set $r^*$,
    i.e., $p_1$ checks whether it has previously written $v$.
    If $v$ is in $r^*$,
        $p_1$ inserts $v$ into (the set of values stored in) the shared register $R_1$; by doing so
        $p_1$ considers that it has now ``signed $v$'',
        and it returns $\success$.
    If $v$ is not in $r^*$,
        $p_1$ returns $\fail$.
\item To \emph{verify} a value $v$, 
    a reader $p_k\in \{p_2,...,p_n\}$ calls the $\valid(v)$ procedure. 
    In this procedure, $p_k$ first initializes $\set_1$ and $\set_0$ to empty.

        After this initialization, $p_k$ enters a while loop. 
        Each iteration of this loop is a ``round''.
        
    In each round:
    \begin{itemize}
    \item  First $p_k$ increments a shared SWMR register $C_k$.

    \item Then $p_k$ enters a repeat-until loop.
        In this loop, $p_k$ reads the tuple $\langle r_j,c_j\rangle$
        from the SWSR register $R_{jk}$
            \emph{of every process $p_j$ that is not in $\set_1$ and $\set_0$}.
         
        As we will see below, $r_j$ is a set of values, and $c_j$ is a timestamp.
        Intuitively, $p_j$ writes a tuple $\langle r_j,c_j\rangle$ into $R_{jk}$ to tell $p_k$ that $p_j$ is a {\wof} all the values in the set $r_j$ at ``time $c_j$''.
        
    If $p_k$ reads any $\langle r_j,c_j\rangle$ with $c_j \ge C_k$,
    $p_k$ exits the repeat-until loop.
    
    \item Then $p_k$ checks whether $v$ is in $r_j$, 
            the set of values that $p_j$ is a {\wof}:
    \begin{itemize}
    \item If $v\in r_j$,
        $p_k$ inserts $p_j$ into $\set_1$, and resets $\set_0$ to empty.
    \item If $v\not\in r_j$,
        $p_k$ just inserts $p_j$ into $\set_0$.
    \item After the above changes to $\set_1$ or $\set_0$,
        if $|\set_1|\ge n-f$, $p_k$ returns $\true$.
        Otherwise, if $|\set_0| > f$, $p_k$ returns $\false$.
        If none of these conditions hold,
        $p_k$ goes to the next round.
        \end{itemize}

    \end{itemize}
\end{itemize}

In addition to the $\Set(-)$, $\Test()$, $\sign(-)$, and $\valid(-)$ procedures,
    the algorithm also includes a $\fresh()$ procedure
    that each process executes in the background (even when it is not currently performing any operation on the implemented register):
    intuitively,
    a process executes $\fresh()$
    to assist ongoing executions of the \valid$(-)$ procedure by other processes.

In the $\fresh()$ procedure,
    each process $p_j$ maintains the set of values
    that it is witness of, and stores this set in a SWMR register $R_j$.
    
In the procedure,
    a process $p_j\in \{p_1,...,p_n\}$ first initializes its local register $prev\_c_k$ to~$0$ for every reader $p_k\in\{p_2,...,p_n\}$.
The local variable
    $prev\_c_k$ stores the last value that $p_j$ observed in $C_k$ (the SWMR register written by $p_k$);
    intuitively, $prev\_c_k$ tracks the last round $p_j$ helped $p_k$.
After this initialization, $p_j$ enters an infinite loop, 
    continuously monitoring and assisting ongoing operations.
In each round:

\begin{itemize}
    \item Process $p_j$ first retrieves the latest value of $C_k$ and stores it in $c_k$ for every reader $p_k$.
    
        \item Then $p_j$ determines the set of \emph{askers} --- readers whose $C_k$ increased compared to $prev\_c_k$.
    Intuitively, \emph{askers} are readers who have ongoing $\valid$ operations and need additional witnesses to verify some value.
    
    \item If no \emph{askers} are found, 
    $p_j$ does nothing and repeats the loop.
    
    \item Otherwise:
    \begin{compactitem}
     \item $p_j$ updates $R_j$, which stores the sets of values it is a {\wof}, as follows:

    \begin{itemize}
        \item $p_j$ reads the SWMR register $R_i$ of every process $p_i\in\{p_1,...,p_n\}$.
        \item If $R_1$ (the writer's register) contains a value $v$, or 
    at least $f+1$ processes $p_i$ have $v$ in their~$R_i$,
    then $p_j$ inserts $v$ into its own register $R_j$.   
Intuitively, 
    $p_j$ becomes a {\wof} a value $v$ if it sees that the writer $p_1$ has ``signed'' $v$
    or there are at least $f+1$ {\wsof} $v$.
    
        \item $p_j$ reads $R_j$ into its local variable $r_j$.

    \end{itemize}
    
    \item Finally, $p_j$ helps each asker $p_k$ by updating $R_{jk}$ to $\langle r_j, c_k \rangle$, and then it sets $prev\_c_k$ to~$c_k$.

     \end{compactitem}

\end{itemize}

In Appendix~\ref{a-mar}, we prove that in Algorithm~\ref{code-mar}
    all the correct processes complete their operations,
    and this implementation is Byzantine linearizable.
More precisely:

\begin{restatable}{theorem}{verifycorrect}
In a system with $n > 3f$ processes,
    where $f$ processes can be Byzantine,
    Algorithm~\ref{code-mar} is a correct implementation of a SWMR {\mar}.
\end{restatable}

\section{Authenticated Register}\label{tar-sec}

A SWMR \emph{authenticated} register has $\Set$, $\Test$, and $\valid$ operations.
Intuitively, every value written into this register is ``automatically signed with the writer's signature''.
Any reader can apply the operation $\valid$ on any value $v$
    to verify the authenticity of $v$:
    $\valid(v)$ returns $\true$ if and only if
    $v$ was indeed written ``and signed with the writer's signature''.
An {\tar} can be initialized to any value $v_0$
    ``signed with the writer's signature''.

More precisely, we define the sequential specification of a multivalued SWMR {\mar} where
    the writer can write any value from some domain $\mathcal V$, and the register is initialized to a value $v_o \in \mathcal V$, as follows:
    
\begin{definition}\label{def-tar}
A SWMR {\tar} has three operations.

\textbf{Read and write operations:}

\vspace*{1mm}

\begin{compactitem}
    \item $\Set(v)$ by the writer takes a value $v\in \mathcal{V}$ and returns $\done$.

    \item $\Test$ by any reader returns a value $v \in \mathcal{V}$ such that:
        \begin{itemize}
            \item either there is a $\Set(v)$ before it and this $\Set(v)$ is the last $\Set(-)$ before it,
            \item or $v = v_0$ (the initial value of the register) and there is no $\Set(-)$~before~it.
        \end{itemize}

\end{compactitem}

\textbf{Verify operation:}

\begin{compactitem}
     \item $\valid(v)$ by any reader takes a value $v \in  \mathcal{V}$, and returns $\true$ or $\false$ such that:
   \begin{itemize}

    \item A $\valid(v)$ returns $\true$ if and only if there is a $\Set(v)$~before~it or $v=v_0$.
\end{itemize}
 \end{compactitem}
\end{definition}

The above definition implies that an {\tar} has the following properties:

\begin{observation}\label{validity-tar}
    [\textsc{validity}] If the writer executes a $\Set(v)$,
    then every subsequent $\valid(v)$ by any reader returns $\true$.
\end{observation}

\begin{observation}\label{unforgeability-tar}
       [\textsc{unforgeability}] If a $\valid(v)$ returns $\true$ to some reader, then the writer previously executed
       a $\Set(v)$ or $v=v_0$.
\end{observation}

\begin{observation}\label{relay-tar}
       [\textsc{relay}] If a $\valid(v)$ returns $\true$ to some reader, then every subsequent $\valid(v)$ by any reader also returns $\true$.
\end{observation}
These properties also imply this:

\begin{observation}\label{read-val}
    If a $\Test$ returns $v$ to some reader,
    then every subsequent $\valid(v)$ by any reader returns $\true$.
\end{observation}

\section{Implementation of an Authenticated Register}\label{tar-imp}
\begin{algorithm}[]
	\caption{Implementation of a SWMR multivalued \tar~writable by process~$p_1$ (the \emph{writer}) and readable by every process $p_k \in \{p_2,\dots, p_n\}$ (the \emph{readers}), 
		for $n > 3f$.
	}\label{code-tar} 
	\ContinuedFloat
     \raggedright
        \footnotesize
       
\begin{multicols}{2}		   
	\emph{Every process $p_i \in\{p_1,\dots,p_n\}$ has the following:}

	\hspace{4mm} \underline{shared registers}
	\vspace{.7mm}

	\hspace{1cm}
		For each reader $p_j \in \{p_2,\dots,p_n\}$:
		
		\hspace{1.5cm} $R_{ij}$: $\textsf{SWSR}$ register readable by $p_j$; 

		\hspace{2.2cm} initially $ \langle \emptyset, 0 \rangle$
	
	\columnbreak
	 \footnotesize
        \emph{The writer $p_1$ also has the following:}

	\hspace{4mm} \underline{shared registers}
	\vspace{.7mm}
		   
	 \hspace{1,1cm}$R_1$: $\textsf{SWMR}$ register; initially $\{ \langle 0, v_0 \rangle \}$
\color{black}  

		\hspace{4mm} \underline{local variable}
	
\vspace{.7mm}
        
        \hspace{1,1cm}$\ell$: a counter; initially $0$

 \vspace{1.8mm}
    \mbox{\emph{Each reader $p_k \in\{p_2,\dots,p_n\}$ also has the following:}}

	\hspace{4mm} \underline{shared register}
	\vspace{.7mm}

 \hspace{1,1cm}	$R_k$: $\textsf{SWMR}$ register; initially $\{ v_0 \}$
 \color{black}   
 
	 \hspace{1,1cm}	$C_{k}$: $\textsf{SWMR}$ register; initially $0$

	\end{multicols}

	{
       
      \vspace{1mm}
      
	 \hrule
	\vspace{-2mm}
	\begin{algorithmic}[1]
	 \begin{multicols}{2}
    \vspace{1mm}
        \Statex \hspace{-5mm} \textit{$\triangleright$ Executed by the writer $p_1$}
        \vspace{2mm}
	\Statex \textsc{\Set($v$):} 
	
	\Indent
        
	\State \label{r-tar}$\ell \gets \ell+1$
	\State \label{r1-tar} $R_1 \gets R_1 \cup \{ \langle \ell, v \rangle \}$
	\State \Return $\Done$
	\EndIndent
  
       \vspace{1mm}

	 \columnbreak

        \Statex \hspace{-5mm} \textit{$\triangleright$ Executed by any reader $p_k \in \{p_2, \dots, p_n\}$}
        \vspace{1mm}
        
        \Statex
        \textsc{\Test():} 
	
	\Indent
	\State \label{readerr1-tar}$r \gets R_1$
	\State \label{checkform-tar}\textbf{if} $r$ is a set of tuples of the form $\langle \ell, v \rangle $ \textbf{then}
	 \Indent
    	\State\label{maxl-tar} Let $\langle \ell, v \rangle \in r $ s.t.
        \newline\hspace*{1.5cm}\mbox{$\forall \langle \ell', v' \rangle \in r: \langle \ell, v \rangle  \ge  \langle \ell', v' \rangle $}

	\State $\label{verifyp-tar}\mathit{verified} \gets \valid(v)$
	     \State \label{checksign-tar}\textbf{if} $\mathit{verified} = \true$ \textbf{then} \Return $v$
     
        	\EndIndent
            \State \label{returnv0-tar} \Return $v_0$
	\EndIndent

        \vspace{5mm}
	\noindent
	\textsc{$\valid$($v$):}
	\Indent
	\State $\set_0, \set_1 \gets \emptyset$
	\State \label{whileloop-tar} \textbf{while} true \textbf{do}
	\Indent
	\State \label{ckplus-tar} $C_k \gets C_k+1$
	\State \label{repeat-tar} \textbf{repeat}
	\Indent 
	\State \label{findone-tar} \textbf{for every} process $p_j \notin \set_1 \cup \set_0$ \textbf{do}
		\Indent
			\State \label{readri-tar}$\langle r_j ,c_j \rangle \gets R_{jk}$
		\EndIndent
	\EndIndent
	\State \label{until-tar} \textbf{until} $\exists~ p_j \notin \set_1 \cup \set_0$ s.t. $c_j \ge C_k$
	\State \label{check1-tar} \textbf{if} $v \in r_j$ \textbf{then} 
	\Indent
	\State \label{set1-tar} $\set_1 \gets \set_1 \cup \lbrace p_j \rbrace$
	\State \label{empty0-tar} $\set_0 \gets \emptyset$
	\EndIndent
	\State  \label{notv-tar}\textbf{if} $v \not \in r_j$ \textbf{then}
	\Indent
	\State \label{set0-tar} $\set_0 = \set_0 \cup \lbrace p_j \rbrace$
	\EndIndent
	\State \label{return1-tar} \textbf{if} $|\set_1| \ge n-f$ \textbf{then} \textbf{return} $\true$
	\State \label{return0-tar} \textbf{if} $|\set_0| > f$ \textbf{then} \textbf{return} $\false$

	\EndIndent
	
	\EndIndent
	\end{multicols}
	\vspace{-1mm}
	\hrule
	\vspace{4mm}
	\Statex

        \Statex \hspace{-5mm} \textit{$\triangleright$ Executed by every process $p_j\in\{p_1,\dots,p_n\}$}
        \vspace{1mm}
        
	\noindent
	\fresh():
        \Indent
        \State \label{collectck-tar}\textbf{for every} reader $p_k\in \{p_2, \dots, p_n\}$ \textbf{do} $prev\_c_k \gets 0$ 
	\State \textbf{while} true \textbf{do}
	\Indent
	\State \label{collectck-tar}\textbf{for every} reader $p_k\in \{p_2, \dots, p_n\}$ \textbf{do} $c_k \gets C_k$  
	\State \label{askers-tar}$askers \gets \{ p_k\in \{p_2, \dots, p_n\} ~|~ c_k > prev\_{c_{k}} \}$

	\State\label{replyasker-tar}\textbf{if} $askers \neq \emptyset$ \textbf{then} 
	\Indent 

	\State \label{readr1-tar}$r \gets R_1$
	\State \label{extract-r1-tar}$r_1 \gets \{v ~|~  \langle -, v \rangle  \in r \}$
        \State  \textbf{if} $j \neq 1$ \textbf{then}
        \Indent
	\State \label{sets-tar} \textbf{for} every process $p_i \in\{p_2,\dots,p_n\}$ \textbf{do} $ r_i \gets R_i$
	
	\State \label{followcondition-tar} \textbf{for} each value $v$ such that $v \in r_1$ or $|\{ r_i~|~1 \le i \le n~\text{and}~ v \in r_i\}| \ge f+1$ \textbf{do}
		\Indent
	\State \label{follow-tar} $R_j \gets R_j \cup \{v\}$ 
        \EndIndent
        \State \label{readrj-tar} $r_j \gets R_j$
	\EndIndent

	\State \label{tellasker-tar}\textbf{for every} $p_k \in askers$ \textbf{do}
	
	\Indent 
	
	\State \label{fresh1-tar} $R_{jk} \gets \langle r_j, c_k \rangle$
	\State \label{setprev-tar}$prev\_c_k \gets c_k$
	\EndIndent
	\EndIndent
	
	\EndIndent
	\EndIndent
	
	\end{algorithmic}
	}
	\vspace{1mm}
	\end{algorithm}

We now explain how to implement an {\tar} in a system with SWMR registers and $n>3f$ processes,
	$f$ of which may be Byzantine.
The implementation is similar to the one for a {\mar} given in Algorithm~\ref{code-mar}.
We first explain the differences in Section~\ref{idea-tar}, and then describe the implementation in detail in Section~\ref{algo-tar}.

\subsection{Basic Ideas}\label{idea-tar}

Since {\mar}s and {\tar}s have similar Validity, Unforgeability,
and Relay properties, the signature and signature verification mechanisms of these two registers
	are quite similar.
In particular, the $\valid(-)$ procedure is the same, and the $\fresh(-)$ procedure is almost the same.
The $\Set(-)$ and $\Test$ procedures, however, are different because any implementation 
	of an {\tar} must ensure that the writing \emph{and} the signing of each value \emph{takes effect atomically}.
This is the reason for the following changes.

Recall that with our {\mar} implementation (Algorithm~\ref{code-mar}), the writer maintains two separate registers:
	$R^*$, which contains only the last value that it wrote,
	and $R_1$, which contains the values that it wrote \emph{and} subsequently ``signed''.
So to read the latest value written, the reader can simply read $R^*$.

In the {\tar} implementation, however, \emph{all} the values are atomically ``signed'' when they are written,
	and so the writer now maintains a \emph{single} register $R_1$ that contains these values (there is no separate register $R^*$ here).
The writer must now ``timestamp'' every value that it inserts into $R_1$ to indicate which one is the latest
	one that it wrote.

When a reader reads $R_1$ it selects the value $v$ with the highest timestamp from $R_1$.
But the reader can \emph{not} simply return $v$ as the value it read.
This is because if the writer is Byzantine, it can remove $v$ from $R_1$,
	and so a reader that executes $\valid(v)$ now may get $\false$ ---
	thus violating a property of {\tar}s (which is stated in Observation~\ref{read-val}).
Thus before returning $v$ as the value read, the reader now calls the $\valid(v)$ procedure, and it returns $v$ only if this procedure returns $\true$: it knows that from now on any $\valid(v)$ will also return $\true$
(and so Observation~\ref{read-val} holds).

But what value can a reader return
	if $v$ is the latest value that it sees in $R_1$ but $\valid(v)$ returns $\false$?
By the properties of the $\valid(-)$ procedure, this situation can occur \emph{only} if the writer is Byzantine.
So in this case, the reader returns the \emph{initial value $v_0$} of the register (which is deemed to have been ``signed'' by the writer):
	in our implementation every
	$\valid(v_0)$ is guaranteed to return $\true$ (and so Observation~\ref{read-val} also holds in this case).

In summary, to execute a $\Test$ operation, a reader
reads $R_1$,
selects the value $v$
	in $R_1$ with the latest timestamp,
calls the $\valid(v)$ procedure, and
	if this $\valid(v)$ returns $\true$,
	then the reader returns $v$ as the value read,
	otherwise it returns $v_0$.\footnote{Note that the $\valid(-)$ procedure is used to implement the $\valid$ operation,
	and it is \emph{also} used ``inside'' the implementation of the $\Test$ operation.
	This ``dual-use'' is reminiscent to the dual-use of the \textsc{Scan} procedure in the \emph{Atomic Snapshot} algorithm
	given in~\cite{snapshot}: in this algorithm, the \textsc{Scan}  procedure is used to implement the \textsc{Scan} operation,
	and it is also used ``inside''  the implementation of each \textsc{Update}$_i(-)$ operation.}
We now describe the algorithm in more detail.

\subsection{The Algorithm}\label{algo-tar}

Algorithm~\ref{code-tar} gives a linearizable implementation of a SWMR {\tar} that is writable by process $p_1$ (the writer) and readable by the set $\{p_2,...,p_n\}$ of processes (the readers).
It works for systems 
    with $n >3f$ processes, $f$ of which can be Byzantine, as follows.

The writer $p_1$  maintains a shared register  $R_1$ that stores a set of tuples of the form $\langle \ell, v \rangle$, where $\ell$ is a timestamp and $v$ is a value in $\mathcal{V}$. 
    Initially, $R_1$ stores the set $\{ \langle 0, v_0 \rangle \}$, where $v_0 \in \mathcal{V}$ is the initial value of the  {\tar} that the algorithm implements.

\begin{itemize}
\item
To \emph{write} a value $v$, 
    the writer $p_1$ calls the $\Set(v)$ procedure. 
    In this procedure, $p_1$ first increments its local counter $\ell$ and 
        then inserts the tuple $\langle \ell,  v\rangle$ into (the set of tuples stored in) the shared register $R_1$.
\item
To \emph{read} a value, 
    a reader $p_k\in \{p_2,...,p_n\}$ calls the $\Test()$ procedure. 
    In this procedure:
    \begin{itemize}
    \item First $p_k$ reads the set of tuples stored in $R_1$ into $r$. 

    \item If $r$ contains a set of tuples of the form $\langle \ell, v \rangle $ then:

    \begin{itemize}
        \item $p_k$ finds the tuple $\langle \ell, v \rangle \in r$ such that  $ \langle \ell, v \rangle \ge  \langle \ell', v' \rangle $
    for all $\langle \ell', v' \rangle \in r$,
    and executes the $\valid(v)$ procedure.\footnote{Here
    	$ \langle \ell, v \rangle \ge  \langle \ell', v' \rangle $ if and only if $\ell > \ell'$ or $\ell = \ell'$ and $v \ge v'$.}
        \item If the $\valid(v)$ returns $\true$, $p_k$ returns the value $v$.
    \end{itemize}

    \item If $p_k$ 
 has not returned yet, $p_k$ returns $v_0$.

    \end{itemize}
   
\item To \emph{verify} a value $v$, 
    a reader $p_k\in \{p_2,...,p_n\}$ calls the $\valid(v)$ procedure.
    We omit to describe this procedure here since it is identical to the one given in Algorithm~\ref{code-mar}.

\end{itemize}

In addition to the $\Set(-)$, $\Test()$, and $\valid(-)$ procedures,
    the algorithm also includes a $\fresh()$ procedure
    that each process executes in the background (even when it is not currently performing any operation on the implemented register):
    intuitively,
    a process executes $\fresh()$
    to assist ongoing executions of the \valid$(-)$ procedure by other processes.

In the $\fresh()$ procedure,
    each process $p_j$ maintains the set of values
    that it is witness of, and stores this set in a SWMR register $R_j$
    (intuitively, $p_j$ is a {\wof} $v$ if it knows that the writer wrote $v$ into $R_1$).
    
In the procedure,
    a process $p_j\in \{p_1,...,p_n\}$ first initializes its local register $prev\_c_k$ to~$0$ for every reader $p_k\in\{p_2,...,p_n\}$.
The local variable
    $prev\_c_k$ stores the last value that $p_j$ observed in $C_k$ (the SWMR register written by $p_k$);
    intuitively, $prev\_c_k$ tracks the last round $p_j$ helped $p_k$.
After this initialization, $p_j$ enters an infinite loop, 
    continuously assisting ongoing {\valid} executions that need help.
In each iteration of the loop:

\begin{itemize}
    \item Process $p_j$ first retrieves the latest value of $C_k$ and stores it in $c_k$ for every reader $p_k$.
    
    \item Then $p_j$ determines the set of \emph{askers} --- readers whose $C_k$ increased compared to $prev\_c_k$.
    Intuitively, \emph{askers} are readers who have ongoing executions of the $\valid(-)$ procedure and need additional witnesses to verify some value.
    
    \item If no \emph{askers} are found, 
    $p_j$ does nothing and repeats the loop.

        \item Otherwise:
        
        \begin{compactitem}
           \item  $p_j$ inserts into $r_1$ all the values $v$ such that $\langle -, v \rangle$ is in $R_1$.

        \item If $j\ne 1$, $p_j$ also updates $R_j$ --- the sets of values it is a {\wof} --- as follows:\footnote{Note that while the register $R_1$ contains \emph{timestamped} values, for $j \neq 1$ the register $R_j$ contains just values without timestamps.} 
            \begin{itemize}
        
        \item $p_j$ reads the SWMR register $R_i$ into $r_i$ for every process $p_i \ne p_1$.

        \item $p_j$ updates $R_j$ by inserting into $R_j$ every value $v$ such that $v$ is in $r_1$
        or $v$ is in at least $f+1$
        distinct $r_i$'s.
    Intuitively, $p_j$ becomes a {\wsof} $v$ if it sees that the writer $p_1$ wrote $v$
    or it sees that there are at least $f+1$ {\wsof} $v$.

    \item $p_j$ reads $R_j$ into its local variable $r_j$.
 \end{itemize}
    
    \item Finally, $p_j$ helps each asker $p_k$ by updating $R_{jk}$ to $\langle r_j, c_k \rangle$, and then it sets $prev\_c_k$ to~$c_k$.
 \end{compactitem}
\end{itemize}

In Appendix~\ref{a-tar}, we prove that in Algorithm~\ref{code-tar}
    all the correct processes complete their operations,
    and this implementation is Byzantine linearizable.
More precisely:

\begin{restatable}{theorem}{tarcorrect}
In a system with $n > 3f$ processes,
    where $f$ processes can be Byzantine,
    Algorithm~\ref{code-tar} is a correct implementation of a SWMR {\tar}.
\end{restatable}

\section{Sticky Register}\label{srsec}

Intuitively a SWMR {\sr} is a register such that once a value is written into it,
    the register never changes its value again.
More precisely, the sequential specification of a multivalued SWMR \emph{{\sr}} where
    the writer can write any value from some domain $\mathcal V$, and the register is initialized to a special value $\bot \notin \mathcal V$, is as follows:

\begin{definition}\label{def-sar}
A  SWMR \sr~has two operations.

\begin{itemize}
    \item $\Set(v)$ by the writer takes a value $v\in \mathcal V$ and returns $\done$.
    \item $\Test$ by any reader returns a value $v \in \mathcal V \cup \{\bot \}$ such that:
    \begin{itemize}
        \item If $v \in \mathcal V$ then $\Set(v)$ is the first $\Set(-)$ operation and it is before the $\Test$. 
        \item If $v =\bot$ then no $\Set(-)$ is before the $\Test$.
    \end{itemize}

\end{itemize}
\end{definition}

The above definition implies that a {\sr} has the following properties:

\begin{observation}\label{validity-sar}
    [\textsc{validity}] If $\Set(v)$ is the first $\Set(-)$ operation,
    then every subsequent $\Test$ by any reader returns $v$.
\end{observation}

\begin{observation}\label{unforgeability-sar}
       [\textsc{unforgeability}] If a $\Test$ returns $v \neq \bot$ to some reader,
       then $\Set(v)$ is the \emph{first} $\Set(-)$ operation
       and it is before the $\Test$.
\end{observation}

\begin{observation}\label{uniqueness-sar}
       [\textsc{uniqueness}] If a $\Test$ returns $v \neq \bot$ to some reader, then every subsequent $\Test$ by any reader also returns $v$.
\end{observation}

Note that when a process $p$ reads a value $v \neq \bot$ from a sticky register $R$,
    it can easily relay it to any other process $q$:
    $q$ can also verify that the owner of $R$ indeed wrote $v$ just by reading $R$.
Since $R$ is sticky, the value $v$ that $p$ saw is guaranteed to remain in $R$,
    and so $q$ will also see it there.
So in some precise sense,
    the Uniqueness property of sticky registers (Observation~\ref{uniqueness-sar})
    also provides the ``relay'' property of signed values.

\section{Implementation of a {\src}}\label{srsec-imp}

We now explain how to implement a {\sr} in a system with SWMR registers and $n>3f$ processes, $f$ of which may be Byzantine.

\medskip\noindent
\textbf{Notation.} For convenience, henceforth $v$ denotes a \emph{non-$\bot$ value}.
When we refer to the initial~value~$\bot$ of the sticky register, we do so explicitly. 
    
\subsection{Basic Ideas}
Since the Validity and Unforgeability properties of a {\sr} (Observations~\ref{validity-sar}-\ref{unforgeability-sar})
    are similar to their counterparts in
    a {\mar} (Observations~\ref{validity}-\ref{unforgeability}),
   and also in
    an {\tar} (Observations~\ref{validity-tar}-\ref{unforgeability-tar}),
    the implementations of these three types of registers have similar parts that ensure
    these properties.
    
But in contrast to a {\mar} or an {\tar}, which allows a process to write \emph{several} values, a SWMR {\sr}
    restricts the writer
    to write a \emph{single} value that it cannot change or erase later (even if it is Byzantine!).
This is captured by the Uniqueness property of sticky registers ({Observation~\ref{uniqueness-sar}).
So to implement a {\sr} we need to integrate the mechanism that ensures Validity and Unforgeability
    that we saw in Section~\ref{marc-imp} with a mechanism that ensures Uniqueness,
    and do so in a way that preserves all three properties.
    
To ensure Uniqueness, we proceed as follows.
    To write a value $v$, the writer first writes $v$ in a SWMR register $E_1$ --- the ``echo'' register of $p_1$.
        It then \emph{waits} to see that $v$ is ``witnessed'' by at least $n-f$ processes (we will see how processes become witnesses below).
    When this occurs, it returns done.

    The \emph{first time} a process $p_j$ sees a value $v$ in $E_1$,
    it ``echoes'' $v$ by writing $v$ in its own SWMR ``echo'' register $E_j$.
    Note that (a)~$p_j$ does this only for the first non-$\bot$ value that it see in $E_1$,
    and (b)~$p_j$ is not yet willing to become a {\wof} $v$.
    Process $p_j$ becomes a {\wof} $v$ if it sees that at least $n-f$
    processes have $v$ in their own ``echo'' registers.
    When this happens $p_j$ writes $v$ into a SWMR $R_j$ --- the ``witness'' register of $p_j$.
    Process $p_j$ also becomes a {\wof} $v$ if it sees that at
    least $n-f$ processes have $v$ in their own ``witness'' registers.
    
    Note that this policy for becoming a witness is more strict than the one used
    in the implementation of a {\mar} in Algorithm~\ref{code-mar}: in that algorithm, a process $p_i$
    was willing to become a {\wof} a value $v$ as soon as it saw $v$ in the register $R_1$ of the writer.
    The stricter policy used for the implementation of a {\sr}
    prevents correct processes from becoming witnesses for different values,
    and this also prevents readers from reading different values.
    
It is worth noting that
    in the implementation of a {\mar} (Algorithm~\ref{code-mar}),
    to write or sign a value $v$ the writer does \emph{not} wait for other processes
    to become {\wsof}~$v$: it returns done (almost) immediately
    after writing a single SWMR register.
So a reader may wonder why in our implementation of a {\sr}, to write a value $v$, the writer must wait for
$n-f$ {\wsof} $v$ before returning done.
It turns out that \emph{without this wait}, a process may invoke a $\Test$ after a $\Set(v)$ completes
    and get back $\bot$ rather~than~$v$.
Intuitively, this is because the stricter policy for becoming a {\wof} $v$
	(in our sticky register implementation) may delay
    the ``acceptance'' of $v$ as a legitimate return value.
    
 We now describe the algorithm in more detail.

\begin{algorithm}[]
	\caption{Implementation of a \sr~writable by process~$p_1$ (the \emph{writer}) and readable by every process $p_k \in \{p_2,\dots, p_n\}$ (the \emph{readers}), 
		for $n > 3f$.
	}\label{code-sar} 
	\ContinuedFloat
	
	\vspace{-3mm}
	
	{
	\footnotesize
	\begin{multicols}{2}
        \emph{Every process $p_i \in\{p_1,\dots,p_n\}$ has the following:}
        
	\hspace{4mm} \underline{shared registers}
	\vspace{.7mm}
        
        \hspace{1cm}
            $E_{i}$: $\textsf{SWMR}$ register; initially $\bot$
	
	\hspace{1cm}
		$R_{i}$: $\textsf{SWMR}$ register; initially $\bot$

	\hspace{1cm}
		For each reader $p_j \in \{p_2,\dots,p_n\}$:
		
		\hspace{1.5cm} $R_{ij}$: $\textsf{SWSR}$ register readable by $p_j$; 

		\hspace{2.2cm} initially $ \langle \bot, 0 \rangle$
	\columnbreak

		 \mbox{\emph{Each reader $p_k \in\{p_2,\dots,p_n\}$ also has the following:}}

	\hspace{4mm} \underline{shared register}
	\vspace{.7mm}

	\hspace{1cm} $C_{k}$: $\textsf{SWMR}$ register; initially $0$
\end{multicols}
	
	\hrule \vspace{1mm}

	\begin{algorithmic}[1]
	  \begin{multicols}{2}
        \Statex \hspace{-5mm} \textit{$\triangleright$ Executed by the writer $p_1$}
        \vspace{1mm}
	\Statex 
	\textsc{\Set($v$):}
	
	\Indent
	\State \label{stick-sar}\textbf{if} $E_1 \ne \bot$  \textbf{then} \Return $\Done$

	\State \label{setter1-sar}$E_1 \gets v$ 
	\State \label{repeat-sar}\textbf{repeat}
	\Indent
	\State \label{writerreadsr-sar}\textbf{for} every $p_i \in\{p_1,\dots,p_n\}$ \textbf{do} $ r_i \gets R_i$

	\EndIndent
	\State \label{vconfirm-sar}\textbf{until} $|\{ r_i ~|~ r_i = v\}| \ge n-f$

	\State \label{wr-sar}\Return $\Done$
	\EndIndent

	\columnbreak
        \vspace{1mm}

        \Statex \hspace{-5mm} \textit{$\triangleright$ Executed by a reader $p_k \in \{p_2, \dots, p_n\}$}
        \vspace{1mm}
	\Statex
	\textsc{$\Test$():}
	\Indent
	\State $\setb, \setv \gets \emptyset$
	
	\State \label{whileloop-sar} \textbf{while} true \textbf{do} 
	\Indent
        \State \label{ckplus-sar} $C_k \gets C_k+1$
        \State \label{BS} $S \gets \{ p_j ~|~ p_j \not\in \setb \mbox{~and~} \langle-,p_j\rangle\not\in\setv \}$
        \State \label{repeat-sar} \textbf{repeat}
	\Indent 
 	\State \label{findone-sar} \textbf{for every} process $p_j \in S$ \textbf{do} 
	\Indent
	\State \label{readri-sar}$\langle u_j,c_j\rangle \gets R_{jk}$  
        \EndIndent
	\EndIndent
	\State \label{until-sar} \textbf{until} $\exists~ p_j \in S$ such that $c_j \ge C_k$

	\State \label{check1-sar} \textbf{if} $u_j \ne \bot$ \textbf{then} 
	\Indent
        \State \label{setv-sar} $\setv \gets \setv \cup \lbrace \langle u_j,p_j\rangle \rbrace$
	\State \label{empty0-sar} $\setb \gets \emptyset$
	\EndIndent
	\State \label{checkb-sar} \textbf{if} $u_j = \bot$ \textbf{then} 
	\Indent
	\State \label{set0-sar} $\setb = \setb \cup \lbrace p_j \rbrace$
        \EndIndent
        
	\State \label{return1-sar}  \mbox{\textbf{if} $\exists v$ such that} 
    \mbox{$|\{p_j|\langle v,p_j\rangle\in \setv\}|\ge n-f$}   
    \Indent
    \State  \label{return1-sar-2} \label{return1-sar-1} \textbf{then} \textbf{return} $v$
    \EndIndent
	\State \label{return0-sar} \textbf{if} $|\setb| > f$ \textbf{then} \textbf{return} $\bot$

	\EndIndent

	\EndIndent
	
        \vspace{1mm}

	\end{multicols}
	\vspace{-1mm}
	\hrule
	\vspace{2mm}

	\newcommand{\help}{help}
	\Statex \hspace{-5mm} \textit{$\triangleright$ Executed by every process $p_j\in\{p_1,\dots,p_n\}$}
        \vspace{1mm}
        \Statex
	\textsc{\fresh():} 
	
        \Indent
        \State \label{collectck-init-sar}\textbf{for every} reader $p_k\in \{p_2, \dots, p_n\}$ \textbf{do} $prev\_c_k \gets 0$ 
	\State \textbf{while} true \textbf{do}
	\Indent
        
    	\State\label{helpecho-sar}\textbf{if} $E_j = \bot$ then
    	\Indent
    	\State \label{reade1-sar}$e_j \gets E_1$
            \State \label{echo-sar}$E_j \gets e_j$
    	\EndIndent
        
    	\State \textbf{if} $R_j = \bot$ then
    	\Indent
    	\State \label{readecho-sar}\textbf{for} every process $p_i \in\{p_1,\dots,p_n\}$ \textbf{do} $ e_i \gets E_i$
    	\State \label{echosupport-sar}$\exists v \neq \bot$ such that $|\{ e_i ~|~ e_i =  v \}| \ge n-f$ \textbf{do} $ R_j \gets v$
    	\EndIndent

    	\State \label{collectck-sar}\textbf{for every} reader $p_k\in \{p_2, \dots, p_n\}$ \textbf{do} $c_k \gets C_k$  
    	\State \label{askers-sar}$askers \gets \{ p_k\in \{p_2, \dots, p_n\} ~|~ c_k > prev\_c_k \}$
    	\State\label{replyasker-sar}\textbf{if} $askers \neq \emptyset$ \textbf{then} 
	    \Indent 
        	\State \label{ifnotone-sar} \textbf{if} $R_j = \bot$ \textbf{then} 
        	\Indent
        	
        	\State \label{sets-sar} \textbf{for} every process $p_i \in\{p_1,\dots,p_n\}$ \textbf{do} $ r_i \gets R_i$ 
        	\State \label{followcondition-sar} \textbf{if} $\exists v \neq \bot$ such that $|\{ r_i~|~ r_i =v \}| \ge f+1$ 
            \textbf{do} $R_j \gets v$
            \EndIndent
            \State \label{rj-sar}$r_j\gets R_j$

        	\State \label{tellasker-sar}\textbf{for every} $p_k \in askers$ \textbf{do} 
    	\Indent 
    	\State \label{fresh1-sar} $R_{jk} \gets \langle r_j, c_k \rangle$
    	\State \label{setprev-sar}$prev\_c_k \gets c_k$
    	\EndIndent
        
	\EndIndent
	\EndIndent
        \EndIndent
	\end{algorithmic}
	}
	\vspace{2mm}
	\end{algorithm}

\subsection{The Algorithm}

Algorithm~\ref{code-sar} gives a linearizable implementation of a SWMR {\sr} that is writable by process $p_1$ (the writer) and readable by the set $\{p_2,...,p_n\}$ of processes (the readers).
It works for systems 
    with $n >3f$ processes, $f$ of which can be Byzantine, as follows.
    \begin{compactitem}

\item To \emph{write} a value $v$, 
    the writer $p_1$ calls the $\Set(v)$ procedure. In this procedure,
    \begin{compactitem}
        \item  $p_1$ first checks whether its SWMR register $E_1 =\bot$.
        \item  If $E_1 \neq \bot$
                (this indicates that $p_1$ previously wrote some value),
                $p_1$ returns done.
        \item  Otherwise, $p_1$ writes $v$ into its SWMR register $E_1$.

        \item Then $p_1$ repeatedly checks whether
                at least $n-f$ processes have $v$ in their SWMR registers~$R_i$. 
            When the condition holds, $p_1$ returns done.
            Intuitively, 
                $p_1$ returns done only after it sees that
                at least $n-f$ processes are {\wsof} $v$.
    \end{compactitem}

\item To \emph{read} a value, 
    a reader $p_k\in \{p_2,...,p_n\}$ calls the $\Test()$ procedure. 
    In this procedure, $p_k$ first initializes $\setv$ and $\setb$ to empty.
        Intuitively, $\setv$ contains tuples $\langle v,p_j\rangle$ such that
        $p_j$ informed $p_k$ that it is a {\wof}
        the value $v$ (during this execution of $\Test$),
        while $\setb$ contains processes that
        (a) informed $p_k$ that they are not {\wsof} any value
        and (b) they did so after the last round in which some process told $p_k$ that it is a {\wof} some value.
    In the following, when we say that a process $p$ is in $\setv$,
        we mean $\langle -,p\rangle \in \setv$.

      After this initialization, $p_k$ enters a while loop. 
        Each iteration of this loop is a ``round''.
        
    In each round:
    \begin{compactitem}
     \item  First $p_k$ increments a shared SWMR register $C_k$.
     \item  Then $p_k$ constructs the set $S$ of \emph{of all processes $p_j$ that are not in $\setb$ and $\setv$}.
      \item Then $p_k$ enters a repeat-until loop.
        In this loop, $p_k$ reads the tuple $\langle u_j,c_j\rangle$
        from the SWSR register $R_{jk}$
            of every process $p_j$ in $S$.
    
        As we will see below, $u_j$ is a value, and $c_j$ is a timestamp.
        Intuitively, $p_j$ writes a tuple $\langle u_j,c_j\rangle$ into $R_{jk}$ to tell $p_k$ that it is a {\wof} the value $u_j$ at ``time $c_j$''.
        
          If $p_k$ reads any $\langle u_j,c_j\rangle$ with $c_j \ge C_k$ from any process $p_j \in S$,
    $p_k$ exits the repeat-until loop.
        
        \item Then $p_k$ checks whether $u_j=\bot$:
        \begin{compactitem}
            \item If $u_j\ne\bot$,        
                $p_k$ inserts $\langle u_j,p_j\rangle$ into $\setv$, and resets $\setb$ to empty.
            \item If $u_j=\bot$, $p_k$ just inserts $p_j$ into $\setb$.
            \item After the above changes of $\setb$ or $\setv$,
                     if there is a value $v$ that is witnessed by at least $n-f$ processes in $\setv$, $p_k$ returns $v$.
                 Otherwise, if $|\setb|>f$, $p_k$ returns $\bot$.
                If none of these conditions hold,
                    $p_k$ goes to the next round.
        \end{compactitem}

     \end{compactitem}
\end{compactitem} 

In addition to the $\Set(-)$ and $\Test()$ procedures,
    the algorithm also includes a $\fresh()$ procedure
    that each process executes in the background (even when it is not currently performing any operation on the implemented register):
    intuitively,
    a process executes $\fresh()$
    to assist ongoing executions of the \valid$(-)$ procedure by other processes.

In the $\fresh()$ procedure,
    each process $p_j$ stores in the register $R_j$ the (unique) value $v$ that $p_j$ is a {\wof};
    initially, $R_j = \bot$ and it remains $\bot$ as long as $p_j$ is not a {\wof} any value.
    
In the procedure,
    a process $p_j\in \{p_1,...,p_n\}$ first initializes its local register $prev\_c_k$ to~$0$ for every reader $p_k\in\{p_2,...,p_n\}$.
The local variable
    $prev\_c_k$ stores the last value that $p_j$ observed in $C_k$ (the SWMR register written by $p_k$);
    intuitively, $prev\_c_k$ tracks the last round $p_j$ helped $p_k$.
After this initialization, $p_j$ enters an infinite loop, 
    to continuously monitor and help ongoing operations.
In each iteration:
\begin{compactitem}

    \item Process $p_j$ first check if $E_j=\bot$.
    If so, $p_j$ reads the latest value of $E_1$ 
        and $p_j$ writes this value into its register $E_j$.
     Intuitively, 
            $p_j$ ``echos'' a value $v$ if it reads $v$ directly from the writer.
    \item Then $p_j$ check if $R_j=\bot$.
    If so,
        $p_j$ reads the SWMR register $E_i$ of every process $p_i\in\{p_1,...,p_n\}$. 
        If for some value $v$, at least $n-f$ processes $p_i$ have $v$ in their~$E_i$,
            then $p_j$ writes the value $v$ to its own register $R_j$.
        Intuitively, 
            $p_j$ becomes a {\wof} a value $v$ 
                if it sees that at least $n-f$ processes ``echoed'' $v$.
   
           \item Then $p_j$ determines the set of \emph{askers} --- readers whose $C_k$ increased compared to $prev\_c_k$.
    Intuitively, \emph{askers} are readers who have ongoing $\Test$ operations and need additional witnesses to verify some value.

\item If no \emph{askers} are found, 
    $p_j$ does nothing and repeats the loop.
\item Otherwise:
       \begin{compactitem}
       	\item 
    If $R_j$ still contains $\bot$, $p_j$ updates $R_j$ as follows:
	       \begin{compactitem}

            \item $p_j$ reads the SWMR register $R_i$ of every process $p_i\in\{p_1,...,p_n\}$.
            \item If for some value $v$, at least $f+1$ processes $p_i$ have $v$ in their~$R_i$,
                     then $p_j$ writes $v$ to its own register $R_j$.   
                Intuitively, 
                    $p_j$ becomes a {\wof} a value $v$ if it sees that there are at least $f+1$ {\wsof} $v$.
             \end{compactitem}
             \item    $p_j$ reads its updated $R_j$ into its local variable $r_j$.
                 
        \item Finally, $p_j$ helps each asker $p_k$ by updating $R_{jk}$ to $\langle r_j, c_k \rangle$, and then it sets $prev\_c_k$ to~$c_k$.

       \end{compactitem} 
\end{compactitem}

\smallskip
In Appendix~\ref{a-sar}, we prove that in Algorithm~\ref{code-sar}
    all the correct processes complete their operations,
    and this implementation is Byzantine linearizable.
More precisely:

\begin{restatable}{theorem}{stickycorrect}
In a system with $n > 3f$ processes,
    where $f$ processes can be Byzantine,
  Algorithm~\ref{code-sar} is a correct implementation of a SWMR {\sr}.
\end{restatable}

\section{Optimality Result}\label{Impossibility-Result}

We now show that our implementations of a {\mar}, {\tar}, and {\sr}, are optimal in the number of Byzantine processes that they can tolerate. 
More precisely, we show that there are no correct implementations of these objects from SWMR registers
    that tolerate $f$ Byzantine processes if $3 \le n \le 3f$. 
To do so we define a much simpler object, called \emph{{\br}}, and show that
    (1)~it cannot be implemented from SWMR registers when $3 \le n \le 3f$, 
    and
    (2) it can be implemented from a {\mar}, from an {\tar}, or from a {\sr}.

Intuitively, \emph{{\br}} is a register initialized to $0$ that can be set to 1 by a single process (the \emph{setter}), and tested by other processes (the \emph{testers}).
More precisely:

\begin{definition}\label{def-ts}
A \br~object has two operations:
\begin{itemize}
    \item $\rSet$ by the setter.
    \item $\rTest$ by any tester returns 0 or 1. 
    $\rTest$ returns 1 if and only if a $\rSet$ occurs before the $\rTest$. 
\end{itemize}
\end{definition}

Note that this atomic object has the following properties:

\begin{observation}\label{ob-ts} 
~
\begin{enumerate}
\item\label{zero}
       If a $\rSet$ occurs before a $\rTest$, then the $\rTest$ returns~1.
\item\label{uno}
          If a $\rTest$ returns $1$, then a $\rSet$ occurs before the $\rTest$.

\item\label{due}
       If a $\rTest$ returns $1$, and it precedes a $\rTest'$, then $\rTest'$ also returns $1$.
\end{enumerate}
\end{observation}

Henceforth we consider the \emph{one-shot} version of the {\br} object: the setter can invoke at most one $\rSet$ operation,
    and each tester can invoke at most one $\rTest$ operation.

To prove our result we start with the following lemma:

\begin{lemma}\label{TheSaver}
Consider any correct implementation of a {\br} object.
Every history of this implementation satisfies the following three properties:
\begin{enumerate}
\item \label{Zero} 
    If the setter is \emph{correct} and a $\rSet$ operation precedes a $\rTest$ operation by a \emph{correct} tester,
        then the $\rTest$ returns $1$.
\item \label{Uno} 
    If a $\rTest$ operation by a \emph{correct} tester returns $1$, and the setter is \emph{correct}, then the setter invokes a $\rSet$ operation
    before the tester returns $1$.
\item \label{Due} 
    If a $\rTest$ operation by a \emph{correct} tester returns $1$, and it precedes a $\rTest'$ operation by a \emph{correct} tester, then $\rTest'$ also returns~$1$.
\end{enumerate}
\end{lemma}
 
\begin{proof}[Proof Sketch]
Let $\AW$ be a correct
    implementation of {\br},
    and $H$ be any history of $\AW$.
Since $\AW$ is correct:

(A) $H$ is Byzantine linearizable with respect to {\br}.

(B) All processes that are correct in $H$ complete all their operations.

Let $\hct$ be the set of correct processes in $H$.
By (A) and Definition~\ref{def-hbl},
        there is a history $H'$ such that 
        $H'|\hct$ $= H|\hct$ and $H'$ is linearizable with respect to a {\br} (Definition~\ref{def-ts}).
Let $H'$ be such a history and let $L$ be a linearization of $H'$ such that:
    (a) $L$ respects the precedence relation between the operations of $H'$, and
    (b) $L$ conforms to the sequential specification of {\br}.
We now prove that the history $H$ of implementation $\AW$ satisfies each property stated in the lemma. 

\begin{enumerate}

     \item Suppose the setter is \emph{correct} and a $\rSet$ operation precedes a $\rTest$ operation by a \emph{correct} tester in $H$.
     Then by (B),
        this $\rTest$ operation completes, and so it returns 0 or 1 in $H$.
    Since $H'|\hct = H|\hct$,
        $\rSet$ and $\rTest$ are also in $H'$ and $\rSet$ precedes $\rTest$ in $H'$.
    Since~$L$ is a linearization of $H'$,
        $\rSet$ and $\rTest$ are also in $L$.
    Since $L$ respects the precedence relation between the operations of $H'$,
        $\rSet$ precedes $\rTest$ in $L$.
    Since $L$ conforms to the sequential specification of \mbox{{\br}},
        by Observation~\ref{ob-ts}(\ref{zero}),
        $\rTest$ returns 1 in $L$.
    So this $\rTest$ also returns 1 in $H'$ and in $H$.

    \item Suppose a $\rTest$ operation by a \emph{correct} tester returns $1$, and the setter is \emph{correct} in $H$.
   Since $H'|\hct = H|\hct$,
        $\rTest$ is also in $H'$.
    Since $L$ is a linearization of $H'$ and $\rTest$ is in $H'$,
        $\rTest$ is in $L$.
    Since $\rTest$ returns 1,
        and $L$ conforms to the sequential specification of {\br},
        by Observation~\ref{ob-ts}(\ref{uno}),
        a $\rSet$ precedes this $\rTest$ in $L$.
    Since $L$ is a linearization of $H'$,
        this $\rSet$ is in $H'$,
        and so both $\rSet$ and $\rTest$ are in $H'$.
    Since
        (i) this $\rSet$ precedes this $\rTest$ in $L$, and
        (ii)~$L$ respects the precedence relation between the operations of $H'$,
        the invocation of $\rSet$ is before the $\rTest$ returns~$1$ in $H'$.
Since both the setter and the tester are correct,
    and $H'|\hct = H|\hct$,
    the invocation of $\rSet$ is before the $\rTest$ returns $1$ in $H$.

    \item Suppose a $\rTest$ operation by a \emph{correct} tester returns $1$, and it precedes a $\rTest'$ operation by a \emph{correct} tester in $H$.
    Then by (B),
        this $\rTest'$ operation completes, and so it returns 0 or 1 in $H$.
    Since $H'|\hct = H|\hct$,
        $\rTest$ and $\rTest'$ are also in $H'$, and $\rTest$ precedes $\rTest'$ in $H'$.
    Since~$L$ is a linearization of $H'$,
        $\rTest$ and $\rTest'$ are also in $L$.
    Since $L$ respects the precedence relation between the operations of~$H'$,
        $\rTest$ precedes $\rTest'$ in $L$.
    Since $\rTest$ returns 1,
        and $L$ conforms to the sequential specification of {\br},
        by Observation~\ref{ob-ts}(\ref{due}), $\rTest'$ returns 1 in $L$.
    So this $\rTest'$ also returns 1 in $H'$ and in $H$.
\end{enumerate}
\end{proof}

\begin{theorem}\label{Theo-Impossibility-Result}
In a system with $3 \le n\le 3f$ processes, where $f$ processes can be Byzantine,
    there is no correct implementation
    of a {\br} object from SWMR registers.
\end{theorem} 	

\begin{proof}
Consider a system with $3 \le n\le 3f$ processes, where $f$ processes can be Byzantine.
Suppose, for contradiction, 
	there is a correct implementation $\AW$
	of a {\br} object $R$ from SWMR registers in this system.
This implies that for every history $H$ of $\AW$ in this system:

\begin{compactenum}[(1)]
    \item\label{BL} $H$ is Byzantine linearizable with respect to {\br}.
    \item\label{TER} All processes that are correct in $H$ complete all their operations.
\end{compactenum}

Let $s$ be the setter of $R$, and $p_a$ and $p_b$ be two testers of $R$.
We partition the set of processes as follows:
    \{$s$\},
    \{$p_a$\},
    \{$p_b$\},
    and three disjoint subsets $Q_1$, $Q_2$, $Q_3$ of size at most $f-1$ each.

	\begin{figure}[t]
		\vspace{-6mm} %
		\minipage{0.48\textwidth}
			\centering 
			\includegraphics[width=0.9\textwidth]{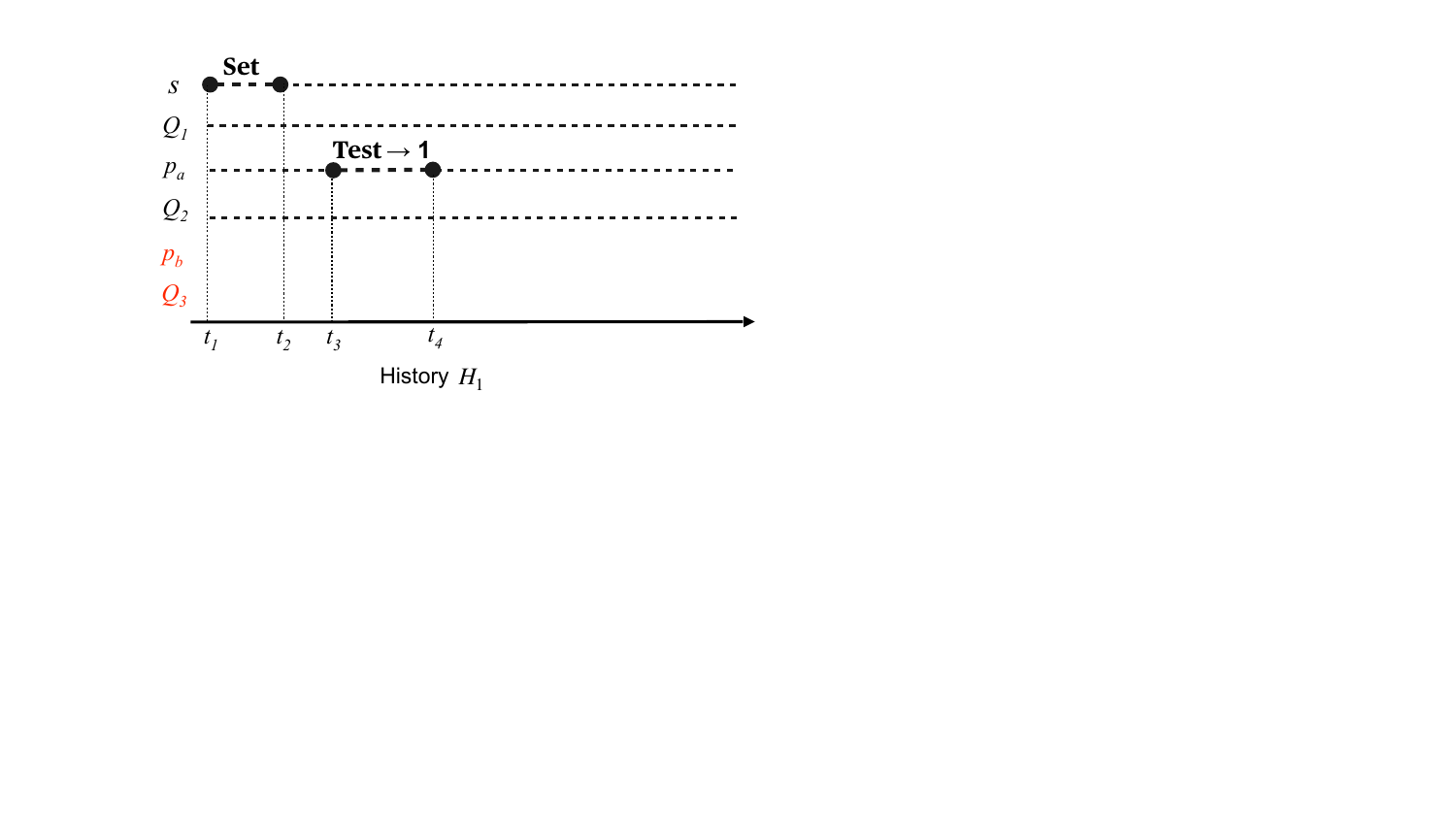}
			\label{E1}
		\endminipage
		\minipage{0.48\textwidth}
			\centering 
			\includegraphics[width=0.9\textwidth]{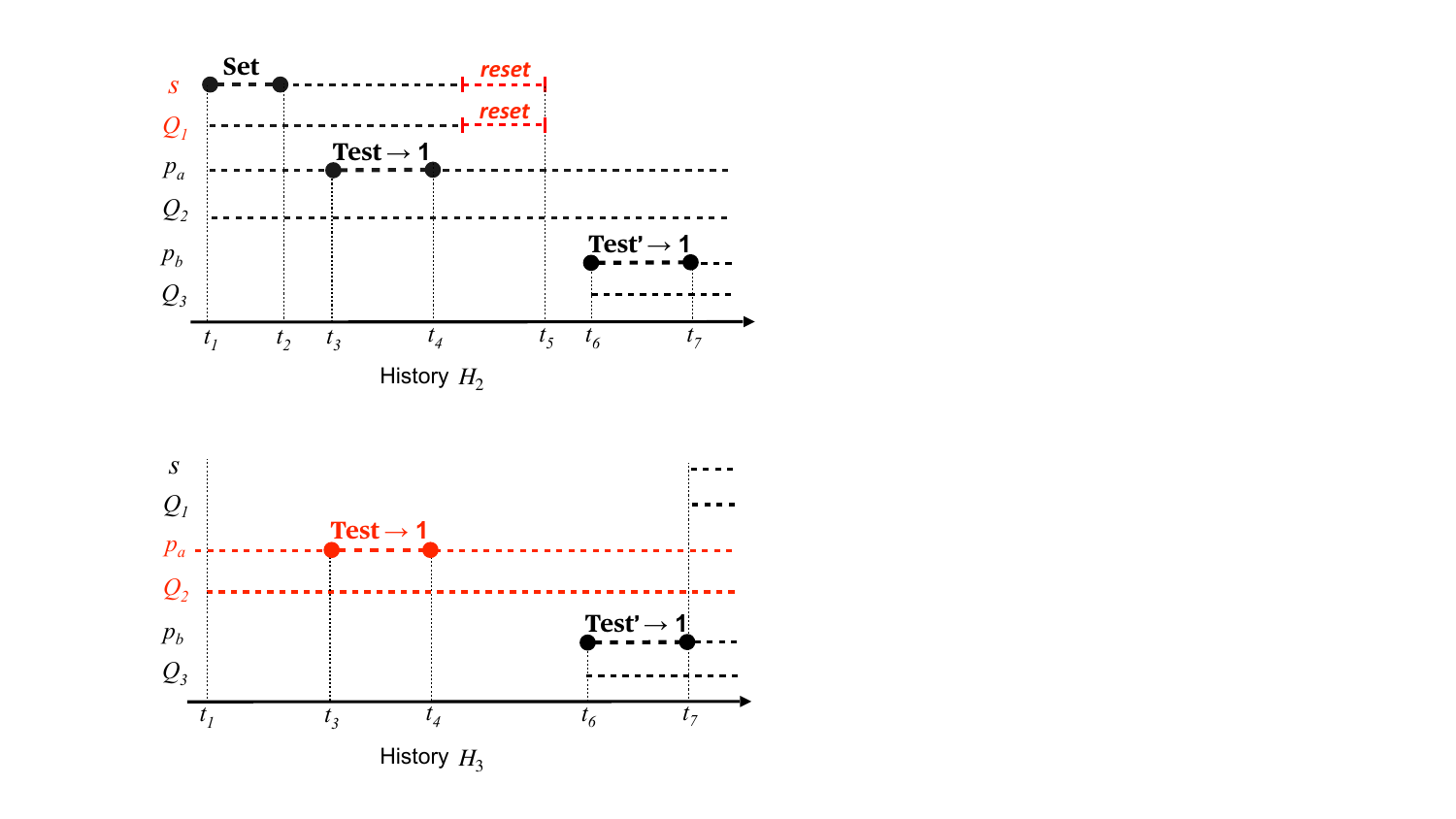}

			\label{E2}
		\endminipage 
        
            \minipage{0.48\textwidth}
			\centering 
			\includegraphics[width=0.9\textwidth]{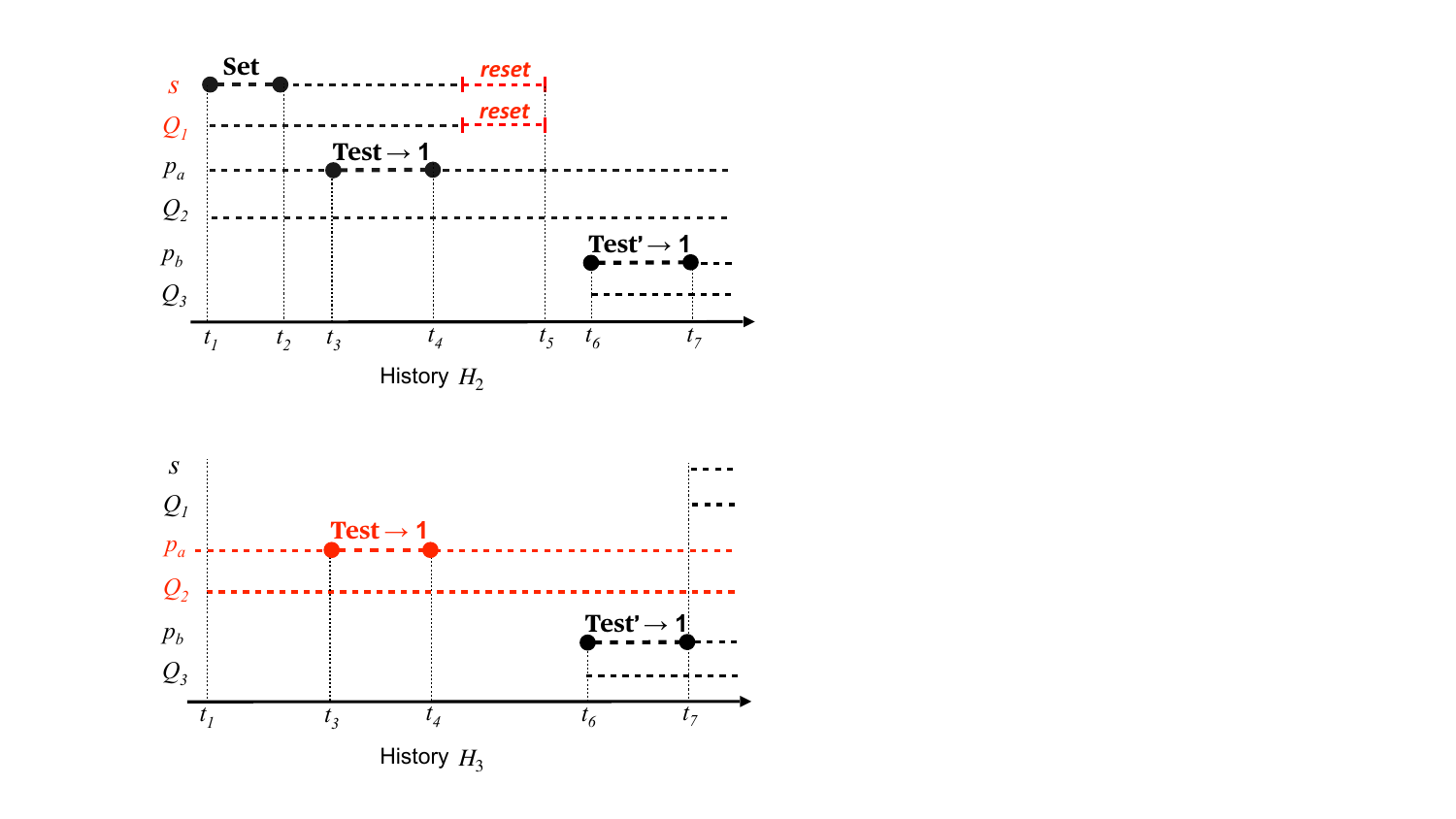}
            
			\label{E3}
		\endminipage
		\captionsetup{justification=centering}
            \captionof{figure}{Histories for the proof of Theorem~\ref{Theo-Impossibility-Result}}
		\end{figure}

We now construct a sequence of histories of the implementation $\AW$ of $R$, and prove that the last one is \emph{not} a history of a correct implementation of {\br} (see Figure~\ref{E3}).

In these figures, correct processes are in black font, 
    while Byzantine processes are colored in \textcolor{red}{red font}. 
The steps that a process takes 
    are indicated by dash lines.
An interval for a process is ``blank'' 
    if that process does not take any steps during that interval.

Let $H_1$ be the following history of the implementation $\AW$ of $R$:

\begin{itemize}
\item All the processes in $\{s,p_a\} \cup Q_1 \cup Q_2$ are correct.
\item The processes in $\{p_b\} \cup Q_3$ are Byzantine and do not take any steps (note that $|\{p_b\} \cup Q_3| \le f$).
\item The setter $s$ invokes a $\rSet$ operation at time $t_1$.

Since there are at most $f$ Byzantine processes,
    by~(\ref{TER}),
	the $\rSet$ operation by $s$ completes, say at time $t_2$.

\item The tester $p_a$ invokes a $\rTest$ operation at some time $t_3 > t_2$.

Since there are at most $f$ Byzantine processes,
    by~(\ref{TER}),
	the $\rTest$ operation completes, say at time $t_4$.
Moreover, since the setter $s$ is correct and its $\rSet$ operation precedes the $\rTest$
    operation by the correct
    tester $p_a$,
        by (\ref{BL}) and Lemma~\ref{TheSaver}(\ref{Zero}),
	this $\rTest$ must~return~$1$.

\end{itemize}

Let $H_2$ be the following history of the implementation $\AW$ of $R$:

\begin{itemize}

\item All the processes in $\{p_a,p_b\} \cup Q_2 \cup Q_3$ are correct.
\item All the processes in $\{s\} \cup Q_1$ are Byzantine (note that $|\{s\} \cup Q_1| \le f$).

\item Up to and including time $t_4$, all the processes behave exactly the same as in $H_1$.
This is possible because
    up to time $t_4$, the processes in $\{p_a\} \cup Q_2$ cannot distinguish between $H_1$ and~$H_2$.
    So $p_a$ performs a $\rTest$ that returns 1 at time $t_4$,
    while the processes in $\{p_b\} \cup Q_3$ are still ``asleep'' and take no steps.

\item After time $t_4$,
    while the processes in $\{p_b\} \cup Q_3$ continue to take no steps,
	the Byzantine processes in $\{s\} \cup Q_1$ reset all their SWMR registers to their initial states.
	Let $t_5$ denote the time when these processes complete this reset.
	The processes in $\{s\} \cup Q_1$ take no steps after~$t_5$.
So after time $t_5$ all the SWMR registers of the processes in $\{s\} \cup Q_1$ remain
    in their initial state ``as if'' these processes never took any step.
    
\item At time $t_6 > t_5$, 
	the processes in $\{p_b\} \cup Q_3$ ``wake up'' and
	process $p_b$ invokes a $\rTest'$ operation.
Since there are at most $f$ Byzantine processes,
    by~(\ref{TER}),
	the $\rTest'$ operation completes, say at time $t_7$.
Moreover, since the $\rTest$ by correct tester $p_a$ returns~$1$,
    and it precedes the $\rTest'$ by correct tester $p_b$,
    then
	by (\ref{BL}) and Lemma~\ref{TheSaver}(\ref{Due}),
	this $\rTest'$ must also return~$1$.
\end{itemize}

Finally we construct the history $H_3$ of the implementation \mbox{$\AW$ of~$R$:}

\begin{itemize}
\item All the processes in $\{p_b\} \cup Q_3$ are correct.

\item All the processes in $\{s\} \cup Q_1$ are also correct, but they are ``asleep'' and do not take any steps until time $t_7$.

\item All the processes in $\{p_a\} \cup Q_2$ are Byzantine (note that $|\{p_a\} \cup Q_2| \le f$).

\item The Byzantine processes in $\{p_a\} \cup Q_2$ behave exactly as in $H_2$.
        More precisely, they write exactly the same values into their SWMR registers at the same time as they did in $H_2$.
		
\item The processes in $\{p_b\} \cup Q_3$ behave exactly as in $H_2$ up to and including time $t_7$. 
That is: (a) they take no steps before $t_6$, and (b) at time~$t_6$, they ``wake up'' and
	process $p_b$ invokes a $\rTest'$ operation, and this operation returns 1 at time $t_7$.
This is possible because, up to time $t_7$ the processes in $\{p_b\} \cup Q_3$ cannot distinguish between $H_2$ and~$H_3$. To see why, note that in both $H_2$ and $H_3$:
    \begin{itemize}
        \item All the registers of all processes in $\{s\} \cup Q_1$ are in their initial state from the time $t_6$ when the processes in $\{p_b\} \cup Q_3$ ``wake up'' to the time $t_7$ when $p_b$'s $\rTest'$ returns 1.
        \item All the SWMR registers of the processes in $\{p_a\} \cup Q_2$ contain the same values, at the same times.
    \end{itemize}

\item At time $t_7$ the processes in $\{s\} \cup Q_1$, who are correct but were ``asleep'',
    wake up and start taking steps.

	\end{itemize}

Since $\AW$ is a correct implementation of a {\br},
    by Lemma \ref{TheSaver}(\ref{Uno}), the history $H_3$ of $\AW$ must satisfy the following property:
    if a $\rTest$ operation by a \emph{correct} tester returns $1$, and the setter is \emph{correct},
    then the setter invokes a $\rSet$ operation before the $\rTest$ returns 1.
But $H_3$ violates this property because
    the $\rTest'$ operation by the \emph{correct} tester $p_b$ returns $1$ at time $t_7$,
    and the setter $s$ is \emph{correct}, but $s$ does \emph{not} invoke a $\rSet$ operation before time $t_7$ --- a contradiction.
\end{proof}

We now show that there is a correct implementation of the {\br} object
    from a {\mar}, from an {\tar}, and also from a {\sr}.
These simple implementations are given below.

\smallskip
\textbf{Implementation of {\br} from a {\mar}.}
Let $R$ be a SWMR {\mar} initialized to $v_0=0$.
We use $R$ to implement {\br} as follows:

\begin{compactitem}
    \item To $\rSet$, the writer executes $\Set(1)$ followed by $\sign(1)$ on $R$.
    \item To $\rTest$, a reader executes $\valid(1)$ on $R$.
    If $\valid(1)$ returns $\true$,
         then the $\rTest$ returns~1;
    otherwise, the $\rTest$ returns 0.
\end{compactitem}

It is easy to see that this implementation is Byzantine linearizable.
Intuitively, the linearization point of a $\rSet$ operation is when its $\sign(1)$ occurs,
    and the linearization point of a $\rTest$ operation is when its $\valid(1)$ occurs.

\smallskip
\textbf{Implementation of {\br} from an {\tar}.}
Let $R$ be a SWMR {\tar} initialized to $v_0=0$.
We use $R$ to implement {\br} as follows:

\begin{compactitem}
    \item To $\rSet$, the writer executes $\Set(1)$ on $R$.
    \item To $\rTest$, a reader executes $\valid(1)$ on $R$.
    If $\valid(1)$ returns $\true$,
         then the $\rTest$ returns~1;
    otherwise, the $\rTest$ returns 0.
\end{compactitem}

It is easy to see that this implementation is Byzantine linearizable.
Intuively, the linearization point of a $\rSet$ operation is when its $\Set(1)$ occurs,
    and the linearization point of a $\rTest$ operation is when its $\valid(1)$ occurs.

\smallskip
\textbf{Implementation of {\br} from a {\sr}.}
Let $R$ be a SWMR {\sr} initialized to $\bot$.
 We use $R$ to implement {\br} as follows:

\begin{compactitem}
    \item To $\rSet$, the writer executes a $\Set(1)$ on $R$.
    \item To $\rTest$, a reader executes a $\Test$ on $R$.
        If $\Test$ returns 1,
         then the $\rTest$ returns 1;
    otherwise, the $\rTest$ returns 0.
\end{compactitem}

It is easy to see that this implementation is Byzantine linearizable.
Intuitively, the linearization point of a $\rSet$ operation is when its $\Set(1)$ occurs,
    and the linearization point of a $\rTest$ operation is when its $\Test$ occurs.

Note that the above implementations are \emph{wait-free}~\cite{Herlihy88}, and so:

\begin{observation}\label{isthisalemma}
In a system with $n>f$ processes, where $f$ processes can be Byzantine,
    there is a correct
    implementation of {\br} from a SWMR {\mar},
    from a SWMR {\tar},
    and from~a~SWMR~{\sr}.
\end{observation}

Theorem~\ref{Theo-Impossibility-Result} and Observation~\ref{isthisalemma} imply:

\begin{theorem}\label{BIGTheo-Impossibility-Result}
In a system with $3\le n\le 3f$ processes, where $f$ processes can be Byzantine,
    there are no correct implementations
    of SWMR {\mar}s,
    of SWMR {\tar}s, or
    of SWMR {\sr}s,
    from SWMR registers.
\end{theorem} 	

For $n=2$ there is only one writer and one reader.
It is trivial to implement {\mar}s, {\tar}s, and {\sr}s from SWMR registers in this case.

\section{Conclusion}\label{conclusion}

In this paper, we propose three types of registers that provide
    some properties of unforgeable digital signatures,
    and show how to implement them from plain SWMR registers without actually using signatures.

We believe that these registers can be used to simplify the design of algorithms for
    systems for Byzantine processes by breaking the design process into two steps:
    first develop an algorithm that assumes that processes can write \emph{and sign} values with signatures into SWMR registers;
    if the signature properties that the algorithm relies on are provided
    by the properties of {\mar}s, {\tar}s, or {\sr}s,
    the algorithm can be transformed into a signature-free algorithm by using our corresponding 
    register implementations when $n>3f$.

The following are open problems: what \emph{other} properties of digital signatures can be used to solve problems in systems with registers?
Can we achieve these properties without signatures, and if so, how?
The same questions apply when digital signatures are used in systems where processes share objects \emph{other than registers.}

\section{Acknowledgement}
    We thank Myles Thiessen and the anonymous reviewers for their helpful comments on this paper.
    This work was partially supported by the NSERC Discovery award RGPIN-2022-03304.

\bibliographystyle{plain}
\bibliography{biblio}

\newpage
\appendix

\section{Correctness of the {\marc} Implementation}\label{a-mar}

In the following, we prove that Algorithm~\ref{code-mar} is a
    correct implementation of a SWMR {\mar} for systems 
    with $n >3f$ processes, $f$ of which can be Byzantine.
To do so, in the following, we consider an arbitrary (infinite) history $H$ of Algorithm~\ref{code-mar} (where $n>3f$), and prove:
    \begin{compactitem}
        \item \textsc{[Byzantine linearizability]} The history $H$ is Byzantine linearizable with respect to a SWMR {\mar} (Section~\ref{linear-mar}).
        \item \textsc{[Termination]} All processes that are correct in $H$ complete all their operations (Section \ref{wait-free-mar}).
    \end{compactitem}

\medskip\noindent
\textbf{Notation.}
Since each register $R_i$ stores a set of values, we write $v \in R_i$,
    and say that ``$v$ is in $R_i$'', if the set stored in $R_i$ contains the value $v$.

\subsection{Termination}\label{wait-free-mar}
\begin{observation}\label{monock-mar}
	For every {\ct} reader $p_k\in \{p_2,\ldots,p_n \}$, 
		the value of $C_k$ is non-decreasing.
\end{observation}

\begin{observation}\label{correctonlywrites1-mar}

For every {\ct} process $p_i\in \{p_1,\ldots,p_n\}$,
	if $v \in R_i$ at time $t$, then $v \in R_i$ at all times $t' \ge t$.
\end{observation}

\begin{observation}\label{monoset1-mar}
	Let $\valid(-)$ be an operation by a {\ct} reader $p_k$.
	In $\valid(-)$,
	 $|\set_1|$ of $p_k$ is non-decreasing.
\end{observation}

	\begin{lemma}\label{asker1-mar}
		Suppose there is a time $t$ such that
		$v\in R_1$ at all times $t' \ge t$, or at least $f+1$ processes $p_i$ have $v \in R_i$ at all times $t' \ge t$.
		Consider any iteration of the while loop of the $\fresh()$ procedure executed by a {\ct} process $p_j$.
		If $p_j$ inserts $p_k$ into $askers$ in line~\ref{askers-m} at some time $t_a \ge t$, 
		then $p_j$ writes $\langle r_j , - \rangle$ with $v\in r_j$ into $R_{jk}$ at line~\ref{fresh1-m}.
		\end{lemma} 
	
	\begin{proof}
		Suppose there is a time $t$ such that
		$v\in R_1$ at all times $t' \ge t$, or at least $f+1$ processes $p_i$ have $v \in R_i$ at all times $t' \ge t$.
		Consider any iteration of the while loop of the $\fresh()$ procedure executed by a {\ct} process $p_j$.
		Suppose $p_j$ inserts $p_k$ into $askers$ in line~\ref{askers-m} at some time $t_a \ge t$.
        In line~\ref{sets-m}, $p_j$ reads every $R_i, 1\le i \le n$, after time~$t$.
        There are two cases:
        \begin{enumerate}
             \item $v \in R_1$ at all times $t' \ge t$.
                Then $p_j$ reads $R_1$ with $v\in R_1$ at line~\ref{sets-m} after time $t$.
                So $p_j$ finds that the condition $v \in r_1$ holds in line~\ref{followcondition-m}.
            \item At least $f+1$ processes $p_i$ have $v \in R_i$ at all times $t' \ge t$.
                Then there are at least $f+1$ processes $p_i$ such that $p_j$ reads $R_i$ with $v\in R_i$ at line~\ref{sets-m} after time $t$.
                So $p_j$ finds that the condition $|\{ r_i|~ v \in r_i\}| \ge f+1$ holds in line~\ref{followcondition-m}.
        \end{enumerate}
        In both cases, $p_j$ finds that the condition $v \in r_1$ or $|\{ r_i|~ v \in r_i\}| \ge f+1$ holds in line~\ref{followcondition-m} and so inserts $v$ into $R_j$ at line~\ref{follow-m}.
        Therefore, $p_j$ writes $\langle r_j , - \rangle$ with $v\in r_j$ into $R_{jk}$ at line~\ref{fresh1-m}.
    \end{proof}

	\begin{lemma}\label{ri1beforeset1-mar}
		Let $\valid(v)$ be an operation by a {\ct} reader $p_k$.
		If a {\ct} process $p_j$ is in $\set_1$ of $p_k$ in $\valid(v)$ at time $t$,
			then $v \in R_j$ at all times $t' \ge t$.
		\end{lemma}
		\begin{proof}		
			Let $\valid(v)$ be an operation by a {\ct} reader $p_k$.
			Consider any time $t$ and any {\ct} process $p_j \in\set_1$ at time $t$.  
			Consider the iteration of the while loop at line~\ref{whileloop-m}
                of $\valid(v)$ in which 
				$p_k$ inserts $p_j$ into $\set_1$ at line~\ref{set1-m},
				say at time $t_k^1 \le t$.
			By lines~\ref{readri-m} and \ref{check1-m}, 
				$p_k$ reads $\langle r, - \rangle$ with $v\in r$ from $R_{jk}$ at line~\ref{readri-m},
				say at time $t_k^0 < t_k^1$.
			Since $p_j$ is {\ct} and $R_{jk}$ is initialized to $\langle \emptyset, - \rangle$, 
				$p_j$ wrote $\langle r, - \rangle$ with $v\in r$ into $R_{jk}$ at line~\ref{fresh1-m} by time $t_k^0$.
			By line~\ref{readrj-m} and line~\ref{fresh1-m},
				$v \in R_j$ by time $t_k^0$.
			Since $t_k^0 < t_k^1$ and $t_k^1 \le t$,
				$t_k^0 < t$.
			Since $p_j$ is {\ct},
				by Observation~\ref{correctonlywrites1-mar},
				$v \in R_j$ at all times $t' \ge t$.
		\end{proof}

\begin{lemma}\label{onecorrectset0-mar}
	Let $\valid(v)$ be an operation by a {\ct} reader $p_k$.
	For each iteration of the while loop at line~\ref{whileloop-m} of $\valid(v)$, 
		the following loop invariants hold at line~\ref{whileloop-m}:
	\begin{enumerate}[(1)]
             \item\label{inv0-m} $\set_1$ and $\set_0$ are disjoint.
		\item\label{inv1-m} $|\set_1| < n-f$ and $|\set_0| \le f$.
		\item\label{inv2-m} If there are at least $f+1$ {\ct} processes in $\set_1$,
				there is no {\ct} process in $\set_0$.
	\end{enumerate}
\end{lemma}

\begin{proof}
	Let $\valid(v)$ be an operation by a {\ct} reader $p_k$.
    We now prove the invariants by induction on the number of iterations of the while loop at line~\ref{whileloop-m} of $\valid(v)$.
    \begin{itemize}
    \item Base Case: 
        Consider the first iteration of the while loop at line~\ref{whileloop-m} of $\valid(v)$.
        Since $p_k$ initializes $\set_1$ and $\set_0$ to $\emptyset$,
            (\ref{inv0-m}) holds trivially.
        Furthermore,
            $p_k$ has $|\set_1| = 0$ and $|\set_0| = 0 $.
        Since $n \ge 3f+1$ and $f \ge 1$,
            $|\set_1| = 0 < n-f$ and $|\set_0|= 0 \le f$ and so (\ref{inv1-m}) holds.
        Since $\set_0$ is empty,
        (\ref{inv2-m}) holds trivially.
	
    \item Inductive Case:
        Consider any iteration $I$ of the while loop at line~\ref{whileloop-m} of $\valid(v)$.
        Assume that (\ref{inv0-m}), (\ref{inv1-m}), and (\ref{inv2-m}) hold at the beginning of the iteration.
			
        If $p_k$ does not find that the condition of line~\ref{until} holds in $I$,
            $p_k$ does not move to the next iteration of the while loop,
            and so (\ref{inv0-m}), (\ref{inv1-m}), and (\ref{inv2-m}) trivially hold at the start of the ``next iteration'' (since it does not exist).
        Furthermore, if $\valid(v)$ returns at line~\ref{return1-m} or line~\ref{return0-m} in $I$,
            (\ref{inv0-m}), (\ref{inv1-m}), and (\ref{inv2-m}) trivially hold
            at the start of the ``next iteration'' (since it does not exist).

        We now consider any iteration $I$ of the while loop at line~\ref{whileloop-m} of $\valid(v)$
            in which $p_k$ finds the condition of line~\ref{until} holds 
            and $\valid(v)$ does not return at line~\ref{return1-m} or line~\ref{return0-m}.
        We first show that (\ref{inv0-m}) 
            remains true at the end of $I$.
        Since $p_k$ finds the condition of line~\ref{until} holds,
            $p_k$ inserts a process, say $p_j$, into $\set_1$ or $\set_0$ in $I$.
        By line~\ref{until},
            $p_j\not\in \set_1\cup\set_0$.
        Note that $p_j$ is the only process that $p_k$ inserts into $\set_1$ or $\set_0$ in $I$.
        So since $\set_1$ and $\set_0$ are disjoint at the beginning of $I$ and $p_j\not\in \set_1\cup\set_0$,
             $\set_1$ and $\set_0$ are disjoint at the end of $I$.
        Therefore (\ref{inv0-m}) holds at the end of $I$.

        We now show that (\ref{inv1-m}) remains true at the end of $I$.
        Since $\valid(v)$ does not return at line~\ref{return1-m} or line~\ref{return0-m},
            $p_k$ finds that $|\set_1| < n-f$ holds at line~\ref{return1-m} and $|\set_0| \le f$ holds at line~\ref{return0-m}.
        Since $p_k$ does not change $\set_1$ after line~\ref{return1-m} and does not change $\set_0$ after line~\ref{return0-m} in $I$,
                $|\set_1| < n-f$ and $|\set_0| \le f$ remain true at the end of $I$ and so (\ref{inv1-m}) holds.

        We now show that (\ref{inv2-m}) remains true at the end of $I$.
        There are two cases:
        \begin{enumerate}
            \item Case 1: $p_k$ executes line~\ref{set1-m} in $I$.
                Then $p_k$ changes $\set_0$ to $\emptyset$ at line~\ref{empty0-m}.
                Since $p_k$ does not change $\set_0$ after line~\ref{empty0-m} in $I$,
                    $\set_0$ remains $\emptyset$ at the end of $I$.
                So (\ref{inv2-m}) holds trivially at the end of $I$.

        \item Case 2: $p_k$ does not execute line~\ref{set1-m} in $I$.
                So $p_k$ does not change $\set_1$ in $I$.
                By assumption,
                    $p_k$ finds the condition of line~\ref{until} holds in $I$.
                    So $p_k$ executes line~\ref{set0-m} in $I$.
                Let $p_a$ be the process that $p_k$ inserts into $\set_0$ in line~\ref{set0-m}.
                There are two cases:
                \begin{itemize}
                \item Case 2.1: $p_a$ is faulty. 
                    Then the number of {\ct} processes in $\set_0$ does not change in $I$.
                    So, since $\set_1$ also does not change in $I$,
                    (\ref{inv2-m}) remains true at the end of $I$.

                \item Case 2.2: $p_a$ is {\ct}. 
                There are two cases:

                \begin{itemize}
                \item Case 2.2.1: There are fewer than $f+1$ {\ct} processes in $\set_1$ at the beginning of $I$.
                    Since $p_k$ does not change $\set_1$ in $I$, 
                    (\ref{inv2-m}) remains true at the end of $I$.

                \item Case 2.2.2: There are at least $f+1$ {\ct} processes in $\set_1$ at the beginning of $I$.
                    We now show that this case is impossible.

                    Let $p_b$ be the last process that $p_k$ inserted into $\set_1$ before the iteration $I$;
                    say $p_b$ was inserted into $\set_1$ at time $t_k^0$.
                    Since there are at least $f+1$ {\ct} processes in $\set_1$ at the beginning of $I$
                        and $p_b$ is the last process that $p_k$ inserted into $\set_1$ before $I$,
                        there are at least $f+1$ {\ct} processes in $\set_1$ of $p_k$ in $\valid(v)$ at time $t_k^0$.
                    By Lemma~\ref{ri1beforeset1-mar},
                        at least $f+1$ {\ct} processes $p_i$ have $v \in R_i$ at all times $t'\ge t_k^0$~($\star$).

                    Recall that $p_k$ inserts the correct process $p_a$ into $\set_0$ in the iteration $I$.
                    So in $I$, $p_k$ increments $C_k$ at line~\ref{ckplus-m}, say at time  $t_k^1$.
                    Note that $t_k^1 > t_k^0$. 
                    Let $c^*$ be the value of $C_k$ after $p_k$ increments $C_k$ at time~$t_k^1$.
                    Since $C_k$ is initialized to 0,
                        by Observation~\ref{monock-mar}, $c^* \ge 1$.
                    By line~\ref{readri-m}, line~\ref{until}, and line~\ref{notv-mar},
                        $p_k$ reads $\langle r, c \rangle$ with $v\not\in r$ and $c \ge c^*$ from $R_{ak}$ at line~\ref{readri-m} in $I$.
                    Since $c \ge c^* \ge 1$ and $R_{ak}$ is initialized to $\langle \emptyset,0\rangle$,
                        it must be that $p_a$ wrote $\langle r, c \rangle$ with $v\not\in r$ into $R_{ak}$ at line~\ref{fresh1-m} (**).

                    Consider the iteration of the while loop of the $\fresh()$ procedure in which $p_a$ writes $\langle r, c \rangle$ with $v\not\in r$ into $R_{ak}$ at line~\ref{fresh1-m}.
                    Note that $c$ is the value that $p_a$ read from $C_k$ in line~\ref{collectck-m} of this iteration;
                         say this read occurred at time $t_a^1$.
                    Since $c \ge c^*$,
                        by Observation~\ref{monock-mar},
                        $t_a^1 \ge t_k^1$.
                    Then $p_a$ inserts $p_k$ into $askers$ at line~\ref{askers-m}, say at time $t_a^2 > t_a^1$.
                    Since $t_a^1 \ge t_k^1$ and $t_k^1 > t_k^0$,
                        $t_a^2 > t_k^0$.
                    So $p_a$ inserts $p_k$ into $askers$ at line~\ref{askers-m} after $t_k^0$.
                    Thus, by ($\star$) and Lemma~\ref{asker1-mar},
                        $p_a$ writes $\langle r , c \rangle$ with $v\in r$ into $R_{ak}$ at line~\ref{fresh1-m} in this iteration,
                        a contradiction to (**). 
                    So this case is impossible.

                \end{itemize}
                \end{itemize}
				\end{enumerate}
				So in all the possible cases, we showed that  (\ref{inv0-m}), (\ref{inv1-m}), and (\ref{inv2-m}) remain true at the end of the iteration.
			\end{itemize}
	\end{proof}

    \begin{lemma}\label{correctoutside-mar}
    Let $\valid(v)$ be an operation by a {\ct} reader $p_k$.
            Every time when $p_k$ executes line~\ref{whileloop-m} of $\valid(v)$, 
    	there is a {\ct} process $p_i \not \in \set_1\cup\set_0$.
    \end{lemma} 	
    \begin{proof}
        Let $\valid(v)$ be an operation by a {\ct} reader $p_k$.
        Suppose $p_k$ executes line~\ref{whileloop-m} of $\valid(v)$ at time $t$.
        Consider (the values of) $\set_0$ and $\set_1$ at time $t$.
        By Lemma~\ref{onecorrectset0-mar}(\ref{inv0-m}),
            $\set_0$ and $\set_1$ are disjoint sets.

        We now prove that $\set_1\cup\set_0$ contains fewer than $n-f$ {\ct} processes; this immediately implies that there is a {\ct} process $p \not \in \set_1 \cup \set_0$.
            There are two possible cases:
			\begin{enumerate}
				\item Case 1: $\set_1$ contains at most $f$ {\ct} processes.
				By Lemma~\ref{onecorrectset0-mar}(\ref{inv1-m}), $|\set_0| \le f$.
				So $\set_1 \cup \set_0$ contains at most $2f$ {\ct} processes.
				Since $3f < n$, we have $2f < n-f$.
				
				\item Case 2: $\set_1$ contains at least $f+1$ {\ct} processes.
				By Lemma~\ref{onecorrectset0-mar}(\ref{inv2-m}), $\set_0$ does \emph{not} contain any {\ct} process.
				So $\set_1 \cup \set_0$ contains at most $|\set_1|$ {\ct} processes.
				By Lemma~\ref{onecorrectset0-mar}(\ref{inv1-m}),  $|\set_1| < n-f$.
				\end{enumerate} 
				In both cases, $\set_1\cup\set_0$ contains fewer than $n-f$ {\ct} processes.
				\end{proof}

\begin{lemma}\label{line7inf-mar}
	Let $\valid(v)$ be an operation by a {\ct} reader $p_k$.
	Every instance of the Repeat-Until loop at lines~\ref{repeat}-\ref{until} of $\valid(v)$ terminates.
\end{lemma}

\begin{proof}
	Let $\valid(v)$ be an operation by a {\ct} reader $p_k$.
	Assume for contradiction that there is an instance of the Repeat-Until loop at lines~\ref{repeat}-\ref{until} of $\valid(v)$ that does not terminate.
	Let $t$ be the time when $p_k$ executes line~\ref{repeat} for the first time in this instance of the Repeat-Until loop.
        Since this instance does not terminate, 
		$p_k$ never finds the condition of line~\ref{until} holds after $t$ ($\star$).
	
        Since lines~\ref{whileloop-m}-\ref{repeat} do not change $\set_1$ and $\set_0$,
        Lemma~\ref{correctoutside-mar} implies that 
            there is a {\ct} process $p_a \not\in \set_1\cup\set_0$ in line~\ref{repeat} at time $t$.
	Let $c^*$ be the value of $C_k$ at time $t$;
    	by line~\ref{ckplus-m},
    		$c^* \ge 1$.
        Since lines~\ref{repeat}-\ref{until} of the Repeat-Until loop do not change $\set_1$, $\set_0$, or $C_k$,
            and $p_k$ never exits this loop,
           $p_a \not\in \set_1\cup\set_0$ and $C_k=c^* $ at all times $t' \ge t$.

	\begin{claim}\label{rakck-mar}
		There is a time $t'$ such that $R_{ak}$ contains $\langle -, c^* \rangle$ for all times after $t'$.
	\end{claim}
	\begin{proof}
	Let $\fresh()$ be the help procedure of $p_a$.
	Since $C_k = c^*$ for all times after $t$ and $p_a$ is {\ct},
		there is an iteration of the while loop of $\fresh()$
		in which $p_a$ reads $c^*$ from $C_k$ at line~\ref{collectck-m}.
	Consider the \emph{first} iteration of the while loop of $\fresh()$ in which $p_a$ reads $c^*$ from $C_k$ at line~\ref{collectck-m}.
	Since $p_k$ is {\ct}, 
				by Observation~\ref{monock-mar},
				$p_a$ reads non-decreasing values from $C_k$,
				and so $p_a$ has $ c^* \ge prev\_c_k$ at line~\ref{askers-m}.
	Since $c^* \ge 1$,
		there are two cases.
		\begin{itemize}
			\item Case 1: $c^*=1$.
			Then $p_a$ has $prev\_c_k \le 1$ at line~\ref{askers-m}.
			Since $prev\_c_k$ is initialized to 0 (line~\ref{collectck-init}) and $p_a$ reads $c^*=1$ from $C_k$ at line~\ref{collectck-m} for the first time,
				$prev\_c_k = 0$ at line~\ref{askers-m}. 
			So $p_a$ finds $c^*=1 > prev\_c_k=0$ holds and inserts $p_k$ into $askers$ at line~\ref{askers-m}. 
			\item Case 2: $c^*>1$.
			Since $p_a$ reads $c^*$ from $C_k$ at line~\ref{collectck-m} for the first time and $ c^* \ge prev\_c_k$,
				$c^* > prev\_c_k$ at line~\ref{askers-m}. 
			So $p_a$ inserts $p_k$ into $askers$ at line~\ref{askers-m}. 
		\end{itemize}
	So in both cases, $p_a$ inserts $p_k$ into $askers$ at line~\ref{askers-m}. 
	Since $p_a$ is {\ct},
		in the same iteration of the while loop of $\fresh()$,
		$p_a$ writes $\langle - , c^* \rangle$ into $R_{ak}$ at line~\ref{fresh1-m},
		say at time $t'$ 
		and then it sets $prev\_c_k$ to $c^*$ in line~\ref{setprev-m} (i).
	Since $C_k = c^*$ for all times after $t$,
		by Observation~\ref{monock-mar},
		$C_k= c^*$ for all times after $t'$.
	Furthermore, by line~\ref{collectck-m},
		$p_a$ has $c_k=c^*$ for all times after $t'$ (ii).
	From (i) and (ii),
		 by line~\ref{askers-m},
		$p_a$ does not insert $p_k$ into $askers$ in any future iteration of the while loop of $\fresh()$.
	So by line~\ref{tellasker-m},
		$p_a$ never writes to $R_{ak}$ after $t'$,
		i.e., $R_{ak}$ contains $\langle -, c^* \rangle$ for all times after $t'$.
	\end{proof}

	Since $p_k$ executes lines~\ref{findone-m}-\ref{until} infinitely often after $t$,
		$p_k$ reads from $R_{ak}$ at line~\ref{readri-m} infinitely often after $t$.
	By Claim~\ref{rakck-mar},
		eventually $p_k$ reads $\langle -, c^* \rangle$ from $R_{ak}$ after $t$.
	Thus, since $p_a \not\in \set_1 \cup\set_0$ and $p_k$ has $C_k =c^*$ for all times after $t$, 
		$p_k$ finds the condition of line~\ref{until} holds after $t$
		--- a contradiction to ($\star$). 
\end{proof}

\begin{observation}\label{set1increment-mar}
	Let $\valid(v)$ be an operation by a {\ct} reader $p_k$.
	When $p_k$ executes line~\ref{set1-m} of $\valid(v)$, $p_k$ increments the size of $\set_1$.
\end{observation}

\begin{observation}\label{set0increment-mar}
	Let $\valid(v)$ be an operation by a {\ct} reader $p_k$.
	When $p_k$ executes line~\ref{set0-m} of $\valid(v)$, $p_k$ increments the size of $\set_0$.
\end{observation}

\begin{restatable}{theorem}{verifyt} \label{termination-mar}
    [\textsc{Termination}] Every $\Test$, $\Set$, $\sign$, and $\valid$ operation by a {\ct} process completes.
\end{restatable}

\begin{proof}
	From the code of the $\Set(-)$, $\sign(-)$, and $\Test()$ procedures,
		it is obvious that every $\Set$, $\sign$, and $\Test$ operation by a {\ct} process completes.
	We now prove that every $\valid$ operation by a {\ct} reader also completes.

	Assume for contradiction, 
		there is a $\valid(v)$ operation by a {\ct} reader $p_k$ that does not complete,
		i.e.,
		$p_k$ takes an infinite number of steps in the $\valid(v)$ procedure.
	By Lemma~\ref{line7inf-mar},
        $p_k$ must execute infinitely many iterations of the while loop at line~\ref{whileloop-m} of $\valid(v)$.
	So $p_k$~executes line~\ref{set1-m} or line~\ref{set0-m} of $\valid(v)$ infinitely often.	
	\begin{itemize}
		\item Case 1: $p_k$ executes line~\ref{set1-m} infinitely often.	
				By Observations~\ref{monoset1-mar} and \ref{set1increment-mar},
					there is an iteration of the while loop at line~\ref{whileloop-m} of this $\valid(v)$ 
					in which $p_k$ has $|\set_1| \ge n-f$ at line~\ref{set1-m}. 
				Since $p_k$ does not change $|\set_1|$ between line~\ref{set1-m} and line~\ref{return1-m},
					in that iteration,
					$p_k$ finds that the condition $|\set_1| \ge n-f$ holds in line~\ref{return1-m} of $\valid(v)$.
				So $\valid(v)$ returns at line~\ref{return1-m}.
		\item Case 2: $p_k$ executes line~\ref{set1-m} only a finite number of times.
			Then $p_k$ executes line~\ref{set0-m} infinitely often.
			So there is a time $t$ such that 
				$p_k$ executes line~\ref{set0-m} infinitely often after $t$
				while $p_k$ never executes line~\ref{set1-m} after $t$.
			This implies $p_k$ never executes line~\ref{empty0-m} after $t$.
			Since $|\set_0|$ decreases only at line~\ref{empty0-m},
				$|\set_0|$ never decreases after $t$.
			Since $p_k$ executes line~\ref{set0-m} infinitely often after $t$
				and $|\set_0|$ never decreases after $t$,
				by Observations~\ref{set0increment-mar},
					there is an iteration of the while loop at line~\ref{whileloop-m} of $\valid(v)$
					in which $p_k$ has $|\set_0| > f$ at line~\ref{set0-m} of $\valid(v)$. 
			Since $p_k$ does not change $|\set_0|$ between line~\ref{set0-m} and line~\ref{return0-m},
				in that iteration,
				$p_k$ finds that the condition $|\set_0|>f$ holds in line~\ref{return0-m} of $\valid(v)$.
			So $\valid(v)$ returns at line~\ref{return0-m}.
	\end{itemize}
	In both cases,
		this $\valid(v)$ returns, 
			a contradiction to the assumption that $\valid(v)$ does not complete.
\end{proof}

\subsection{Byzantine Linearizability}\label{linear-mar}
Recall that $H$ is an arbitrary history of the implementation given by Algorithm~\ref{code-mar}.
We now prove that $H$ is Byzantine linearizable with respect to a SWMR {\mar} (Definition~\ref{def-mar}).
We start by proving the following observations and lemmas.

\begin{lemma}~\label{notreturn0-mar}
	Suppose there is a time $t$ such that
		$v \in R_1$ at all times $t' \ge t$ or at least $f+1$ processes $p_i$ have $v \in R_i$ at all times $t' \ge t$.
	If a {\ct} reader invokes a $\valid(v)$ operation after time $t$, 
		then it does not insert any {\ct} process into $\set_0$ in this $\valid(v)$ operation.
\end{lemma}

\begin{proof}
	Assume for contradiction,
		there is a time $t$ such that
		(1) $v \in R_1$ at all times $t' \ge t$ or at least $f+1$ processes $p_i$ have $v \in R_i$ at all times $t' \ge t$, and
		(2) a {\ct} reader $p_k$ invokes a $\valid(v)$ operation after time $t$,  
		but $p_k$ inserts a {\ct} process $p_j$ into $\set_0$ in this $\valid(v)$ operation.
 
	Consider the iteration of the while loop at line~\ref{whileloop-m} of the $\valid(v)$ in which $p_k$ inserts $p_j$ into $\set_0$ at line~\ref{set0-m}. 
	Let $c_k$ be the value of $C_k$ after $p_k$ increments $C_k$ at line~\ref{ckplus-m} in this iteration.
	Since $C_k$ is initialized to $0$, 
		$c_k \ge 1$ and $p_k$ writes $c_k$ into $C_k$ after time $t$.
	By Observation~\ref{monock-mar},
		for all times $t' \le t$, $C_k$ contains values that are less than $c_k$ ($\star$).
	Since $p_k$ inserts $p_j$ into $\set_0$ at line~\ref{set0-m},
		$p_k$ reads $\langle r, c \rangle$ with $v\not\in r$ and $c \ge c_k$ from $R_{jk}$ at line~\ref{readri-m}.
	Since $c_k \ge 1$, $c\ge 1$.
	Since $R_{jk}$ is initialized to $\langle \emptyset, 0 \rangle$,
		it must be that $p_j$ writes $\langle r, c \rangle$ with $v\not\in r$ into $R_{jk}$ 
		at line~\ref{fresh1-m} in some iteration of the while loop of $\fresh()$ ($\star\star$).

	Consider the iteration of the while loop of $\fresh()$ procedure
		in which $p_j$ writes $\langle r, c \rangle$ into $R_{jk}$ at line~\ref{fresh1-m}.
	Since $p_j$ is {\ct}, in that iteration:	
		(a) $p_j$ reads $c$ from $C_k$ at line~\ref{collectck-m}, say at time $t_j^1$ and then 
		(b) $p_j$ inserts $p_k$ into $askers$ at line~\ref{askers-m}, say at time $t_j^2 > t_j^1$.
	Since $c \ge c_k$,
		by ($\star$),
		$p_j$ reads $c$ from $C_k$ at line~\ref{collectck-m} after time $t$, i.e., $t_j^1 > t$.
	Since $t_j^2 > t_j^1$, $t_j^2 > t$, 
		i.e., $p_j$ inserts $p_k$ into $askers$ at line~\ref{askers-m} at time $t_j^2 > t$.
	Thus, by (1) and Lemma~\ref{asker1-mar}, 
		$p_j$ writes $\langle r, c \rangle$ with $v\in r$ into $R_{jk}$ at line~\ref{fresh1-m} --- a contradiction to ($\star\star$).
\end{proof}

\begin{lemma}~\label{testreturn1-mar}
	Suppose there is a time $t$ such that $v \in R_1$ at all times $t' \ge t$ or at least $f+1$ processes $p_i$ have $v \in R_i$ at all times $t' \ge t$.
	If a {\ct} reader invokes a $\valid(v)$ operation after time $t$,  
		then $\valid(v)$ returns $\true$.
\end{lemma}

\begin{proof}
	Suppose that there is a time $t$ such that
        (1) $v \in R_1$ at all times $t' \ge t$ or at least $f+1$ processes $p_i$ have $v \in R_i$ at all times $t' \ge t$, and 
		(2) a {\ct} reader $p_k$ invokes a $\valid(v)$ operation after time $t$.
        Note that $\set_0$ is initialized to $\emptyset$ at the start of this operation.
	By Lemma~\ref{notreturn0-mar},
            $p_k$ does not insert any {\ct} process into $\set_0$ in this $\valid(v)$ operation.
            Thus, $\set_0$ can contain only faulty processes, and so 
		$|\set_0| \le f$ in $\valid(v)$.
	So $p_k$ never finds that the condition $|\set_0| > f$ holds in line~\ref{return0-m} of $\valid(v)$ and $\valid(v)$            never returns $\false$ at line~\ref{return0-m}.
	Thus,
		by Theorem~\ref{termination-mar} and the code for the $\valid(v)$ procedure,
		$\valid(v)$ returns $\true$.
\end{proof}

The following lemma captures the ``relay'' property of signed values: intuitively,
    if a process validates (the signature of) a value $v$,
    thereafter every process will also be able to validate it.
    
\begin{lemma}\label{testatestb-mar}
	Let $\valid(v)$ and $\valid(v)'$ be operations by {\ct} readers.
	If $\valid(v)$ returns $\true$ and it precedes $\valid(v)'$, then $\valid(v)'$ also returns $\true$.
\end{lemma}

\begin{proof}
	Let $\valid(v)$ and $\valid(v)'$ be operations by {\ct} readers $p_a$ and $p_b$ respectively.
	Suppose that $\valid(v)$ precedes $\valid(v)'$ and $\valid(v)$ returns $\true$.
	Let $t$ be the time when $\valid(v)$ returns $\true$.
	Since $\valid(v)$ precedes $\valid(v)'$,
		$p_b$ invokes $\valid(v)'$ after time~$t$~($\star$).

	Since $\valid(v)$ returns $\true$ at time $t$,
		$p_a$ finds that the condition $|\set_1| \ge n-f$ holds in line~\ref{return1-m} of $\valid(v)$, say at time $t_a \le t$.
	Since $n \ge 3f+1$, $|\set_1| \ge 2f+1$ at time $t_a$.
	Since there are at most $f$ faulty processes,
		there are at least $f+1$ {\ct} processes in $\set_1$ of $p_a$ at time $t_a$.
	By Lemma~\ref{ri1beforeset1-mar},
		at least $f+1$ {\ct} processes $p_i$ have $v\in R_i$ at all times $t' \ge t_a$.
	Since $t_a < t$,
		at least $f+1$ {\ct} processes $p_i$ have $v\in R_i$ at all times $t' \ge t$.
	Thus, by ($\star$) and Lemma~\ref{testreturn1-mar}, $\valid(v)'$ returns~$\true$.	
\end{proof}

\begin{definition}~\label{t0t1-mar}
        For any value $v$,
	\begin{itemize}
		\item Let $t^v_0$ be the max invocation time of any $\valid(v)$ operation by a {\ct} reader that returns $\false$ in $H$;
		if no such $\valid(v)$ operation exists, $t^v_0 = 0$.
		\item Let $t^v_1$ be the min response time of any $\valid(v)$ operation by a {\ct} reader that returns $\true$ in $H$;
		if no such $\valid(v)$ operation exists, $t^v_1 = \infty$.
	\end{itemize}
\end{definition}

\begin{lemma}\label{iexists-mar}
      For any value $v$, $t^v_1 > t^v_0$ and so the interval $(t^v_0,t^v_1)$ exists.
\end{lemma}

\begin{proof}
    Let $v$ be any value.
    There are four cases:
    \begin{itemize}
         \item Case 1: no $\valid(v)$ operation by a {\ct} process returns $\false$,
            and no $\valid(v)$ operation by a {\ct} process returns $\true$.
            Then $t^v_0 = 0$ and $t^v_1 = \infty$.
            So $t^v_1 > t^v_0$.
        \item Case 2: no $\valid(v)$ operation by a {\ct} process returns $\false$ and some $\valid(v)$ operation by a {\ct} process returns $\true$.
            Then $t^v_0 = 0$ and  $t^v_1 > 0$.
            So $t^v_1 > t^v_0$.
        \item Case 3: some $\valid(v)$ operation by a {\ct} process returns $\false$
            and
            no $\valid(v)$ operation by a {\ct} process returns $\true$.
            Then  $t^v_0 < \infty$ and $t^v_1 = \infty$.
            So $t^v_1 > t^v_0$.
            
        \item Case 4: some $\valid(v)$ operation by a {\ct} process returns $\false$
        and
        some $\valid(v)$ operation by a {\ct} process returns $\true$.
            Let $\valid(v)'$ be the operation with the max invocation time of any $\valid(v)$ operation by a {\ct} process that returns $\false$. 
            By Definition~\ref{t0t1-mar}, the invocation time of $\valid(v)'$ is $t^v_0$~($\star$).
            Let $\valid(v)''$ be the operation with the min response time of any $\valid(v)$ operation by a {\ct} process that returns $\true$.
            By Definition~\ref{t0t1-mar}, the response time of $\valid(v)''$ is $t^v_1$ ($\star\star$).
                         
            Since $\valid(v)''$ returns $\true$ and $\valid(v)'$ returns $\false$, 
                by Lemma~\ref{testatestb-mar},
                $\valid(v)''$ does not precede $\valid(v)'$.
            So the response time of $\valid(v)''$ is greater than the invocation time of $\valid(v)'$.
            Thus, by ($\star$) and ($\star\star$),
                 $t^v_1 > t^v_0$.
    \end{itemize}
    Therefore in all cases,  $t^v_1 > t^v_0$ and so the interval $(t^v_0,t^v_1)$ exists.
\end{proof}

\begin{definition}
Let ${\hct}$ be the set of processes that are correct in the history $H$.
\end{definition}

\begin{definition}
Let $H|{\hct}$ be the history consisting of all the steps of all the \emph{{\ct}} processes in $H$ (at the same times they occur in $H$).
\end{definition}

By Theorem~\ref{termination-mar}, all the processes that are correct in $H$ (i.e., all the processes in {\hct}) complete their operations, so:

\begin{observation}\label{allcomplete-mar}
    Every operation in $H|{\hct}$ is complete (i.e., it has both an invocation and a response).
\end{observation}

\begin{observation}\label{correctop-mar}
    An operation $o$ by a process $p$ is in $H|{\hct}$ if and only if $o$ is also in $H$ and $p \in {\hct}$.
\end{observation}

\begin{observation}\label{sameop-mar}
For all processes $p \in {\hct}$,
    $p$ has the same operations (i.e., the same invocation and response operation steps) in both $H$ and $H|{\hct}$.
    Furthermore, the $[invocation, response]$ intervals of these operations are the same in both $H$ and $H|{\hct}$.
\end{observation}

By the Definition~\ref{t0t1-mar} of $t^v_0$ and $t^v_1$, and Observation~\ref{sameop-mar}, we have the following.

\begin{observation}\label{sameip-mar}
        In $H|{\hct}$, for any value $v$,

	\begin{itemize}
		\item The max invocation time of any $\valid(v)$ operation that returns $\false$ is $t^v_0$.
		\item The min response time of any $\valid(v)$ operation that returns~$\true$ is $t^v_1$.
	\end{itemize}
\end{observation}

To show that $H$ is Byzantine linearizable with respect to a SWMR \mar, 
    we must show that there is a history $H'$ such that:
    (a)~$H'|{\hct} = H|{\hct}$, and
    (b)~$H'$ is linearizable with respect to a SWMR {\mar}.

There are two cases, depending on whether the writer is {\ct} in $H$.

\subsubsection*{Case 1: the writer $p_1$ is {\ct} in $H$.}\label{case-writer-correct}
Let $H' = H|{\hct}$.
We now show that the history $H'$ is linearizable with respect to a SWMR \mar.
To do so, 
    we first define the linearization points of the operations in $H'$; 
and then we use these linearization points to define a linearization $L$ of $H'$ such that:
(a) $L$ respects the precedence relation between the operations of $H'$, and 
(b) $L$ conforms to the sequential specification of a SWMR {\mar} (which is given in Definition~\ref{def-mar}).

Note that to execute a $\Set(v)$, the {\ct} writer writes $v$ into $R^*$ at line~\ref{r-m} of $\Set(v)$.
Moreover, to execute a $\Test$, a {\ct} reader reads $R^*$ at line~\ref{read-m} of $\Test$,
    and returns the value that it reads.
Since $R^*$ is an atomic SWMR register (initialized to $v_0$), the following holds in $H$.
\begin{lemma}\label{register-mar-0}
        Let $\Test$ be an operation by a {\ct} reader that returns a value $v$ in $H$.
        Then:
        \begin{compactitem}
            \item either $v$ is the last value that the writer writes into $R^*$ at line~\ref{r-m} of a $\Set(v)$
            before the reader reads $v$ from $R^*$ at line~\ref{read-m} of the $\Test$ in $H$,
            \item or $v = v_0$ and the writer never writes into $R^*$ at line~\ref{r-m} of a $\Set(-)$ before the reader reads $v$ from $R^*$ at line~\ref{read-m} of the $\Test$ in $H$.
        \end{compactitem}

\end{lemma}

Since the writer is correct in $H$ and $H'=H|{\hct}$, by Observations~\ref{correctop-mar} and~\ref{sameop-mar},
    we have the following corollary.
\begin{corollary}\label{register-mar-00}
        Let $\Test$ be an operation that returns a value $v$ in $H'$.
	  Then:
        \begin{compactitem}
            \item either $v$ is the last value that the writer writes into $R^*$ at line~\ref{r-m} of a $\Set(v)$
            before the reader reads $v$ from $R^*$ at line~\ref{read-m} of the $\Test$ in $H'$,
            \item or $v = v_0$ and the writer never writes into $R^*$ at line~\ref{r-m} of a $\Set(-)$ before the reader reads $v$ from $R^*$ at line~\ref{read-m} of the $\Test$ in $H'$.
        \end{compactitem}
\end{corollary}

\begin{observation}\label{rmono-mar}
    In $H$, if the writer has $v \in r^*$ at time $t$,
        then $v \in r^*$ at all times $t' > t$.
\end{observation}

\begin{lemma}\label{writesignlinearization-mar}
    In $H$, a $\sign(v)$ returns $\success$ if and only if there is a $\Set(v)$ that precedes the $\sign(v)$.
\end{lemma}
\begin{proof}
We prove each direction separately:
\begin{compactitem}[]

    \item $[\Rightarrow]$ Assume a $\sign(v)$ returns $\success$.
        Since the writer $p_1$ is {\ct},
            $p_1$ returns $\success$ at line~\ref{success-m} of this $\sign(v)$.
        This implies $p_1$ finds that the condition $v \in r^*$ holds in line~\ref{checksign-m} of this $\sign(v)$, say at time $t_1$.
        Since $r^*$ is initially $\emptyset$,
            $p_1$ inserts $v$ into $r^*$ at line~\ref{addtoset-m} of a $\Set(v)$ at time $t_2 < t_1$.
        Since the writer $p_1$ is {\ct},
            all its operations are sequential and so this $\Set(v)$ precedes the $\sign(v)$.

    \item $[\Leftarrow]$ Assume there is a $\Set(v)$ that precedes a $\sign(v)$.
    Since the writer $p_1$ is {\ct},
        $p_1$ inserts $v$ into $r^*$ at line~\ref{addtoset-m} of the $\Set(v)$, say at time $t_1$.
    Since the $\Set(v)$ precedes the $\sign(v)$,
        $p_1$ reads $r^*$ at line~\ref{checksign-m} of the $\sign(v)$ at time $t_2 > t_1$.
    By Observation~\ref{rmono-mar},
       $p_1$ finds that the condition $v \in r^*$ holds in line~\ref{checksign-m} of the $\sign(v)$.
    Then this $\sign(v)$ returns $\success$ at line~\ref{success-m}.
        
\end{compactitem}    
\end{proof}

\begin{lemma}\label{settest-mar}
In $H$,
    \begin{enumerate}[(1)]
        \item
        If a $\valid(v)$ by a {\ct} reader returns $\false$, and
        the writer $p_1$ inserts $v$ into $R_1$ at line~\ref{setter1-m} of a {\sign(v)} at some time $t$,
        then the reader invokes this $\valid(v)$ before time $t$.
        
        \item If a $\valid(v)$ by a {\ct} reader returns $\true$ at time $t$,
	then the writer $p_1$ inserts $v$ into $R_1$ at line~\ref{setter1-m} of a {\sign(v)} before time $t$.
   \end{enumerate}
\end{lemma}

\begin{proof}~

        \textsc{Part (1)}: Suppose, for contradiction, that:
        (1) a $\valid(v)$ by a {\ct} reader returns $\false$, and
        (2) the writer $p_1$ inserts $v$ into $R_1$ at line~\ref{setter1-m} of a {\sign(v)} at some time $t$, but
        (3) the reader invokes this $\valid(v)$ after
        time $t$.
        Since $p_1$ is {\ct},
        by (2) and Observation~\ref{correctonlywrites1-mar},
		$v \in R_1$ at all times $t' \ge t$.
        Thus,
        by (3) and Lemma~\ref{testreturn1-mar},
        this $\valid(v)$ returns $\true$
        --- a contradiction to (1).
 
        \textsc{Part (2)}: Let $\valid(v)$ be an operation by a {\ct} reader $p_k$ that returns $\true$ at time~$t$.
	This implies $p_k$ finds that the condition $|\set_1| \ge n-f$ holds in line~\ref{return1-m} of $\valid(v)$ by time $t$.
	Since $n \ge 3f+1$, $|\set_1| \ge 2f+1$ at time $t$.
	Since there are at most $f$ faulty processes,
		there are at least $f+1$ {\ct} processes in $\set_1$ of $p_k$ at time $t$.
	By Lemma~\ref{ri1beforeset1-mar},
		there are at least $f+1$ {\ct} processes $p_j$ that have $v\in R_j$ at all times $t' \ge t$ ($\star$).

	Let $t^*$ be the earliest time when a {\ct} process $p_i$ has $v \in R_i$, say this is process $p_a$.
	Then by~($\star$), $t^* \le t$.
	Since $R_a$ is initially $\emptyset$, $p_a$ inserts $v$ into $R_a$ at time $t^*$.
	There are two cases:
	\begin{itemize}
		\item Case 1: $p_a$ inserts $v$ into $R_a$ at time $t^*$ at line~\ref{setter1-m} of $\sign(v)$.
			This implies $p_a = p_1$, i.e., $p_a$ is the writer.
			Since the writer $p_1$ is {\ct} and $t^* \le t$,
					$p_1$ inserts $v$ into $R_1$ at line~\ref{setter1-m} of a {\sign(v)} before time $t$.

		\item Case 2: $p_a$ inserts $v$ into $R_a$ at time $t^*$ at line~\ref{follow-m} of $\fresh()$.
			This implies $p_a$ finds that the condition $v \in r_1$ or $|\{ r_i|~ v \in r_i\}| \ge f+1$ holds  in line~\ref{followcondition-m},
				say at time $t_a < t^*$.
			There are two cases:
			\begin{itemize}
			
				\item Case 2.1: $p_a$ finds that the condition $v \in r_1$ of line~\ref{followcondition-m} holds at time $t_a$.
					This implies that $v\in R_1 $ when $p_a$ reads $R_1$ at line~\ref{sets-m} before $t_a$.
					Since $t_a < t^*$, $v \in R_1$ before time $t^*$.
					Since the writer $p_1$ is {\ct},
						this contradicts that $t^*$ is the earliest time when a {\ct} process $p_i$ has $v\in R_i$.
					So this case is impossible.

				\item Case 2.2: $p_a$ finds that the condition $|\{ r_i|~ v \in r_i\}| \ge f+1$ holds in line~\ref{followcondition-m} at time $t_a$.
					Since there are at most $f$ faulty processes,
						there is at least one {\ct} process $p_k$ such that $v\in r_k$ at time~$t_a$. 
					This implies $p_a$ reads $v\in R_k$ at line~\ref{sets-m} before time $t_a$.
					Since $t_a < t^*$, $v \in R_k$ before time~$t^*$.
					Since $p_k$ is {\ct}, 
						this contradicts that $t^*$ is the earliest time when a {\ct} process $p_i$ has $v\in R_i$.
					So this case is impossible.	
			\end{itemize}
	\end{itemize}
\end{proof}

        By Lemma~\ref{settest-mar}(1):

\begin{corollary}\label{VerifyFalse}
    In $H$, if a $\valid(v)$ by a {\ct} reader returns $\false$,
		then there is no $\sign(v)$ that precedes this $\valid(v)$ and returns $\success$.
\end{corollary}

        By Lemma~\ref{settest-mar}(2):

\begin{corollary}\label{VerifyTrue}
    In $H$, if a $\valid(v)$ by a {\ct} reader returns $\true$, 
      then there is a $\sign(v)$ that precedes or is concurrent with this $\valid(v)$ and returns $\success$.
\end{corollary}

\begin{lemma}
    \label{intersectsintervalp-mar-xing}
     In $H$, if there is a $\sign(v)$ that returns $\success$,
        then in the first $\sign(v)$\footnote{The \emph{first} $\sign(v)$ is well-defined because we make the standard assumption that each process applies operations sequentially.
            So when the writer is {\ct} in $H$, $\sign(-)$ and $\Set(-)$ operations by the writer are sequential in $H$.}
            that returns $\success$, 
        the writer $p_1$ inserts $v$ into $R_1$ at line~\ref{setter1-m} at time $t$ such that $t\in (t^v_0, t^v_1)$.   
\end{lemma}

\begin{proof}
    Suppose there is a $\sign(v)$ that returns $\success$ in $H$.
    Let $\sign(v)'$ be the first $\sign(v)$ that returns $\success$ in $H$.
    Let $t$ be the time when the writer $p_1$ inserts $v$ into $R_1$ at line~\ref{setter1-m} of $\sign(v)'$.
    Note that $t$ is the first time when the writer $p_1$ inserts $v$ into $R_1$ at line~\ref{setter1-m} in $H$.
We need to show that $t > t^v_0$ and $t < t^v_1$.

We first prove that $t > t^v_0$.
    There are two cases:
    \begin{compactitem}
        \item Case 1: no $\valid(v)$ by a {\ct} reader returns $\false$ in $H$.
            Then by Definition~\ref{t0t1-mar}, $t^v_0 = 0$, and so $t > t^v_0$.
            
         \item Case 2: some $\valid(v)$ by a {\ct} reader returns $\false$ in $H$.
            By Lemma~\ref{settest-mar}(1),
            $t$ is greater than the invocation time of this $\valid(v)$.
            So by Definition~\ref{t0t1-mar}, $t>t^v_0$.
     \end{compactitem}
    So in both cases, $t>t^v_0$.
    
We now prove that $t < t^v_1$.
    There are two cases:
        \begin{compactitem}

        \item Case 1: no $\valid(v)$ by a {\ct} reader returns $\true$ in $H$.
            Then by Definition~\ref{t0t1-mar}, $t^v_1 = \infty$, and so $t < t^v_1$.
            
         \item Case 2: some $\valid(v)$ by a {\ct} reader returns $\true$ in $H$.
            By Lemma~\ref{settest-mar}(2),
            $t$ is smaller than the response time of this $\valid(v)$.
            So by Definition~\ref{t0t1-mar}, $t < t^v_1$.
     \end{compactitem}
    So in both cases, $t < t^v_1$.
    
\end{proof}

The above Lemma~\ref{writesignlinearization-mar}, Corollaries~\ref{VerifyFalse} and~\ref{VerifyTrue},
    and Lemma~\ref{intersectsintervalp-mar-xing}, are about history $H$.
Since the writer is correct in $H$ and $H'=H|{\hct}$,
    by Observations~\ref{correctop-mar} and~\ref{sameop-mar},
    they also hold for the history $H'$, as~stated~below:

\begin{corollary}\label{register-mar-1}
    In $H'$, a $\sign(v)$ returns $\success$ if and only if there is a $\Set(v)$ that precedes the $\sign(v)$.
\end{corollary}

\begin{corollary}\label{VerifyFalse-xing}
    In $H'$, if a $\valid(v)$ by a {\ct} reader returns $\false$,
		then there is no $\sign(v)$ that precedes this $\valid(v)$ and returns $\success$.
\end{corollary}

\begin{corollary}\label{VerifyTrue-xing}
    In $H'$, if a $\valid(v)$ by a {\ct} reader returns $\true$, 
      then there is a $\sign(v)$ that precedes or is concurrent with this $\valid(v)$ and returns $\success$.
\end{corollary}

\begin{corollary}
    \label{intersectsintervalp-xing}
     In $H'$, if there is a $\sign(v)$ that returns $\success$,
        then in the first $\sign(v)$ operation that returns $\success$, 
        the writer $p_1$ inserts $v$ into $R_1$ at line~\ref{setter1-m} at time $t$ such that $t\in (t^v_0,t^v_1)$.   
\end{corollary}

We now prove that the history $H'$ is linearizable with respect to a \mar.
First, we define the \emph{linearization point} of each operation in $H'$ as follows.
\begin{definition}~\label{testsetlinearization-mar}
Let $o$ be an operation in  $H'$.
    \begin{itemize}
        \item  If $o$ is a $\Set(-)$ operation,
        then the linearization point of $o$ is the time when the writer executes line~\ref{r-m} of $o$. 

           \item  If $o$ is a $\Test$ operation,
        then the linearization point of $o$ is the time when the reader executes line~\ref{read-m} of $o$. 

	\item  If $o$ is a $\valid(-)$ operation,
        then: \begin{compactitem}
          \item If $o$ returns $\true$,
	   the linearization point of $o$ is the time of the response step of $o$.
            \item If $o$ returns $\false$,
	   the linearization point of $o$ is the time of the invocation step of~$o$.
        \end{compactitem}

         \item If $o$ is a $\sign(v)$ operation for some value $v$, then:
            \begin{compactitem}
                \item If $o$ returns $\success$,
                    the linearization point of $o$ is the time when the writer $p_1$ inserts $v$ into $R_1$ at line~\ref{setter1-m} in $o$.

                \item If $o$ returns $\fail$, the linearization point of $o$ is the time of the response step of $o$.
            \end{compactitem}
	
\end{itemize}
\end{definition}

Note that in Definition~\ref{testsetlinearization-mar} the linearization point of every operation $o$ is between the invocation and response time of $o$. 
Thus, the following holds.
\begin{observation}\label{precedes-mar}
    If an operation $o$ precedes an operation $o'$ in $H'$,
        then the linearization point of $o$ is before the linearization point of $o'$.
\end{observation}

By Definition~\ref{testsetlinearization-mar} and Corollary~\ref{intersectsintervalp-xing}:
\begin{observation}\label{signl-mar}
	The linearization point $t$ of the \emph{first} $\sign(v)$ operation that returns $\success$~in~$H'$ is such that $t\in (t^v_0,t^v_1)$.
\end{observation}

We now use the linearization points of all the operations in $H'$ to define the following linearization $L$ of $H'$:

\begin{definition}\label{hlp-mar}
	Let $L$ be a \emph{sequence} of operations such that:
	\begin{enumerate}
		\item\label{uno} An operation $o$ is in $L$ if and only if it is in $H'$.
		\item\label{due} An operation $o$ precedes an operation $o'$ in $L$ if and only if
     the linearization point of~$o$ is before the linearization point of $o'$ in $H'$.
	\end{enumerate}
\end{definition}

By Definition~\ref{hlp-mar}(\ref{uno}), $L$ is a linearization of $H'$.
By Observation~\ref{precedes-mar} and Definition~\ref{hlp-mar}(\ref{due}), $L$ respects
    the precedence relation between the operations of $H'$. More precisely:
\begin{observation}\label{precedes-l-mar}
    If an operation $o$ precedes an operation $o'$ in $H'$,
        then $o$ precedes $o'$ in $L$.
\end{observation}

In the following theorems, we prove that the linearization $L$ of $H'$ conforms to the sequential specification of a SWMR \mar.

\begin{theorem}\label{rw-mar}
    In $L$,
        if a $\Test$ returns a value $v$, then:
        \begin{compactitem}
            \item either some $\Set(v)$ precedes it and this $\Set(v)$ is the last $\Set(-)$ that precedes it,
            
            \item or $v = v_0$ (the initial value of the register) and no $\Set(-)$ precedes it.
        \end{compactitem}
       
\end{theorem}

\begin{proof}
    Suppose a $\Test$ returns a value $v$ in $L$.
    By the Definition~\ref{hlp-mar}(1) of $L$,
        this $\Test$ is also in~$H'$.
    By Corollary~\ref{register-mar-00}, there are two cases:
    \begin{compactitem}
        \item Case 1: $v$ is the last value that the writer writes into $R^*$ at line~\ref{r-m} of a $\Set(v)$
            before the reader reads $v$ from $R^*$ at line~\ref{read-m} of the $\Test$ in~$H'$.
            So by Definition~\ref{testsetlinearization-mar},
                the linearization point $t$ of this $\Set(v)$ is before the linearization point $t'$ of the $\Test$, and there is no other $\Set(-)$ with a linearization point between $t$ and $t'$.     
             Thus, by the Definition~\ref{hlp-mar}(2) of $L$,
                 this $\Set(v)$ precedes the $\Test$ in $L$, 
                 and it is the last $\Set(-)$ that precedes this $\Test$ in $L$.
        \item Case 2: $v = v_0$ and the writer never writes into $R^*$ at line~\ref{r-m} of a $\Set(-)$ before the reader reads $v$ from $R^*$ at line~\ref{read-m} of the $\Test$ in $H'$. 
        By Definition~\ref{testsetlinearization-mar},
            there is no $\Set(-)$ with a linearization point before the linearization point of this $\Test$.
        So by the Definition~\ref{hlp-mar}(2) of $L$,
            no $\Set(-)$ precedes this $\Test$ in $L$.        
    \end{compactitem}
    Therefore, in $L$,
        either a $\Set(v)$ precedes the $\Test$ and it is the last $\Set(-)$ that precedes this $\Test$,
        or $v = v_0$ and no $\Set(-)$ precedes the $\Test$.
\end{proof}

\begin{theorem}\label{ws-mar}
    In $L$, 
	 a $\sign(v)$ returns $\success$ if and only if there is a $\Set(v)$ that precedes this $\sign(v)$.
\end{theorem}

\begin{proof}
We prove each direction separately:
    \begin{compactitem}[]

    \item $[\Rightarrow]$ Suppose a $\sign(v)$ returns $\success$ in $L$.
    By Definition~\ref{hlp-mar}(1),
        this $\sign(v)$ is also in $H'$.
    By Corollary~\ref{register-mar-1}, there is a $\Set(v)$ that precedes this $\sign(v)$ in $H'$.
    Thus, by Observation \ref{precedes-l-mar},  
        there is a $\Set(v)$ that precedes this $\sign(v)$ in $L$. 
        
    \item $[\Leftarrow]$
    Suppose a $\Set(v)$ precedes a $\sign(v)$ in $L$.
    By Definition~\ref{hlp-mar},
        (a) these $\Set(v)$ and $\sign(v)$ operations are also in $H'$, and
        (b) the linearization point of the $\Set(v)$ is before the linearization point of the $\sign(v)$ in $H'$.
    Since the writer $p_1$ is {\ct},
        the $\Set(v)$ and the $\sign(v)$ are sequential in $H'$.
    Thus, since the linearization point of every operation is between the invocation and response time of the operation, and the linearization point of the $\Set(v)$ is before the linearization point of the $\sign(v)$ in $H'$,
        $\Set(v)$ precedes $\sign(v)$ in $H'$.       
    By Corollary~\ref{register-mar-1}, the $\sign(v)$ returns $\success$ in $H'$.
                Therefore it also returns $\success$ in $L$.
    \end{compactitem}
\end{proof}

\begin{lemma}\label{lsv1-mar}
        In $L$,
		if a $\valid(v)$ returns~$\false$,
            then there is no $\sign(v)$ that returns $\success$ and precedes this $\valid(v)$.
\end{lemma}

\begin{proof}
	Assume for contradiction that,
		in $L$,
		a $\valid(v)$ returns $\false$ and there is a $\sign(v)$ that returns $\success$ and precedes this $\valid(v)$.
        By Definition~\ref{hlp-mar},
        these $\valid(v)$ and $\sign(v)$ operations are also in $H'$,
        and the linearization point $t$ of the $\sign(v)$ is before the linearization point $t'$ of the $\valid(v)$ in $H'$.
        So $t < t'$.
        Let $\sign(v)'$ denote the \emph{first} $\sign(v)$ that returns $\success$ in $H'$ 
            and let $t^*$ be the linearization point of $\sign(v)'$.
        By Observation~\ref{precedes-mar},
            $t^* \le t$.
        So $t^* < t'$.
        By Observation~\ref{signl-mar}, $t^*\in(t^v_0,t^v_1)$.
        Thus, $t^v_0 < t^* < t' $ ($\star$).
        Since the $\valid(v)$ returns $\false$:
            (a) by Definition~\ref{testsetlinearization-mar},
            $t'$ is the time of the invocation step of $\valid(v)$ in $H'$, and
            (b) since $H'=H|{\hct}$, by Observation~\ref{sameip-mar},
            $t' \le  t^v_0$ --- a contradiction to ($\star$).
   \end{proof}

\begin{lemma}\label{lsv2-mar}
	In $L$, 
	if a $\valid(v)$ returns $\true$,
	then there is a $\sign(v)$ that returns $\success$ and precedes this $\valid(v)$.
\end{lemma}

\begin{proof}
	Suppose in $L$, a $\valid(v)$ returns $\true$.
        By the Definition~\ref{hlp-mar}(1) of $L$,
            this $\valid(v)$ is also in $H'$.
        By Definition~\ref{testsetlinearization-mar},
            the linearization point $t$ of $\valid(v)$ is the time of the response step of $\valid(v)$ in $H'$.
        Since $H'=H|{\hct}$,
            by Observation~\ref{sameip-mar},
            $t \ge t^v_1$ in $H'$~($\star$).  

        Since the $\valid(v)$ returns $\true$ in $H'$,
            by Corollary~\ref{VerifyTrue-xing},
            there is a $\sign(v)$ operation that returns $\success$ in $H'$.
        Let $\sign(v)'$ be the first $\sign(v)$ operation that returns $\success$ in $H'$.
        By Observation~\ref{signl-mar},
            the linearization point $t'$ of $\sign(v)'$ is in $(t^v_0,t^v_1)$.
        By ($\star$),
            $t' < t$. 
        Thus, by the Definition~\ref{hlp-mar}(2) of $L$, $\sign(v)'$ precedes this $\valid(v)$ in~$L$.
\end{proof}

By Lemmas~\ref{lsv1-mar} and~\ref{lsv2-mar},
    we have the following theorem.

\begin{theorem}\label{sv-mar}
    In $L$, 
	a $\valid(v)$ returns $\true$ if and only if
	there is a $\sign(v)$ that returns $\success$ and precedes this $\valid(v)$.
\end{theorem}

\begin{theorem}\label{linearthe-mar-case1}
    If the writer is correct in $H$, the history $H$ is Byzantine linearizable with respect to a SWMR {\mar}.
\end{theorem}

\begin{proof}
Assume that the writer $p_1$ is correct in $H$.
    To prove that $H$ is
        Byzantine linearizable with respect to a SWMR {\mar} (Definition~\ref{def-mar}),
    we must prove that there is a history $H'$ such that:
     \begin{compactenum}[(1)]
        \item $H'|{\hct} = H|{\hct}$.
        \item $H'$ is linearizable with respect to a SWMR {\mar}. 
    \end{compactenum}
    
     Let $H'= H|{\hct}$.
        Clearly, $H'$ satisfies~(1).   
        
    Let $L$ be the linearization of $H'$ given in Definition~\ref{hlp-mar}.
    \begin{compactenum}[(a)]
        \item By Observation~\ref{precedes-l-mar}, $L$ respects the precedence relation between the operations of $H'$.
        \item By Theorems~\ref{rw-mar},~\ref{ws-mar}, and~\ref{sv-mar}, $L$ conforms to the sequential specification of a SWMR {\mar}.
        \end{compactenum}
    Thus, $H'$ also satisfies~(2).
\end{proof}

\subsubsection*{Case 2: the writer $p_1$ is not {\ct} in $H$.}\label{case-writer-not-correct}

Recall that
$H$ is an arbitrary history of the implementation Algorithm~\ref{code-mar}, and
    {\hct} is the set of processes that are correct in $H$.
To show that $H$ is Byzantine linearizable with respect to a SWMR {\mar},
    we must prove that there is a history $H'$ such that (1) $H'|{\hct} =H|{\hct}$ and 
(2)~$H'$ is linearizable with respect to a SWMR \mar.

We construct $H'$ from $H$ as follows.
We start from $H|{\hct}$.
    Since the writer is faulty in $H$,
    $H|{\hct}$ \emph{does not contain any operation by the writer};
    in contrast, $H|{\hct}$ contains all the $\Test$ and $\valid(-)$ operations by the \emph{correct} processes in $H$.
So to ``justify'' the responses of these $\Test$ and $\valid(-)$ operations,
    we must add some $\Set(-)$ and $\sign(-)$ operations by the writer to $H|{\hct}$.
And we must add them such that
    resulting history conforms to the sequential specification of a {\mar}. 
To do so, for every $\valid(v)$ that returns $\true$ in $H|{\hct}$,
    we add a $\sign(v)$ that returns $\success$;
    this $\sign(v)$ is added after the \emph{invocation} of the \emph{last} $\valid(v)$ that returns $\false$ and
        before the \emph{response} of the \emph{first} $\valid(v)$ that returns $\true$.
Then, for every $\Test$ operation that returns $v$ in $H|{\hct}$, and also
      for every $\sign(v)$ operation that we added,
      we also add a $\Set(v)$ operation ``just before'' this operation in $H|{\hct}$.
    We describe this construction more precisely below.

\begin{definition}\label{steps-mar-2}
Let $H'$ be the history constructed from $H$ as follows.
\begin{enumerate}[Step 1:]
	\item\label{s1} $H' = H|{\hct}$.

    \item\label{addsign} 
           For every value $v$ such that a $\valid(v)$ returns $\true$ in the history $H'$ constructed in Step 1,
		we add to $H'$ a $\sign(v)$ that returns $\success$ such that
            the $[invocation,response]$ interval of this $\sign(v)$ is within the interval $(t^v_0, t^v_1)$.

 \item\label{addset} 
            For each operation $o$ such that $o$ is a $\Test$ that returns some value $v$, or a $\sign(v)$ operation that returns $\success$ in the history $H'$ constructed in Step 2,
		  we add  to $H'$ a $\Set(v)$ such that
            the $[invocation,response]$ interval of the $\Set(v)$
            \emph{immediately precedes} $o$.\footnote{An operation $o'$\emph{immediately precedes} an operation $o$
        if and only if $o'$ precedes $o$, and there is no other operation $o^*$
	such that $o'$ precedes $o^*$ and $o^*$ precedes~$o$.}

\end{enumerate}

\end{definition}

    Note that in the above construction we must ensure that all the $\Set(-)$ and $\sign(-)$ operations that we add to $H'$ are sequential (because they are ``executed'' by the same process, namely the writer). We can do so by ``shrinking'' the $[invocation,response]$ interval of every added $\Set(-)$ and $\sign(-)$ as much as necessary to achieve this.

\begin{observation}\label{cc-mar-2}
    $H'|{\hct} = H|{\hct}$.
\end{observation}

To prove that $H$ is Byzantine linearizable, it now suffices to show that $H'$ is linearizable with respect to a SWMR \mar.
\begin{observation}\label{validfromcorrect-mar-2}
   A $\Test$ or $\valid(-)$ operation is in $H'$ if and only if it is in $H|{\hct}$.
\end{observation}

\begin{observation}\label{sameip-mar-2}
        In $H'$, for any value $v$,
	\begin{itemize}
		\item The max invocation time of any $\valid(v)$ operation that returns $\false$ is $t^v_0$.
		\item The min response time of any $\valid(v)$ operation that returns~$\true$ is $t^v_1$.
	\end{itemize}
\end{observation}

\begin{observation}
   A $\sign(-)$ operation is in $H'$ if and only if it is added in Step 2 of the construction of $H'$ (Definition~\ref{steps-mar-2}).
\end{observation}

\begin{observation}\label{validfromcorrect-mar-3}
   Every $\sign(-)$ operation in $H'$ returns $\success$.
\end{observation}

\begin{observation}\label{allsignadd-mar-2}
    The $[invocation,response]$ interval of every $\sign(v)$ operation in $H'$
         is in the interval $(t^v_0,t^v_1)$.
\end{observation}

By Step 3 of the construction of $H'$, we have:

\begin{observation}
    \label{register-mar-c2}
      For every $\sign(v)$ operation in $H'$, there is a $\Set(v)$ that precedes this $\sign(v)$.
\end{observation}

We now prove that the constructed history $H'$ is linearizable with respect to a \mar.
First, we define the \emph{linearization point} of each operation in $H'$ as follows.
\begin{definition}~\label{testsetlinearization-mar-2}
Let $o$ be an operation in $H'$.
    \begin{itemize}

	\item If $o$ is a $\Set(-)$ operation,
	   the linearization point of $o$ is the time of the response step of~$o$.
       \item  If $o$ is a $\Test$ operation,
	   the linearization point of $o$ is the time of the invocation step of $o$.
       
	\item  If $o$ is a $\valid(-)$ operation,
        then: \begin{compactitem}
          \item If $o$ returns $\true$,
	   the linearization point of $o$ is the time of the response step of $o$.
            \item If $o$ returns $\false$,
	   the linearization point of $o$ is the time of the invocation step of~$o$.
        \end{compactitem}

          \item If $o$ is a $\sign(-)$ operation, 
		the linearization point of $o$ is any time $t$ in the $[invocation,response]$ interval of $o$.
\end{itemize}
\end{definition}
\begin{observation}\label{precedes-mar-2}
    If an operation $o$ precedes an operation $o'$ in $H'$,
        then the linearization point of $o$ is before the linearization point of $o'$.
\end{observation}

By Observation~\ref{allsignadd-mar-2} and Definition~\ref{testsetlinearization-mar-2}, we have the following observation.
\begin{observation}\label{signl-mar-2}
     The linearization point of every $\sign(v)$ operation in $H'$
        is in the interval $(t^v_0,t^v_1)$.
\end{observation}

We now use the linearization points of all the operations in $H'$ to define the following linearization $L$ of $H'$:
\begin{definition}\label{hlp-mar-2}
	Let $L$ be a \emph{sequence} of operations such that:
	\begin{enumerate}
		\item\label{uno} An operation $o$ is in $L$ if and only if it is in $H'$.
		\item\label{due} An operation $o$ precedes an operation $o'$ in $L$ if and only if
     the linearization point of~$o$ is before the linearization point of $o'$ in $H'$.
	\end{enumerate}
\end{definition}

By Observation~\ref{precedes-mar-2} and Definition~\ref{hlp-mar-2}, 
    $L$ respects the precedence relation between the operations of $H'$.
More precisely:
\begin{observation}\label{precedes-l-mar-2}
    If an operation $o$ precedes an operation $o'$ in $H'$,
        then $o$ precedes $o'$ in $L$.
\end{observation}

In the following lemmas and theorems we prove that the linearization $L$ of $H'$ conforms to the sequential specification of a SWMR \mar.

\begin{theorem}\label{rw-mar-2}
    In $L$,
        if a $\Test$ operation returns a value $v$, then:
        \begin{compactitem}
          \item either some $\Set(v)$ operation precedes it and this $\Set(v)$ operation is the last $\Set(-)$ operation that precedes it,
            
            \item or $v = v_0$ (the initial value of the register) and no $\Set(-)$ operation {\precedes} it.
        \end{compactitem}
       
\end{theorem}
\begin{proof}
    Suppose a $\Test$ operation returns a value $v$ in $L$.
    By Definition~\ref{hlp-mar-2}(1),
        this $\Test$ operation is also in $H'$.
    By Step 3 of Definition~\ref{steps-mar-2},
    there is a $\Set(v)$ operation that immediately precedes
    the $\Test$ operation in $H'$.
    So by Definition~\ref{testsetlinearization-mar-2},
            the linearization point $t$ of this $\Set(v)$ operation is before the linearization point $t'$ of the $\Test$ operation, and 
            there is no other $\Set(-)$ operation with a linearization point between $t$ and $t'$.     
    Thus, by Definition~\ref{hlp-mar-2}(2),
        this $\Set(v)$ operation precedes the $\Test$ operation and it is the last $\Set(-)$ operation that precedes the $\Test$ operation in $L$.
\end{proof}

\begin{theorem}\label{ws-mar-2}
    In $L$, 
	 a $\sign(v)$ operation returns $\success$ if only if there is a $\Set(v)$ operation that {\precedes} this $\sign(v)$ operation.
\end{theorem}

\begin{proof}
We prove each direction separately:
    \begin{compactitem}[]
    
    \item $[\Rightarrow]$ Suppose a $\sign(v)$ operation returns $\success$ in $L$.
    By Definition~\ref{hlp-mar-2}(1),
        this $\sign(v)$ operation is also in $H'$.
    By Observation~\ref{register-mar-c2}, there is a $\Set(v)$ operation that precedes this $\sign(v)$ operation in $H'$.
    Thus, by Observation \ref{precedes-l-mar-2},  
       this $\Set(v)$ operation precedes the $\sign(v)$ operation in $L$. 

     \item $[\Leftarrow]$
    Suppose a $\Set(v)$ operation precedes a $\sign(v)$ operation in $L$.
    By Definition~\ref{hlp-mar-2},
        these $\Set(v)$ and $\sign(v)$ operations are also in $H'$.
    So by Observation~\ref{validfromcorrect-mar-3},
        the $\sign(v)$ operation returns $\success$ in $H'$.
                Therefore it also returns $\success$ in $L$.
    \end{compactitem}   
\end{proof}

\begin{lemma}\label{vs1}
        In $L$,
            if a $\valid(v)$ operation returns~$\false$,
		then there is no $\sign(v)$ operation that returns $\success$ and {\precedes} the $\valid(v)$ operation.
\end{lemma}
\begin{proof}
	Assume for contradiction that, 
		in $L$,
            a $\valid(v)$ operation returns $\false$ and
		there is a $\sign(v)$ operation that returns $\success$ and {\precedes} the $\valid(v)$ operation.
        By Definition~\ref{hlp-mar-2}, 
            the $\sign(v)$ and $\valid(v)$ operations are also in $H'$, and 
            the linearization point $t$ of the $\sign(v)$ operation is before the linearization point $t'$ of the $\valid(v)$ operation in $H'$.
        So $t < t'$.
        By Observation~\ref{signl-mar-2}, $t\in (t^v_0,t^v_1)$ in $H'$ and so $t^v_0 <t <t' $ ($\star$).
        Since this $\valid(v)$ operation returns $\false$:
            (a) by Definition~\ref{testsetlinearization-mar-2}, 
            $t'$ is the time of the invocation step of $\valid(v)$~in~$H'$,
            and
            (b)~by Observation~\ref{sameip-mar-2},
            $t' \le t^v_0$ --- a contradiction to ($\star$).
\end{proof}

\begin{lemma}\label{vs2}
	In $L$, 
	if a $\valid(v)$ operation returns $\true$,
	then there is a $\sign(v)$ operation that returns $\success$ and {\precedes} the $\valid(v)$ operation.
\end{lemma}

\begin{proof}
	Suppose in $L$, a $\valid(v)$ operation returns $\true$.
        By the Definition~\ref{hlp-mar-2}(1) of $L$,
            this $\valid(v)$ operation is also in $H'$.
         Since $\valid(v)$ operation returns $\true$,
            by Definition~\ref{testsetlinearization-mar-2},
            its linearization point $t$ is the time of its response step in $H'$.
         By Observation~\ref{sameip-mar-2},
            $t \ge t^v_1$ ($\star$).

        Since the $\valid(v)$ operation returns $\true$ in $H'$,
            by Step 3 of Definition~\ref{steps-mar-2},
             there is a $\sign(v)$ operation that returns $\success$ in $H'$.
        By Observation~\ref{signl-mar-2}, 
            the linearization point $t'$ of this $\sign(v)$ operation is in the interval $(t^v_0,t^v_1)$.
        So $t' < t^v_1$.
        By ($\star$),
            $t' < t$. 
        Thus, by the Definition~\ref{hlp-mar-2}(2) of $L$, this $\sign(v)$ operation that returns $\success$ {\precedes} the $\valid(v)$ operation in $L$.
\end{proof}

By Lemma~\ref{vs1} and Lemma~\ref{vs2}, we have the following theorem.
\begin{theorem}\label{sv-mar-2}
    In $L$, 
	a $\valid(v)$ operation returns $\true$ if and only if
	there is a $\sign(v)$ operation that returns $\success$ and {\precedes} this $\valid(v)$ operation.
\end{theorem}

\begin{theorem}\label{linearthe-mar-case2}
    If the writer is not correct in $H$, the history $H$ is Byzantine linearizable with respect to a SWMR {\mar}.
\end{theorem}

\begin{proof}
Assume that the writer $p_1$ is not correct in $H$.
    To prove that $H$ is
        Byzantine linearizable with respect to a SWMR {\mar} (Definition~\ref{def-mar}),
    we must prove that there is a history $H'$ such that:
     \begin{compactenum}[(1)]
        \item $H'|{\hct} = H|{\hct}$.
        \item $H'$ is linearizable with respect to a SWMR {\mar}. 
    \end{compactenum}
    
    Let $H'$ be the history given by Definition~\ref{steps-mar-2}.
    By Observation~\ref{cc-mar-2}, $H'$ satisfies (1).
    
    Let $L$ be the linearization of $H'$ given in Definition~\ref{hlp-mar-2}.
    \begin{compactenum}[(a)]
        \item By Observation~\ref{precedes-l-mar-2}, $L$ respects the precedence relation between the operations of $H'$.
        \item By Theorems~\ref{rw-mar-2},~\ref{ws-mar-2}, and~\ref{sv-mar-2}, $L$ conforms to the sequential specification of a SWMR {\mar}.
        \end{compactenum}
    Thus, $H'$ also satisfies~(2).
\end{proof}

By Theorems~\ref{linearthe-mar-case1} and~\ref{linearthe-mar-case2}, the arbitrary history $H$ of Algorithm~\ref{code-mar} is Byzantine linearizable with respect to a {\mar} (independent of whether the writer $p_1$ is correct or not in $H$):

\begin{restatable}{theorem}{verifybl}\label{linearthe-mar}
    [\textsc{Byzantine linerizability}] The history $H$ is Byzantine linearizable with respect to a SWMR {\mar}.
\end{restatable}

By Theorem~\ref{termination-mar} and Theorem~\ref{linearthe-mar},
    we have the following:

\verifycorrect*

\section{Correctness of the Authenticated Register Implementation}\label{a-tar}
In the following, we prove that Algorithm~\ref{code-tar} is a
    correct implementation of a SWMR {\tar} for systems 
    with $n >3f$ processes, $f$ of which can be Byzantine.
To do so, in the following, we consider an arbitrary (infinite) history $H$ of Algorithm~\ref{code-tar} (where $n>3f$), and prove:
    \begin{compactitem}
        \item \textsc{[Byzantine linearizability]} The history $H$ is Byzantine linearizable with respect to a SWMR {\tar} (Section~\ref{linear-tar}).
        \item \textsc{[Termination]} All processes that are correct in $H$ complete all their operations (Section \ref{wait-free-tar}).
    \end{compactitem}

\medskip
\noindent
In Algorithm~\ref{code-tar}, a process executes the $\valid(-)$ procedure in two cases:
(a) when it applies a $\valid$ operation, and 
(b) when it applies a $\Test$ operation, and it calls the $\valid(-)$ procedure \emph{inside} the $\Test()$ procedure (see line~\ref{verifyp-tar}).
Each one of these two types of $\valid(-)$ executions has an \emph{invocation} and a \emph{response}.

\medskip\noindent
\textbf{Notation.} 
\begin{itemize}
    \item A $\valid(-)$ execution performed to apply a $\valid$ operation is called a
    $\valid(-)$ operation.
    
   \item Note that the register $R_1$ stores a set of $\langle \ell,v \rangle$ tuples.
   For convenience, if the set stored in $R_1$ contains $\langle -,v \rangle$, we write $v \in R_1$,
    and say that ``$v$ is in $R_1$''.

\end{itemize}

\subsection{Termination}\label{wait-free-tar}
\begin{observation}\label{monock-tar}
	For every {\ct} reader $p_k\in \{p_2,\ldots,p_n \}$, 
		the value of $C_k$ is non-decreasing.
\end{observation}

\begin{observation}\label{correctonlywrites1-tar}
For every {\ct} process $p_i\in \{p_1,\ldots,p_n\}$,
	if $v \in R_i$ at time $t$, then $v \in R_i$ at all times $t' \ge t$.
\end{observation}

\begin{observation}\label{correctonlyv0-tar}
For every {\ct} process $p_i\in \{p_1,\ldots,p_n\}$,
	$v_0 \in R_i$ at all times.
\end{observation}

\begin{observation}\label{monoset1-tar}
	In every {\validpg{(-)}} by a {\ct} reader,
	 $|\set_1|$ is non-decreasing.
\end{observation}

\begin{lemma}\label{correctp1v-tar}
    	Suppose there is a time $t$ such that
		at least $f+1$ processes $p_i$ have $v \in R_i$ at time $t$.
		If the writer $p_1$ is {\ct}, 
		then $p_1$ has $v \in R_1$ at time $t$.
\end{lemma}
\begin{proof}
    Assume for contradiction that, (1) there is a time $t$ such that
		at least $f+1$ processes $p_i$ have $v \in R_i$ at time $t$;
     (2) the writer $p_1$ is {\ct} and $p_1$ has $v \not\in R_1$ at time $t$.
    Since $p_1$ is {\ct}, by Observation~\ref{correctonlyv0-tar},
        $v\ne v_0$ ($\star$).
    Since there are at most $f$ faulty processes,
        by~(1),
        there is at least one {\ct} process $p_j$ that has $v \in R_j$ at time $t$.
    Let $t^*$ be the earliest time when a {\ct} process $p_i$ has $v \in R_i$, say this is process $p_a$.
    Then $t^* \le t$.
    By (2) and Observation~\ref{correctonlywrites1-tar},
        $p_a\ne p_1$.
    So since $R_a$ is initially $\{v_0\}$, by $(\star)$, $p_a$ inserts $v$ into $R_a$ at time $t^*$.
    Thus, $p_a$ inserts $v$ into $R_a$ at time $t^*$ at line~\ref{follow-tar} of $\fresh()$.
	This implies $p_a$ finds that the condition $v \in r_1$ or $|\{ r_i~|~1 \le i \le n~\text{and}~ v \in r_i\}| \ge f+1$ holds at line~\ref{followcondition-tar} before $t^*$.
			There are two cases:
			\begin{itemize}
			
				\item Case 1: $p_a$ finds that the condition $v \in r_1$ of line~\ref{followcondition-tar} holds before $t^*$.
					This implies that $v\in R_1 $ when $p_a$ reads $R_1$ at line~\ref{readr1-tar} before  $t^*$.
				Thus, $v \in R_1$ before time $t^*$.
					Since the writer $p_1$ is {\ct},
						this contradicts that $t^*$ is the earliest time when a {\ct} process $p_i$ has $v\in R_i$.

				\item Case 2: $p_a$ finds that the condition $|\{ r_i~|~1 \le i \le n~\text{and}~ v \in r_i\}| \ge f+1$ holds at line~\ref{followcondition-tar} before $t^*$.
					Since there are at most $f$ faulty processes,
						there is at least one {\ct} process $p_k$ such that $v\in r_k$ before $t^*$. 
					This implies $p_a$ reads $v\in R_k$ at line~\ref{readr1-tar} or line~\ref{sets-tar} before $t^*$.
					Since $p_k$ is {\ct}, 
						this contradicts that $t^*$ is the earliest time when a {\ct} process $p_i$ has $v\in R_i$.
			\end{itemize}
	
\end{proof}

	\begin{lemma}\label{asker1-tar}
		Suppose there is a time $t$ such that
		$v\in R_1$ at all times $t' \ge t$, or at least $f+1$ processes $p_i$ have $v \in R_i$ at all times $t' \ge t$.
		Consider any iteration of the while loop of the $\fresh()$ procedure executed by a {\ct} process $p_j$.
		If $p_j$ inserts $p_k$ into $askers$ at line~\ref{askers-tar} at some time $t_a \ge t$, 
		then $p_j$ writes $\langle r_j , - \rangle$ with $v\in r_j$ into $R_{jk}$ at line~\ref{fresh1-tar}.
		\end{lemma} 
	
    \begin{proof}
		Suppose there is a time $t$ such that
		$v\in R_1$ at all times $t' \ge t$, or at least $f+1$ processes $p_i$ have $v \in R_i$ at all times $t' \ge t$ ($\star$).
        
		Consider any iteration of the while loop of the $\fresh()$ procedure executed by a {\ct} process $p_j$.
		Suppose $p_j$ inserts $p_k$ into $askers$ at line~\ref{askers-tar} at some time $t_a \ge t$ and so $askers \ne \emptyset$ at line~\ref{replyasker-tar}. 
        There are two cases:
        \begin{enumerate}
             \item $v \in R_1$ at all times $t' \ge t$.
             Thus, when $p_1$ reads $R_1$ after time $t_a \ge t$ at line~\ref{readr1-tar},
                    $v \in R_1$.
                    So $p_1$ inserts $v$ into $r_1$
                    at line~\ref{extract-r1-tar}.    
                There are two cases:
                \begin{itemize}
                    \item Case 1: $j\ne 1$.
                      Then $p_j$ finds that the condition $v \in r_1$ holds at line~\ref{followcondition-tar}.
                      So $p_j$ inserts $v$ into $R_j$ at line~\ref{follow-tar}.
                 Therefore, $p_j$ writes $\langle r_j , - \rangle$ with $v\in r_j$ into $R_{jk}$ at line~\ref{fresh1-tar}.
    
                 \item Case 2: $j = 1$.
                 So $p_1$ writes $\langle r_1 , - \rangle$ with $v\in r_1$ into $R_{1k}$ at line~\ref{fresh1-tar}.
                \end{itemize}
                
            \item At least $f+1$ processes $p_i$ have $v \in R_i$ at all times $t' \ge t$.
                There are two cases:
                \begin{itemize}
                    \item Case 1: $j\ne 1$. 
                    Then there are at least $f+1$ processes $p_i$ such that $p_j$ reads $R_i$ with $v\in R_i$ at line~\ref{readr1-tar} or line~\ref{sets-tar} after time $t_a \ge t$.
                    So, by ($\star$), $p_j$ finds that the condition $|\{ r_i~|~1 \le i \le n~\text{and}~ v \in r_i\}| \ge f+1$ holds at line~\ref{followcondition-tar}.
                      Then $p_j$ inserts $v$ into $R_j$ at line~\ref{follow-tar}.
                 Therefore, $p_j$ writes $\langle r_j , - \rangle$ with $v\in r_j$ into $R_{jk}$ at line~\ref{fresh1-tar}.
    
                 \item Case 2: $j = 1$.
                 Since at least $f+1$ processes $p_i$ have $v \in R_i$ at all times $t' \ge t$ and the writer $p_1$ is {\ct},
                    by Lemma~\ref{correctp1v-tar},
                    $v\in R_1$ at all times $t' \ge t$.
                    Thus, when $p_1$ reads $R_1$ after time $t_a \ge t$ at line~\ref{readr1-tar},
                    $v \in R_1$.
                    So $p_1$ inserts $v$ into $r_1$
                    at line~\ref{extract-r1-tar}.
                    Then $p_1$ writes $\langle r_1 , - \rangle$ with $v\in r_1$ into $R_{1k}$ at line~\ref{fresh1-tar}.
                \end{itemize}
        \end{enumerate}
       
        Therefore, in all possible cases,
            $p_j$ writes $\langle r_j , - \rangle$ with $v\in r_j$ into $R_{jk}$ at line~\ref{fresh1-tar}.
    \end{proof}

	\begin{lemma}\label{ri1beforeset1-tar}
        Consider any \validpg{(v)} by a {\ct} reader $p_k$.
		If a {\ct} process $p_j$ is in $\set_1$ of $p_k$ in this $\valid(v)$ at time~$t$,
			then $v \in R_j$ at all times $t' \ge t$.
		\end{lemma}
		\begin{proof}
            Consider any \validpg{(v)} by a {\ct} reader $p_k$.
			Suppose a {\ct} process $p_j$ is in $\set_1$ of $p_k$ in this $\valid(v)$ at time~$t$
			Consider the iteration of the while loop at line~\ref{whileloop-tar}
                of the $\valid(v)$ in which 
				$p_k$ inserts $p_j$ into $\set_1$ at line~\ref{set1-tar},
				say at time $t_k^1 \le t$.
			By lines~\ref{readri-tar} and \ref{check1-tar}, 
				$p_k$ reads $\langle r_j, - \rangle$ with $v\in r_j$ from $R_{jk}$ at line~\ref{readri-tar},
				say at time $t_k^0 < t_k^1$.
			Since $p_j$ is {\ct} and $R_{jk}$ is initialized to $\langle \emptyset, - \rangle$, 
				$p_j$ wrote $\langle r_j, - \rangle$ with $v\in r_j$ into $R_{jk}$ at line~\ref{fresh1-tar} by time $t_k^0$.
			By line~\ref{readrj-tar} and line~\ref{fresh1-tar},
				$v \in R_j$ by time $t_k^0$.
			Since $t_k^0 < t_k^1$ and $t_k^1 \le t$,
				$t_k^0 < t$.
			Since $p_j$ is {\ct},
				by Observation~\ref{correctonlywrites1-tar},
				$v \in R_j$ at all times $t' \ge t$.
		\end{proof}

\begin{lemma}\label{onecorrectset0-tar}
	Consider any \validpg{(v)} by a {\ct} reader $p_k$.
	For each iteration of the while loop at line~\ref{whileloop-tar} of this $\valid(v)$, 
		the following loop invariants hold at line~\ref{whileloop-tar}:
	\begin{enumerate}
             \item\label{inv0-tar} $\set_1$ and $\set_0$ are disjoint.
		\item\label{inv1-tar} $|\set_1| < n-f$ and $|\set_0| \le f$.
		\item\label{inv2-tar} If there are at least $f+1$ {\ct} processes in $\set_1$,
				there is no {\ct} process in $\set_0$.
	\end{enumerate}
\end{lemma}

\begin{proof}
	Consider any \validpg{(v)} by a {\ct} reader $p_k$.
    We now prove the invariants by induction on the number of iterations of the while loop at line~\ref{whileloop-tar} of this $\valid(v)$.
    \begin{itemize}
    \item Base Case: 
        Consider the first iteration of the while loop at line~\ref{whileloop-tar} of this $\valid(v)$.
        Since $p_k$ initializes $\set_1$ and $\set_0$ to $\emptyset$,
            (\ref{inv0-tar}) holds trivially.
        Furthermore,
            $p_k$ has $|\set_1| = 0$ and $|\set_0| = 0 $.
        Since $n \ge 3f+1$ and $f \ge 1$,
            $|\set_1| = 0 < n-f$ and $|\set_0|= 0 \le f$ and so (\ref{inv1-tar}) holds.
        Since $\set_0$ is empty,
        (\ref{inv2-tar}) holds trivially.
	
    \item Inductive Case:
        Consider any iteration $I$ of the while loop at line~\ref{whileloop-tar} of this $\valid(v)$.
        Assume that (\ref{inv0-tar}), (\ref{inv1-tar}), and (\ref{inv2-tar}) hold at the beginning of the iteration.
			
        If $p_k$ does not find that the condition of line~\ref{until-tar} holds in $I$,
            $p_k$ does not move to the next iteration of the while loop,
            and so (\ref{inv0-tar}), (\ref{inv1-tar}), and (\ref{inv2-tar}) trivially hold at the start of the ``next iteration'' (since it does not exist).
        Furthermore, if $\valid(v)$ returns at line~\ref{return1-tar} or line~\ref{return0-tar} in $I$,
            (\ref{inv0-tar}), (\ref{inv1-tar}), and (\ref{inv2-tar}) trivially hold
            at the start of the ``next iteration'' (since it does not exist).

        We now consider any iteration $I$ of the while loop at line~\ref{whileloop-tar} of $\valid(v)$
            in which $p_k$ finds the condition of line~\ref{until-tar} holds 
            and $\valid(v)$ does not return at line~\ref{return1-tar} or line~\ref{return0-tar}.
        We first show that (\ref{inv0-tar}) 
            remains true at the end of $I$.
        Since $p_k$ finds the condition of line~\ref{until-tar} holds,
            $p_k$ inserts a process, say $p_j$, into $\set_1$ or $\set_0$ in $I$.
        By line~\ref{until-tar},
            $p_j\not\in \set_1\cup\set_0$.
        Note that $p_j$ is the only process that $p_k$ inserts into $\set_1$ or $\set_0$ in $I$.
        So since $\set_1$ and $\set_0$ are disjoint at the beginning of $I$ and $p_j\not\in \set_1\cup\set_0$,
             $\set_1$ and $\set_0$ are disjoint at the end of $I$.
        Therefore (\ref{inv0-tar}) holds at the end of $I$.

        We now show that (\ref{inv1-tar}) remains true at the end of $I$.
        Since $\valid(v)$ does not return at line~\ref{return1-tar} or line~\ref{return0-tar},
            $p_k$ finds that $|\set_1| < n-f$ holds at line~\ref{return1-tar} and $|\set_0| \le f$ holds at line~\ref{return0-tar}.
        Since $p_k$ does not change $\set_1$ after line~\ref{return1-tar} and does not change $\set_0$ after line~\ref{return0-tar} in $I$,
                $|\set_1| < n-f$ and $|\set_0| \le f$ remain true at the end of $I$ and so (\ref{inv1-tar}) holds.

        We now show that (\ref{inv2-tar}) remains true at the end of $I$.
        There are two cases:
        \begin{enumerate}
            \item Case 1: $p_k$ executes line~\ref{set1-tar} in $I$.
                Then $p_k$ changes $\set_0$ to $\emptyset$ at line~\ref{empty0-tar}.
                Since $p_k$ does not change $\set_0$ after line~\ref{empty0-tar} in $I$,
                    $\set_0$ remains $\emptyset$ at the end of $I$.
                So (\ref{inv2-tar}) holds trivially at the end of $I$.

        \item Case 2: $p_k$ does not execute line~\ref{set1-tar} in $I$.
                So $p_k$ does not change $\set_1$ in $I$.
                By assumption,
                    $p_k$ finds the condition of line~\ref{until-tar} holds in $I$.
                    So $p_k$ executes line~\ref{set0-tar} in $I$.
                Let $p_a$ be the process that $p_k$ inserts into $\set_0$ at line~\ref{set0-tar}.
                There are two cases:
                \begin{itemize}
                \item Case 2.1: $p_a$ is faulty. 
                    Then the number of {\ct} processes in $\set_0$ does not change in $I$.
                    So, since $\set_1$ also does not change in $I$,
                    (\ref{inv2-tar}) remains true at the end of $I$.

                \item Case 2.2: $p_a$ is {\ct}. 
                There are two cases:

                \begin{itemize}
                \item Case 2.2.1: There are fewer than $f+1$ {\ct} processes in $\set_1$ at the beginning of $I$.
                    Since $p_k$ does not change $\set_1$ in $I$, 
                    (\ref{inv2-tar}) remains true at the end of $I$.

                \item Case 2.2.2: There are at least $f+1$ {\ct} processes in $\set_1$ at the beginning of $I$.
                    We now show that this case is impossible.

                    Let $p_b$ be the last process that $p_k$ inserted into $\set_1$ before the iteration $I$;
                    say $p_b$ was inserted into $\set_1$ at time $t_k^0$.
                    Since there are at least $f+1$ {\ct} processes in $\set_1$ at the beginning of $I$
                        and $p_b$ is the last process that $p_k$ inserted into $\set_1$ before $I$,
                        there are at least $f+1$ {\ct} processes in $\set_1$ of $p_k$ in $\valid(v)$ at time $t_k^0$.
                    By Lemma~\ref{ri1beforeset1-tar},
                        at least $f+1$ {\ct} processes $p_i$ have $v \in R_i$ at all times $t'\ge t_k^0$~($\star$).

                    Recall that $p_k$ inserts the correct process $p_a$ into $\set_0$ in the iteration $I$.
                    So in $I$, $p_k$ increments $C_k$ at line~\ref{ckplus-tar}, say at time  $t_k^1$.
                    Note that $t_k^1 > t_k^0$. 
                    Let $c^*$ be the value of $C_k$ after $p_k$ increments $C_k$ at time~$t_k^1$.
                    Since $C_k$ is initialized to 0,
                        by Observation~\ref{monock-tar}, $c^* \ge 1$.
                    By line~\ref{readri-tar}, line~\ref{until-tar}, and line~\ref{notv-tar},
                        $p_k$ reads $\langle r, c \rangle$ with $v\not\in r$ and $c \ge c^*$ from $R_{ak}$ at line~\ref{readri-tar} in $I$.
                    Since $c \ge c^* \ge 1$ and $R_{ak}$ is initialized to $\langle \emptyset,0\rangle$,
                        it must be that $p_a$ wrote $\langle r, c \rangle$ with $v\not\in r$ into $R_{ak}$ at line~\ref{fresh1-tar} ($\star\star$).

                    Consider the iteration of the while loop of the $\fresh()$ procedure in which $p_a$ writes $\langle r, c \rangle$ with $v\not\in r$ into $R_{ak}$ at line~\ref{fresh1-tar}.
                    Note that $c$ is the value that $p_a$ read from $C_k$ at line~\ref{collectck-tar} of this iteration;
                         say this read occurred at time $t_a^1$.
                    Since $c \ge c^*$,
                        by Observation~\ref{monock-tar},
                        $t_a^1 \ge t_k^1$.
                    Then $p_a$ inserts $p_k$ into $askers$ at line~\ref{askers-tar}, say at time $t_a^2 > t_a^1$.
                    Since $t_a^1 \ge t_k^1$ and $t_k^1 > t_k^0$,
                        $t_a^2 > t_k^0$.
                    So $p_a$ inserts $p_k$ into $askers$ at line~\ref{askers-tar} after $t_k^0$.
                    Thus, by ($\star$) and Lemma~\ref{asker1-tar},
                        $p_a$ writes $\langle r , c \rangle$ with $v\in r$ into $R_{ak}$ at line~\ref{fresh1-tar} in this iteration,
                        a contradiction to ($\star\star$). 
                    So this case is impossible.
                \end{itemize}
                \end{itemize}
				\end{enumerate}
				So in all the possible cases, we showed that  (\ref{inv0-tar}), (\ref{inv1-tar}), and (\ref{inv2-tar}) remain true at the end of the iteration.
			\end{itemize}
	\end{proof}

    \begin{lemma}\label{correctoutside-tar}
    Consider any \validpg{(v)} by a {\ct} reader $p_k$.
            Every time when $p_k$ executes line~\ref{whileloop-tar} of this $\valid(v)$, 
    	there is a {\ct} process $p_i \not \in \set_1\cup\set_0$.
    \end{lemma} 	
    \begin{proof}
        Consider any \validpg{(v)} by a {\ct} reader $p_k$.
        Suppose $p_k$ executes line~\ref{whileloop-tar} of this $\valid(v)$ at time $t$.
        Consider (the values of) $\set_0$ and $\set_1$ at time $t$.
        By Lemma~\ref{onecorrectset0-tar}(\ref{inv0-tar}),
            $\set_0$ and $\set_1$ are disjoint sets.

        We now prove that $\set_1\cup\set_0$ contains fewer than $n-f$ {\ct} processes; this immediately implies that there is a {\ct} process $p \not \in \set_1 \cup \set_0$.
            There are two possible cases:
			\begin{enumerate}
				\item Case 1: $\set_1$ contains at most $f$ {\ct} processes.
				By Lemma~\ref{onecorrectset0-tar}(\ref{inv1-tar}), $|\set_0| \le f$.
				So $\set_1 \cup \set_0$ contains at most $2f$ {\ct} processes.
				Since $3f < n$, we have $2f < n-f$.
				
				\item Case 2: $\set_1$ contains at least $f+1$ {\ct} processes.
				By Lemma~\ref{onecorrectset0-tar}(\ref{inv2-tar}), $\set_0$ does \emph{not} contain any {\ct} process.
				So $\set_1 \cup \set_0$ contains at most $|\set_1|$ {\ct} processes.
				By Lemma~\ref{onecorrectset0-tar}(\ref{inv1-tar}),  $|\set_1| < n-f$.
				\end{enumerate} 
				In both cases, $\set_1\cup\set_0$ contains fewer than $n-f$ {\ct} processes.
				\end{proof}

\begin{lemma}\label{line7inf-tar}
	Consider any \validpg{(v)} by a {\ct} reader $p_k$.
	Every instance of the Repeat-Until loop at lines~\ref{repeat-tar}-\ref{until-tar} of this $\valid(v)$ terminates.
\end{lemma}

\begin{proof}
	Consider any \validpg{(v)} by a {\ct} reader $p_k$.
	Assume for contradiction that there is an instance of the Repeat-Until loop at lines~\ref{repeat-tar}-\ref{until-tar} of this $\valid(v)$ that does not terminate.
	Let $t$ be the time when $p_k$ executes line~\ref{repeat-tar} for the first time in this instance of the Repeat-Until loop.
        Since this instance does not terminate, 
		$p_k$ never finds the condition of line~\ref{until-tar} holds after $t$ ($\star$).
	
        Since lines~\ref{whileloop-tar}-\ref{repeat-tar} do not change $\set_1$ and $\set_0$,
        Lemma~\ref{correctoutside-tar} implies that 
            there is a {\ct} process $p_a \not\in \set_1\cup\set_0$ at line~\ref{repeat-tar} at time $t$.
	Let $c^*$ be the value of $C_k$ at time $t$;
    	by line~\ref{ckplus-tar},
    		$c^* \ge 1$.
        Since lines~\ref{repeat-tar}-\ref{until-tar} of the Repeat-Until loop do not change $\set_1$, $\set_0$, or $C_k$,
            and $p_k$ never exits this loop,
           $p_a \not\in \set_1\cup\set_0$ and $C_k=c^* $ at all times $t' \ge t$.

	\begin{claim}\label{rakck-tar}
		There is a time $t'$ such that $R_{ak}$ contains $\langle -, c^* \rangle$ for all times after $t'$.
	\end{claim}
	\begin{proof}
	Let $\fresh()$ be the help procedure of $p_a$.
	Since $C_k = c^*$ for all times after $t$ and $p_a$ is {\ct},
		there is an iteration of the while loop of $\fresh()$
		in which $p_a$ reads $c^*$ from $C_k$ at line~\ref{collectck-tar}.
	Consider the \emph{first} iteration of the while loop of $\fresh()$ in which $p_a$ reads $c^*$ from $C_k$ at line~\ref{collectck-tar}.
	Since $p_k$ is {\ct}, 
				by Observation~\ref{monock-tar},
				$p_a$ reads non-decreasing values from $C_k$,
				and so $p_a$ has $ c^* \ge prev\_c_k$ at line~\ref{askers-tar}.
	Since $c^* \ge 1$,
		there are two cases.
		\begin{itemize}
			\item Case 1: $c^*=1$.
			Then $p_a$ has $prev\_c_k \le 1$ at line~\ref{askers-tar}.
			Since $prev\_c_k$ is initialized to 0 (line~\ref{collectck-init}) and $p_a$ reads $c^*=1$ from $C_k$ at line~\ref{collectck-tar} for the first time,
				$prev\_c_k = 0$ at line~\ref{askers-tar}. 
			So $p_a$ finds $c^*=1 > prev\_c_k=0$ holds and inserts $p_k$ into $askers$ at line~\ref{askers-tar}. 
			\item Case 2: $c^*>1$.
			Since $p_a$ reads $c^*$ from $C_k$ at line~\ref{collectck-tar} for the first time and $ c^* \ge prev\_c_k$,
				$c^* > prev\_c_k$ at line~\ref{askers-tar}. 
			So $p_a$ inserts $p_k$ into $askers$ at line~\ref{askers-tar}. 
		\end{itemize}
	So in both cases, $p_a$ inserts $p_k$ into $askers$ at line~\ref{askers-tar}. 
	Since $p_a$ is {\ct},
		in the same iteration of the while loop of $\fresh()$,
		$p_a$ writes $\langle - , c^* \rangle$ into $R_{ak}$ at line~\ref{fresh1-tar},
		say at time $t'$ 
		and then it sets $prev\_c_k$ to $c^*$ at line~\ref{setprev-tar} (i).
	Since $C_k = c^*$ for all times after $t$,
		by Observation~\ref{monock-tar},
		$C_k= c^*$ for all times after $t'$.
	Furthermore, by line~\ref{collectck-tar},
		$p_a$ has $c_k=c^*$ for all times after $t'$ (ii).
	From (i) and (ii),
		 by line~\ref{askers-tar},
		$p_a$ does not insert $p_k$ into $askers$ in any future iteration of the while loop of $\fresh()$.
	So by line~\ref{tellasker-tar},
		$p_a$ never writes to $R_{ak}$ after $t'$,
		i.e., $R_{ak}$ contains $\langle -, c^* \rangle$ for all times after $t'$.
	\end{proof}

	Since $p_k$ executes lines~\ref{findone-tar}-\ref{until-tar} infinitely often after $t$,
		$p_k$ reads from $R_{ak}$ at line~\ref{readri-tar} infinitely often after $t$.
	By Claim~\ref{rakck-tar},
		eventually $p_k$ reads $\langle -, c^* \rangle$ from $R_{ak}$ after $t$.
	Thus, since $p_a \not\in \set_1 \cup\set_0$ and $p_k$ has $C_k =c^*$ for all times after $t$, 
		$p_k$ finds the condition of line~\ref{until-tar} holds after $t$
		--- a contradiction to ($\star$). 
\end{proof}

\begin{observation}\label{set1increment-tar}
	Consider any \validpg{(v)} by a {\ct} reader $p_k$.
	When $p_k$ executes line~\ref{set1-tar} of this $\valid(v)$, $p_k$ increments the size of $\set_1$.
\end{observation}

\begin{observation}\label{set0increment-tar}
	Consider any \validpg{(v)} by a {\ct} reader $p_k$.
	When $p_k$ executes line~\ref{set0-tar} of this $\valid(v)$, $p_k$ increments the size of $\set_0$.
\end{observation}

\begin{lemma}\label{validpterminate}
    Every \validpg{(-)} by a {\ct} process terminates.
\end{lemma}
\begin{proof}
Assume for contradiction,
        that for some value $v$
		there is a \validpg{(v)}
        by a {\ct} reader $p_k$ that does not terminate,
		i.e.,
		$p_k$ takes an infinite number of steps in this \validpg{(v)}.
	By Lemma~\ref{line7inf-tar},
        $p_k$ must execute infinitely many iterations of the while loop at line~\ref{whileloop-tar} of this $\valid(v)$.
	So $p_k$~executes line~\ref{set1-tar} or line~\ref{set0-tar} of this $\valid(v)$ infinitely often.	
	\begin{itemize}
		\item Case 1: $p_k$ executes line~\ref{set1-tar} infinitely often.	
				By Observations~\ref{monoset1-tar} and \ref{set1increment-tar},
					there is an iteration of the while loop at line~\ref{whileloop-tar} of this $\valid(v)$ 
					in which $p_k$ has $|\set_1| \ge n-f$ at line~\ref{set1-tar}. 
				Since $p_k$ does not change $|\set_1|$ between line~\ref{set1-tar} and line~\ref{return1-tar},
					in that iteration,
					$p_k$ finds that the condition $|\set_1| \ge n-f$ holds at line~\ref{return1-tar} of $\valid(v)$.
				So $\valid(v)$ returns at line~\ref{return1-tar}.
		\item Case 2: $p_k$ executes line~\ref{set1-tar} only a finite number of times.
			Then $p_k$ executes line~\ref{set0-tar} infinitely often.
			So there is a time $t$ such that 
				$p_k$ executes line~\ref{set0-tar} infinitely often after $t$
				while $p_k$ never executes line~\ref{set1-tar} after $t$.
			This implies $p_k$ never executes line~\ref{empty0-tar} after $t$.
			Since $|\set_0|$ decreases only at line~\ref{empty0-tar},
				$|\set_0|$ never decreases after $t$.
			Since $p_k$ executes line~\ref{set0-tar} infinitely often after $t$
				and $|\set_0|$ never decreases after $t$,
				by Observations~\ref{set0increment-tar},
					there is an iteration of the while loop at line~\ref{whileloop-tar} of $\valid(v)$
					in which $p_k$ has $|\set_0| > f$ at line~\ref{set0-tar} of $\valid(v)$. 
			Since $p_k$ does not change $|\set_0|$ between line~\ref{set0-tar} and line~\ref{return0-tar},
				in that iteration,
				$p_k$ finds that the condition $|\set_0|>f$ holds at line~\ref{return0-tar} of $\valid(v)$.
			So $\valid(v)$ returns at line~\ref{return0-tar}.
	\end{itemize}
	In both cases,
		this $\valid(v)$ returns, 
			a contradiction to the assumption that $\valid(v)$ does not complete.
\end{proof}

\begin{restatable}{theorem}{verifyt} \label{termination-tar}
    [\textsc{Termination}] Every $\Set$, $\Test$, and $\valid$ operation by a {\ct} process completes.
\end{restatable}

\begin{proof}
	From the code of the $\Set(-)$ procedure,
		it is obvious that every $\Set$ operation by a {\ct} process completes.
        Moreover,
        from the code of the $\Test$ and $\valid(-)$ procedures, by Lemma~\ref{validpterminate},
            it is clear that every $\Test$ and $\valid$ operation by a {\ct} process completes.
	
\end{proof}

\subsection{Byzantine Linearizability}\label{linear-tar}
Recall that $H$ is an arbitrary history of the implementation given by Algorithm~\ref{code-tar}.
We now prove that $H$ is Byzantine linearizable with respect to a SWMR {\tar} (Definition~\ref{def-tar}).
We start by proving the following observations and lemmas about $H$.

\begin{lemma}~\label{v0alwaystrue}
	Every {\validpg{(v_0)}} by a {\ct} process returns $\true$.
\end{lemma}
\begin{proof}
Suppose, for contradiction, there is a {\validpg{(v_0)}} by a {\ct} process $p_k$ that does not return $\true$.
By Lemma~\ref{validpterminate},
    this $\valid(v_0)$ completes.
According to the code of the \valid$(-)$ procedure,
    this $\valid(v_0)$ returns $\false$ at line~\ref{return0-tar}.
This implies that
    $p_k$ finds the condition $|\set_0| > f$ holds at line~\ref{return0-tar}.
Thus, since $p_k$ is {\ct},
    by lines~\ref{readri-tar} and lines~\ref{notv-tar}-\ref{set0-tar},
    there are at least $f+1$ processes $p_j$ such that $p_k$ reads $\langle r_j ,c_j \rangle$ with $v_0\not\in r_j$ from $R_{jk}$ at line~\ref{readri-tar} of $\valid(v_0)$.
Since there are at most $f$ faulty processes,
    there is at least one {\ct} process $p_a$ such that $p_k$ reads $\langle r_a ,c_a \rangle$ with $v_0\not\in r_a$ from $R_{ak}$ at line~\ref{readri-tar} of $\valid(v_0)$.
Since $p_a$ is {\ct},
    by line~\ref{readrj-tar} and line~\ref{fresh1-tar},
    there is a time $t$ when $v_0\not\in R_a$ --- a contradiction to Observation~\ref{correctonlyv0-tar}.    
\end{proof}

\begin{lemma}~\label{notreturn0-tar}
	Suppose there is a time $t$ such that
		$v \in R_1$ at all times $t' \ge t$ or at least $f+1$ processes $p_i$ have $v \in R_i$ at all times $t' \ge t$.
	If a {\ct} reader executes $\valid(v)$ after time $t$, 
		then it does not insert any {\ct} process into $\set_0$ in this $\valid(v)$.
\end{lemma}

\begin{proof}
	Assume for contradiction,
		there is a time $t$ such that
		(1) $v \in R_1$ at all times $t' \ge t$ or at least $f+1$ processes $p_i$ have $v \in R_i$ at all times $t' \ge t$, and
		(2) a {\ct} reader $p_k$ executes $\valid(v)$ after time $t$,  
		but $p_k$ inserts a {\ct} process $p_j$ into $\set_0$ in this $\valid(v)$.
 
	Consider the iteration of the while loop at line~\ref{whileloop-tar} of the $\valid(v)$ in which $p_k$ inserts $p_j$ into $\set_0$ at line~\ref{set0-tar}. 
	Let $c_k$ be the value of $C_k$ after $p_k$ increments $C_k$ at line~\ref{ckplus-tar} in this iteration.
	Since $C_k$ is initialized to $0$, 
		$c_k \ge 1$ and $p_k$ writes $c_k$ into $C_k$ after time $t$.
	By Observation~\ref{monock-tar},
		for all times $t' \le t$, $C_k$ contains values that are less than $c_k$ ($\star$).
	Since $p_k$ inserts $p_j$ into $\set_0$ at line~\ref{set0-tar},
		$p_k$ reads $\langle r_j, c_j \rangle$ with $v\not\in r_j$ and $c_j \ge c_k$ from $R_{jk}$ at line~\ref{readri-tar}.
	Since $c_k \ge 1$, $c_j\ge 1$.
	Since $R_{jk}$ is initialized to $\langle \emptyset, 0 \rangle$,
		it must be that $p_j$ writes $\langle r_j, c_j \rangle$ with $v\not\in r_j$ into $R_{jk}$ 
		at line~\ref{fresh1-tar} in some iteration of the while loop of $\fresh()$ ($\star\star$).

	Consider the iteration of the while loop of $\fresh()$ procedure
		in which $p_j$ writes $\langle r_j, c_j \rangle$ into $R_{jk}$ at line~\ref{fresh1-tar}.
	Since $p_j$ is {\ct}, in that iteration:	
		(a) $p_j$ reads $c_j$ from $C_k$ at line~\ref{collectck-tar}, say at time $t_j^1$ and then 
		(b) $p_j$ inserts $p_k$ into $askers$ at line~\ref{askers-tar}, say at time $t_j^2 > t_j^1$.
	Since $c_j \ge c_k$,
		by ($\star$),
		$p_j$ reads $c_j$ from $C_k$ at line~\ref{collectck-tar} after time $t$, i.e., $t_j^1 > t$.
	Since $t_j^2 > t_j^1$, $t_j^2 > t$, 
		i.e., $p_j$ inserts $p_k$ into $askers$ at line~\ref{askers-tar} at time $t_j^2 > t$.
	Thus, by (1) and Lemma~\ref{asker1-tar}, 
		$p_j$ writes $\langle r_j, c_j \rangle$ with $v\in r_j$ into $R_{jk}$ at line~\ref{fresh1-tar} --- a contradiction to ($\star\star$).
\end{proof}

\begin{lemma}~\label{testreturn1-tar}
	Suppose there is a time $t$ such that $v \in R_1$ at all times $t' \ge t$ or at least $f+1$ processes $p_i$ have $v \in R_i$ at all times $t' \ge t$.
	If a {\ct} reader executes $\valid(v)$ after time $t$,  
		then this $\valid(v)$ returns $\true$.
\end{lemma}

\begin{proof}
	Suppose that there is a time $t$ such that
        (1) $v \in R_1$ at all times $t' \ge t$ or at least $f+1$ processes $p_i$ have $v \in R_i$ at all times $t' \ge t$, and 
		(2) a {\ct} reader $p_k$ executes $\valid(v)$ after time $t$.
        Note that $\set_0$ is initialized to $\emptyset$ at the start of this operation.
	By Lemma~\ref{notreturn0-tar},
            $p_k$ does not insert any {\ct} process into $\set_0$ in this $\valid(v)$.
            Thus, $\set_0$ can contain only faulty processes, and so 
		$|\set_0| \le f$ in the $\valid(v)$.
	So $p_k$ never finds that the condition $|\set_0| > f$ holds in line~\ref{return0-tar} of the $\valid(v)$ and so this $\valid(v)$ never returns $\false$ at line~\ref{return0-tar}.
	Thus,
		by Lemma~\ref{validpterminate} and the code of the $\valid$ procedure,
		this $\valid(v)$ returns $\true$.
\end{proof}

By Observation~\ref{correctonlywrites1-tar} and Lemma~\ref{testreturn1-tar}, we have the following:

\begin{corollary}~\label{cor-testreturn1-tar}
	Suppose the writer $p_1$ is correct and has $v \in R_1$ at some time $t$.
	If a {\ct} reader executes $\valid(v)$ after time $t$,  
		then this $\valid(v)$ returns $\true$.
\end{corollary}

\begin{definition}
Let ${\hct}$ be the set of processes that are correct in the history $H$.
\end{definition}

\begin{definition}
Let $H|{\hct}$ be the history consisting of all the steps of all the \emph{{\ct}} processes in $H$ (at the same times they occur in $H$).
\end{definition}

By Theorem~\ref{termination-tar}, all the processes that are correct in $H$ (i.e., all the processes in {\hct}) complete their operations, so:

\begin{observation}\label{allcomplete}
    Every operation in $H|{\hct}$ is complete (i.e., it has both an invocation and a response).
\end{observation}

\begin{observation}\label{correctop-tar}
    An operation $o$ by a process $p$ is in $H|{\hct}$ if and only if $o$ is also in $H$ and $p \in {\hct}$.
\end{observation}

\begin{observation}\label{sameop-tar}
For all processes $p \in {\hct}$,
    $p$ has the same operations (i.e., the same steps) in both $H$ and $H|{\hct}$.
    Furthermore, the $[invocation, response]$ intervals of these operations are the same in both $H$ and $H|{\hct}$.
\end{observation}

To show that $H$ is Byzantine linearizable with respect to a SWMR \tar, 
    we must show that there is a history $H'$ such that:
    (a)~$H'|{\hct} = H|{\hct}$, and
    (b)~$H'$ is linearizable with respect to a SWMR {\tar}.
There are two cases, depending on whether the writer is {\ct} in $H$.

\subsubsection*{Case 1: the writer $p_1$ is {\ct} in $H$.}\label{case-writer-correct}

Let $H' = H|{\hct}$.
We now show that the history $H'$ is linearizable with respect to a SWMR \tar.
To do so, 
    we first define the linearization points of the operations in $H'$; 
and then we use these linearization points to define a linearization $L$ of $H'$ such that:
(a) $L$ respects the precedence relation between the operations of $H'$, and 
(b) $L$ conforms to the sequential specification of a SWMR {\tar} (which is given in Definition~\ref{def-tar}).

Since the writer $p_1$ is {\ct},
    we have the following:

\begin{observation}\label{reg-ob-2}
   In $H$,
    if $\langle \ell, v\rangle$ is in $R_1$ at time $t$ such that $\forall \langle \ell', v' \rangle \in R_1: \langle \ell, v \rangle  \ge  \langle \ell', v' \rangle $ at time $t$,
    then:
        \begin{itemize}
            \item either $\langle \ell,v\rangle$ is the last tuple that the writer inserts into $R_1$ at line~\ref{r1-tar} of a $\Set(v)$ operation before time $t$,
            \item or $v = v_0$ and the writer does not insert any tuple
            into $R_1$ at line~\ref{r1-tar} of a $\Set(-)$ operation before time $t$.
        \end{itemize}
      
\end{observation}

\begin{lemma}\label{register-tar-1}
        In $H$,
        every $\Test$ operation by a {\ct} reader returns at line~\ref{checksign-tar}.
\end{lemma}
\begin{proof}
Since the writer $p_1$ is correct in $H$, it is clear that $R_1$
    always stores a set of tuples
    of the form $\langle \ell, v \rangle$.
When a correct reader invokes a $\Test$ operation in $H$, it
    reads $R_1$ at line~\ref{readerr1-tar} at some time $t$.
Thus, in line~\ref{verifyp-tar}, the reader calls the procedure
    $\valid(v)$ for a value $v$ such that
    $\langle \ell, v \rangle $ is in $R_1$ at time $t$.
Since the writer is correct,
    by Corollary~\ref{correctonlywrites1-tar}, $\langle \ell, v \rangle$ is in $R_1$
    at all times $t' \ge t$.
So, by Lemma~\ref{testreturn1-tar}, the $\valid(v)$ procedure returns $\true$ in line~\ref{verifyp-tar}.
Thus, this $\Test$ operation returns in line~\ref{checksign-tar}.
\end{proof}

\begin{lemma}\label{register-tar-0}
        Let $\Test$ be an operation by a {\ct} reader that returns a value $v$ in $H$.
        Then:
        \begin{itemize}
            \item either $\langle -,v\rangle$ is the last tuple that the writer inserts into $R_1$ at line~\ref{r1-tar} of a $\Set(v)$ operation
            before the reader reads $R_1$ at line~\ref{readerr1-tar} of the $\Test$ operation in $H$,
            \item or $v = v_0$ and the writer does not insert any tuple into $R_1$ at line~\ref{r1-tar} of a $\Set(-)$ operation before the reader reads $R_1$ at line~\ref{readerr1-tar} of the $\Test$ operation in $H$.
        \end{itemize}

\end{lemma}
\begin{proof}
    Let $\Test$ be an operation by a {\ct} reader $p_k$ that returns a value $v$ in $H$.
    By Lemma~\ref{register-tar-1},
        $\Test$ returns $v$ at line~\ref{checksign-tar}.
    So by lines~\ref{readerr1-tar}-\ref{checksign-tar},
        $p_k$ reads $R_1$ into $r$ with $\langle \ell, v \rangle \in r$ for some $\ell$ at line~\ref{readerr1-tar}, say at time $t$.
    Furthermore, by line~\ref{maxl-tar},
            $\langle \ell, v \rangle$ is the tuple in $r$ such that $\forall \langle \ell', v' \rangle \in r: \langle \ell, v \rangle  \ge  \langle \ell', v' \rangle $.
    So $\langle \ell, v \rangle$ is the tuple in $R_1$ such that $\forall \langle \ell', v' \rangle \in R_1: \langle \ell, v \rangle  \ge  \langle \ell', v' \rangle $ at time $t$.
    By Observation~\ref{reg-ob-2}:
        \begin{itemize}
            \item either $\langle \ell, v \rangle$ is the last tuple that the writer inserts into $R_1$ at line~\ref{r1-tar} of a $\Set(v)$ operation before time $t$,
            i.e., before the reader reads $R_1$ at line~\ref{readerr1-tar} of the $\Test$ operation in $H$.
            \item or $v = v_0$ and the writer does not insert any tuple into $R_1$ at line~\ref{r1-tar} of a $\Set(-)$ operation before time $t$,
            i.e.,
            before the reader reads $R_1$ at line~\ref{readerr1-tar} of the $\Test$ operation in $H$.
        \end{itemize}
\end{proof}

\begin{lemma}\label{settest-tar-H}
In $H$,
    \begin{enumerate}
        \item
        If a {\validpg{(v)}} by a {\ct} reader returns $\false$, and
        the writer $p_1$ inserts $\langle -,v \rangle$ into $R_1$ at line~\ref{r1-tar} of a $\Set(v)$ operation at some time $t$,
        then the reader invokes this $\valid(v)$ 
        before time $t$ and $v\ne v_0$.
        
        \item If a {\validpg{(v)}} by a {\ct} reader returns $\true$ at time $t$,
	then the writer $p_1$ inserts $\langle -,v \rangle$ into $R_1$ at line~\ref{r1-tar} of a $\Set(v)$ operation before time $t$ or $v = v_0$.
   \end{enumerate}
\end{lemma}

\begin{proof}~

        \textsc{Part (1)}: Suppose, for contradiction, that:
        (1) a \validpg{(v)} by a {\ct} reader returns $\false$, and
        (2) the writer $p_1$ inserts $\langle -,v \rangle$ into $R_1$ at line~\ref{r1-tar} of a $\Set(v)$ operation at some time $t$, but
        (3) the reader invokes this $\valid(v)$ 
        after time $t$ or $v=v_0$.
        \begin{itemize}
            \item Case 1: $v=v_0$. Then by (1) the {\validpg{(v_0)}} by a {\ct} reader returns $\false$ --- a contradiction to Lemma~\ref{v0alwaystrue}.
            \item Case 2: $v\ne v_0$.
            Then by (3), the reader invokes this $\valid(v)$ 
            after time $t$.
            Since $p_1$ is {\ct},
            by (2) and Observation~\ref{correctonlywrites1-tar},
		$v \in R_1$ at all times $t' \ge t$.
        Thus,
        by (3) and Lemma~\ref{testreturn1-tar},
        this $\valid(v)$ 
        returns $\true$
        --- a contradiction to (1).
        \end{itemize}

        \textsc{Part (2)}: Suppose a {\validpg{(v)}} by a {\ct} reader $p_k$ returns $\true$ at time $t$.
	This implies $p_k$ finds that the condition $|\set_1| \ge n-f$ holds at line~\ref{return1-tar} of this $\valid(v)$ 
    by time $t$.
	Since $n \ge 3f+1$, $|\set_1| \ge 2f+1$ at time $t$.
	Since there are at most $f$ faulty processes,
		there are at least $f+1$ {\ct} processes in $\set_1$ of $p_k$ at time $t$.
    Since the writer $p_1$ is {\ct},
        by Lemma~\ref{correctp1v-tar},
        $p_1$ has $v \in R_1$ at time $t$.
    Since $R_1$ is initialized to $\{ \langle 0, v_0 \rangle \}$,
        $p_1$ inserts $\langle -,v \rangle$ into $R_1$ at line~\ref{r1-tar} of a $\Set(v)$ operation before time $t$ or $v = v_0$.
\end{proof}

Since the writer is correct in $H$ and $H'=H|{\hct}$,
    by Observations~\ref{correctop-tar} and~\ref{sameop-tar},
     Lemmas~\ref{v0alwaystrue}, \ref{register-tar-0}, and \ref{settest-tar-H} also hold for the history $H'$, as~stated~below:

\begin{corollary}\label{v0alwaystrue1}
   In $H'$, every {\validpg{(v_0)}} returns $\true$. 
\end{corollary}

\begin{corollary}\label{register-tar-00}
        Let $\Test$ be an operation that returns a value $v$ in $H'$.
	  Then:
        \begin{itemize}
            \item either $\langle -, v \rangle$ is the last tuple that the writer inserts into $R_1$ at line~\ref{r1-tar} of a $\Set(v)$ operation
            before the reader reads $R_1$ at line~\ref{readerr1-tar} of the $\Test$ operation in $H'$,
            \item or $v = v_0$ and the writer does not insert any tuple into $R_1$ at line~\ref{r1-tar} of a $\Set(-)$ operation before the reader reads $R_1$ at line~\ref{readerr1-tar} of the $\Test$ operation in $H'$.
        \end{itemize}
\end{corollary}

\begin{corollary}\label{settest-tar}
In $H'$,
    \begin{enumerate}
        \item
        If a {\validpg{(v)}} returns $\false$, and
        the writer $p_1$ inserts $\langle -,v \rangle$ into $R_1$ at line~\ref{r1-tar} of a $\Set(v)$ operation at some time $t$,
        then the reader invokes this $\valid(v)$ 
        before time $t$ and $v\ne v_0$.
        
        \item If a {\validpg{(v)}} returns $\true$ at time $t$,
	then the writer $p_1$ inserts $\langle -,v \rangle$ into $R_1$ at line~\ref{r1-tar} of a $\Set(v)$ operation before time $t$ or $v = v_0$.
   \end{enumerate}
\end{corollary}

We now prove that the history $H'$ is linearizable with respect to an \tar.
First, we define the \emph{linearization point} of each operation in $H'$ as follows.
\begin{definition}~\label{testsetlinearization-tar}
Let $o$ be an operation in  $H'$.
    \begin{itemize}
    \item If $o$ is a $\Set(v)$ operation for some value $v$, then
                the linearization point of $o$ is the time when the writer $p_1$ inserts $\langle -,v \rangle$ into $R_1$ at line~\ref{r1-tar} in $o$.

           \item  If $o$ is a $\Test$ operation,
        then the linearization point of $o$ is the time when the reader reads $R_1$ at line~\ref{readerr1-tar} in $o$. 

	\item  If $o$ is a $\valid(-)$ operation,
        then: \begin{itemize}
          \item If $o$ returns $\true$,
	   the linearization point of $o$ is the time of the response step of $o$.
            \item If $o$ returns $\false$,
	   the linearization point of $o$ is the time of the invocation step of~$o$.
        \end{itemize}

\end{itemize}
\end{definition}

Note that in Definition~\ref{testsetlinearization-tar} the linearization point of every operation $o$ is between the invocation and response time of $o$. 
Thus, the following holds.
\begin{observation}\label{precedes-tar}
    If an operation $o$ precedes an operation $o'$ in $H'$,
        then the linearization point of $o$ is before the linearization point of $o'$.
\end{observation}

We now use the linearization points of all the operations in $H'$ to define the following linearization $L$ of $H'$:

\begin{definition}\label{hlp-tar}
	Let $L$ be a \emph{sequence} of operations such that:
	\begin{enumerate}
		\item\label{uno-1} An operation $o$ is in $L$ if and only if it is in $H'$.
		\item\label{due-1} An operation $o$ precedes an operation $o'$ in $L$ if and only if
     the linearization point of~$o$ is before the linearization point of $o'$ in $H'$.
	\end{enumerate}
\end{definition}

By Definition~\ref{hlp-tar}(\ref{uno-1}), $L$ is a linearization of $H'$.
By Observation~\ref{precedes-tar} and Definition~\ref{hlp-tar}(\ref{due-1}), $L$ respects
    the precedence relation between the operations of $H'$. More precisely:
\begin{observation}\label{precedes-l-tar}
    If an operation $o$ precedes an operation $o'$ in $H'$,
        then $o$ precedes $o'$ in~$L$.
\end{observation}

In the following theorems, we prove that the linearization $L$ of $H'$ conforms to the sequential specification of a SWMR \tar.

\begin{theorem}\label{rw-tar}
    In $L$,
        if a $\Test$ operation returns a value $v$, then:
        \begin{itemize}
            \item either some $\Set(v)$ operation precedes it and this $\Set(v)$ operation is the last $\Set(-)$ operation that precedes it,
            \item or $v = v_0$ (the initial value of the register) and no $\Set(-)$ operation precedes it.
        \end{itemize}
       
\end{theorem}

\begin{proof}
    Suppose a $\Test$ operation returns a value $v$ in $L$.
    By the Definition~\ref{hlp-tar}(1) of $L$,
        this $\Test$ operation is also in~$H'$.
    By Corollary~\ref{register-tar-00}, there are two cases:
    \begin{itemize}
        \item Case 1: $\langle -, v \rangle$ is the last tuple that the writer inserts into $R_1$ at line~\ref{r1-tar} of a $\Set(v)$ operation
            before the reader reads 
            $R_1$ at line~\ref{readerr1-tar} of the $\Test$ operation in~$H'$.
            So by Definition~\ref{testsetlinearization-tar},
                the linearization point $t$ of this $\Set(v)$ operation is before the linearization point $t'$ of the $\Test$ operation, and there is no other $\Set(-)$ operation with a linearization point between $t$ and $t'$.     
             Thus, by the Definition~\ref{hlp-tar}(2) of $L$,
                 this $\Set(v)$ operation precedes the $\Test$ operation in $L$, 
                 and it is the last $\Set(-)$ operation that precedes this $\Test$ operation in $L$.
        \item Case 2: $v = v_0$ and the writer does not insert any tuple into $R_1$ at line~\ref{r1-tar} of a $\Set(-)$ operation before the reader reads $R_1$ at line~\ref{readerr1-tar} of the $\Test$ operation in $H'$. 
        By Definition~\ref{testsetlinearization-tar},
            there is no $\Set(-)$ operation with a linearization point before the linearization point of this $\Test$ operation.
        So by the Definition~\ref{hlp-tar}(2) of $L$,
            no $\Set(-)$ operation precedes this $\Test$ operation in $L$.        
    \end{itemize}
    Therefore, in $L$,
        either a $\Set(v)$ operation precedes the $\Test$ operation and it is the last $\Set(-)$ operation that precedes this $\Test$ operation,
        or $v = v_0$ and no $\Set(-)$ operation precedes the $\Test$ operation.
\end{proof}

\begin{lemma}\label{lsv1-tar}
        In $L$,
		if a $\valid(v)$ operation returns~$\false$,
            then there is no $\Set(v)$ operation that precedes this $\valid(v)$ operation and $v\ne v_0$.
\end{lemma}

\begin{proof}
	Assume for contradiction that,
		in $L$,
		(1) a $\valid(v)$ operation returns $\false$ and (2) there is a $\Set(v)$ operation that precedes this $\valid(v)$ operation or $v=v_0$. There are two cases:
        \begin{itemize}
            \item Case 1: $v=v_0$.
             By Definition~\ref{hlp-tar}(1), 
                the $\valid(v)$ operation is also in $H'$.
            Since $v=v_0$, 
                by Corollary~\ref{v0alwaystrue1},
                the $\valid(v)$ operation returns $\true$ in $H'$.
                Therefore it also returns $\true$ in $L$ --- a contradiction to (1).

            \item Case 2: $v\ne v_0$.
            Then by (2), 
                there is a $\Set(v)$ operation that {\precedes} this $\valid(v)$ operation in $L$. 
            By Definition~\ref{hlp-tar}, 
            the $\Set(v)$ operation and the $\valid(v)$ operations are also in $H'$, 
            and the linearization point $t$ of the $\Set(v)$ operation is before the linearization point $t'$ of the $\valid(v)$ operation in $H'$.
        So $t < t'$ ($\star$).
        Thus, by Definition~\ref{testsetlinearization-tar},
            the writer $p_1$ inserts $\langle -,v \rangle$ into $R_1$ at line~\ref{r1-tar} in the $\Set(v)$ operation
            before the reader invokes the $\valid(v)$ operation in $H'$
            --- a contradiction to Corollary~\ref{settest-tar}(1).
        \end{itemize}
   \end{proof}

\begin{lemma}\label{lsv2-tar}
	In $L$, 
	if a $\valid(v)$ operation returns $\true$,
	then there is a $\Set(v)$ operation that precedes this $\valid(v)$ operation or $v=v_0$.
\end{lemma}

\begin{proof}
	Suppose in $L$, a $\valid(v)$ operation returns $\true$.
        By the Definition~\ref{hlp-tar}(1) of $L$,
            this $\valid(v)$ operation is also in $H'$.
        Let $t$ be the time when the $\valid(v)$ operation returns $\true$ in $H'$.
        By Corollary~\ref{settest-tar}(2),
            there are two cases:
            \begin{itemize}
                \item Case 1: $v= v_0$. Then this lemma holds.

                \item Case 2:  the writer $p_1$ inserts $\langle -,v \rangle$ into $R_1$ at line~\ref{r1-tar} of a $\Set(v)$ operation before time $t$ in $H'$.
                Let $t'$ be the time when the writer $p_1$ inserts $\langle -,v \rangle$ into $R_1$ at line~\ref{r1-tar} of this $\Set(v)$ operation.
                Then $t' < t$.
                By Definition~\ref{testsetlinearization-tar},
                $t'$ is the linearization point of the $\Set(v)$ operation and
                $t$ is the linearization point of the $\valid(v)$ operation.
        Thus, since $t' < t$,
            by Definition~\ref{hlp-tar}(2),
            this $\Set(v)$ operation precedes this $\valid(v)$ operation in $L$.
            \end{itemize}
\end{proof}

By Lemmas~\ref{lsv1-tar} and~\ref{lsv2-tar},
    we have the following theorem.

\begin{theorem}\label{sv-tar}
    In $L$, 
	a $\valid(v)$ operation returns $\true$ if and only if
	there is a $\Set(v)$ operation that precedes this $\valid(v)$ operation or $v=v_0$.
\end{theorem}
\color{black}

\begin{theorem}\label{linearthe-tar-case1}
    If the writer is correct in $H$, the history $H$ is Byzantine linearizable with respect to a SWMR {\tar}.
\end{theorem}

\begin{proof}
Assume that the writer $p_1$ is correct in $H$.
    To prove that $H$ is
        Byzantine linearizable with respect to a SWMR {\tar} (Definition~\ref{def-tar}),
    we must prove that there is a history $H'$ such that:
    
    \begin{compactenum}[(1)]
        \item $H'|{\hct} = H|{\hct}$.
        \item $H'$ is linearizable with respect to a SWMR {\tar}. 
    \end{compactenum}

    Let $H'= H|{\hct}$.
    Clearly, $H'$ satisfies~(1).
        
    Let $L$ be the linearization of $H'$ given in Definition~\ref{hlp-tar}.
    \begin{compactenum}[(a)]
        \item By Observation~\ref{precedes-l-tar}, $L$ respects the precedence relation between the operations of $H'$.
        \item By Theorems~\ref{rw-tar} and~\ref{sv-tar}, $L$ conforms to the sequential specification of a SWMR {\tar}.
        \end{compactenum}
    Thus, $H'$ also satisfies~(2).
\end{proof}

\subsubsection*{Case 2: the writer $p_1$ is not {\ct} in $H$.}\label{case-writer-not-correct}

We must show that
    the arbitrary history $H$ of the implementation given by Algorithm~\ref{code-tar} is Byzantine linearizable.
We first prove that
    if a process validates
    (the signature of) a value~$v$,
    thereafter every process will also be able to validate it.
    
\begin{lemma}\label{testatestb-tar}
	Consider any \valid$(v)$ and \validpg{(v)'}s by {\ct} readers~in~$H$.
	If $\valid(v)$ returns $\true$ and it precedes $\valid(v)'$, then $\valid(v)'$ returns $\true$.
\end{lemma}

\begin{proof}
	Consider any \valid$(v)$ and \validpg{(v)'}s by {\ct} readers $p_a$ and $p_b$, respectively, in $H$.
	Suppose that $\valid(v)$ precedes $\valid(v)'$ and $\valid(v)$ returns $\true$.
	Let $t$ be the time when $\valid(v)$ returns $\true$.
	Since $\valid(v)$ precedes $\valid(v)'$,
		$p_b$ invokes $\valid(v)'$ after time~$t$~($\star$).

	Since $\valid(v)$ returns $\true$ at time $t$,
		$p_a$ finds that the condition $|\set_1| \ge n-f$ holds in line~\ref{return1-tar} of $\valid(v)$, say at time $t_a \le t$.
	Since $n \ge 3f+1$, $|\set_1| \ge 2f+1$ at time $t_a$.
	Since there are at most $f$ faulty processes,
		there are at least $f+1$ {\ct} processes in $\set_1$ of $p_a$ at time $t_a$.
	By Lemma~\ref{ri1beforeset1-tar},
		at least $f+1$ {\ct} processes $p_i$ have $v\in R_i$ at all times $t' \ge t_a$.
	Since $t_a < t$,
		at least $f+1$ {\ct} processes $p_i$ have $v\in R_i$ at all times $t' \ge t$.
	Thus, by ($\star$) and Lemma~\ref{testreturn1-tar}, $\valid(v)'$ returns~$\true$.	
\end{proof}

\begin{definition}~\label{t0t1-tar}
        For any value $v$,
	\begin{itemize}
		\item Let $t^v_0$ be the max invocation time of any \validpg{(v)} by a {\ct} reader that returns $\false$ in $H$;
		if no such {\validpg{(v)}} exists, $t^v_0 = 0$.
		\item Let $t^v_1$ be the min response time of any {\validpg{(v)}} by a {\ct} reader that returns $\true$ in $H$;
		if no such {\validpg{(v)}} exists, $t^v_1 = \infty$.
	\end{itemize}
\end{definition}

\begin{lemma}\label{iexists-tar}
      For any value $v$, $t^v_1 > t^v_0$ and so the interval $(t^v_0,t^v_1)$ is not empty.
\end{lemma}

\begin{proof}
    Let $v$ be any value.
    There are four cases:
    \begin{itemize}
         \item Case 1: no {\validpg{(v)}} by a {\ct} process returns $\false$,
            and no \linebreak {\validpg{(v)}} by a {\ct} process returns $\true$.
            Then $t^v_0 = 0$ and $t^v_1 = \infty$.
            So $t^v_1 > t^v_0$.
        \item Case 2: no {\validpg{(v)}} by a {\ct} process returns $\false$ and some \linebreak {\validpg{(v)}} by a {\ct} process returns $\true$.
            Then $t^v_0 = 0$ and  $t^v_1 > 0$.
            So $t^v_1 > t^v_0$.
        \item Case 3: some {\validpg{(v)}} by a {\ct} process returns $\false$
            and
            no \linebreak {\validpg{(v)}} by a {\ct} process returns $\true$.
            Then  $t^v_0 < \infty$ and $t^v_1 = \infty$.
            So $t^v_1 > t^v_0$.
            
        \item Case 4: some {\validpg{(v)}} by a {\ct} process returns $\false$
        and
        some \linebreak {\validpg{(v)}} by a {\ct} process returns $\true$.
            Let $\valid(v)'$ be the {\validp{(v)}} with the max invocation time of any {\validp{(v)}} by a {\ct} process that returns $\false$. 
            By Definition~\ref{t0t1-tar}, the invocation time of $\valid(v)'$ is $t^v_0$~($\star$).
            Let $\valid(v)''$ be the {\validp{(v)}} with the min response time of any {\validp{(v)}} by a {\ct} process that returns $\true$.
            By Definition~\ref{t0t1-tar}, the response time of $\valid(v)''$ is $t^v_1$ ($\star\star$).
                         
            Since $\valid(v)''$ returns $\true$ and $\valid(v)'$ returns $\false$, 
                by Lemma~\ref{testatestb-tar},
                $\valid(v)''$ does not precede $\valid(v)'$.
            So the response time of $\valid(v)''$ is greater than the invocation time of $\valid(v)'$.
            Thus, by ($\star$) and ($\star\star$),
                 $t^v_1 > t^v_0$.
    \end{itemize}
    In all cases,  $t^v_1 > t^v_0$ and so the interval $(t^v_0,t^v_1)$ is not empty.
\end{proof}

By Lemma~\ref{v0alwaystrue},
    no {\validpg{(v_0)}} executed by a correct process returns \false. So
by Definition~\ref{t0t1-tar}, we have:

\begin{observation}\label{v00}
$t^{v_{0}}_0 = 0$.
\end{observation}

\begin{lemma}\label{t0Readintervalexists}
If a $\Test$ operation by a correct reader returns $v$ at time $t'$ in $H$,
    then $t' > t_0^v$.
\end{lemma}

\begin{proof}
Suppose a $\Test$ operation by a correct reader $p_k$ returns $v$ at time $t'$ in $H$.
There are two cases:
\begin{itemize}
    \item Case 1: $v=v_0$. Then by Observation~\ref{v00}, $t^{v_0}_0 = 0$.
    So it is clear that $t' >t^{v_0}_0$.
        
    \item Case 2: $v \ne v_0$. 
    Then according to the code of the $\Test$ procedure,
        $p_k$ returns $v$ at line~\ref{checksign-tar} of this $\Test$ operation at time $t'$.
    So $p_k$ finds the condition of line~\ref{checksign-tar} holds in this $\Test$ operation before $t'$.
    By lines~\ref{verifyp-tar}-\ref{checksign-tar},
        $p_k$ executes the $\valid(v)$ procedure at line~\ref{verifyp-tar} of this $\Test$ operation
        and this execution returns $\true$ at some time $t'' < t'$.
    By Definition~\ref{t0t1-tar},
        $t'' > t_1^v$.
    Then by Lemma~\ref{iexists-tar},
        $t'' > t_0^v$.
    Thus, since $t'' < t'$,
        $t' > t_0^v$.
\end{itemize}
\end{proof}

Recall that $H$ is an arbitrary history of the implementation given by Algorithm~\ref{code-tar}, and
    {\hct} is the set of processes that are correct in $H$.
To show that $H$ is Byzantine linearizable with respect to a SWMR {\tar},
    we must prove that there is a history $H'$ such that (1) $H'|{\hct} =H|{\hct}$ and 
(2)~$H'$ is linearizable with respect to a SWMR \tar.

We construct $H'$ from $H$ as follows.
We start from $H|{\hct}$.
    Since the writer is faulty in $H$,
    $H|{\hct}$ \emph{does not contain any operation by the writer};
    in contrast, $H|{\hct}$ contains all the $\Test$ and $\valid(-)$ operations by the \emph{correct} processes in $H$.
So to ``justify'' the responses of these $\Test$ and $\valid(-)$ operations,
    we must add some $\Set(-)$ operations by the writer to $H|{\hct}$.
And we must add them such that
    the resulting history conforms to the sequential specification of an {\tar}. 
To do so,
    for every value $v$ such that a \validpg{(v)} returns $\true$ in $H|{\hct}$,
    we add a $\Set(v)$ operation;
    this $\Set(v)$ operation is added such that its response time is \emph{after} the \emph{invocations} of \emph{all} the \validpg{(v)}s that return $\false$, and
        \emph{before} the \emph{responses} of \emph{all} the \validpg{(v)}s that return $\true$.
Then, to justify the $\Test$ operations,
    for every $\Test$ operation that returns some value $v$ in $H|{\hct}$,
      we also add a $\Set(v)$ operation such that its response step is \emph{after} the \emph{invocations} of \emph{all} the \validpg{(v)}s that return $\false$ and
      ``just before'' the response step of this $\Test$ operation in $H|{\hct}$.
    We describe this construction more precisely below.

\begin{definition}\label{steps-tar-2}
Let $H'$ be the history constructed from $H$ as follows.
\begin{itemize}

	\item\label{s1} Step 1: $H' = H|{\hct}$.

    \item\label{addsign} 
           Step 2: For every value $v$ such that a {\validpg{(v)}} returns $\true$ in the history $H'$ constructed in Step 1,
		we add to $H'$ a $\Set(v)$ operation such that
            the time $t$ of the response step of this $\Set(v)$ operation
            is within the interval $(t^v_0,t^v_1)$.

            Note that, by Lemma~\ref{iexists-tar}, the interval $(t^v_0,t^v_1)$ is not empty, so it is possible to insert
    the response step of the $\Set(v)$ operation
    inside this interval.

 \item\label{addset} 
            Step 3: For each $\Test$ operation that returns some value $v$ in the history $H'$ constructed in Step 2,
		  we add to $H'$ a $\Set(v)$ operation such that
            (a) the time $t$ of the response step of this $\Set(v)$ operation is
            after $t_0^v$ and
            before the time $t'$ of the response step of the $\Test$ operation,
            and
            (b) there are no steps between $t$ and $t'$.

Note that, by Lemma~\ref{t0Readintervalexists}, the interval $(t^v_0,t')$ is not empty, so it is indeed possible to insert
    the response step of the $\Set(v)$ operation
    inside this interval.

\end{itemize}

\end{definition}
    Note that in the above construction we must ensure that all the $\Set(-)$ operations that we add to $H'$ are sequential (because they are ``executed'' by the same process, namely the writer). We can do so by ``shrinking'' the $[invocation,response]$ interval of every added $\Set(-)$ operation as much as necessary to achieve this.
    
\color{black}
\begin{observation}\label{cc-tar-2}
    $H'|{\hct} = H|{\hct}$.
\end{observation}

\begin{observation}\label{validfromcorrect-tar-2}
   A $\Test$ or $\valid(-)$ operation is in $H'$ if and only if it is in $H|{\hct}$.
\end{observation}

By Lemma~\ref{v0alwaystrue} and Observation~\ref{validfromcorrect-tar-2},
    we have the following.
\begin{observation}\label{v0alwaystrue2}
   In $H'$, every $\valid(v_0)$ operation returns $\true$.
\end{observation}

By Definition~\ref{t0t1-tar} and Observation~\ref{validfromcorrect-tar-2}, 
\begin{observation}\label{sameip-tar-2}
        In $H'$, for any value $v$,
	\begin{itemize}
		\item The max invocation time of any {\validpg{(v)}} that returns $\false$ is $t^v_0$.
		\item The min response time of any {\validpg{(v)}} that returns~$\true$ is $t^v_1$.
           
	\end{itemize}
\end{observation}

\begin{observation}\label{addwrite-tar-2}
   A $\Set(-)$ operation is in $H'$ if and only if it is added in Step 2 or Step 3 of the construction of $H'$ (Definition~\ref{steps-tar-2}).
\end{observation}

We now prove that the constructed history $H'$ is linearizable with respect to an \tar.
First, we define the \emph{linearization point} of each operation in $H'$ as follows.

\begin{definition}~\label{testsetlinearization-tar-2}
Let $o$ be an operation in $H'$.
    \begin{itemize}

	\item If $o$ is a $\Set(-)$ operation,
	   the linearization point of $o$ is the time of the response step of~$o$.
       \item  If $o$ is a $\Test$ operation,
	   the linearization point of $o$ is the time of the response step of $o$.
       
	\item  If $o$ is a $\valid(-)$ operation,
        then: \begin{itemize}
          \item If $o$ returns $\true$,
	   the linearization point of $o$ is the time of the response step of $o$.
            \item If $o$ returns $\false$,
	   the linearization point of $o$ is the time of the invocation step of~$o$.
        \end{itemize}

\end{itemize}
\end{definition}

\begin{observation}\label{precedes-tar-2}
    If an operation $o$ precedes an operation $o'$ in $H'$,
        then the linearization point of $o$ is before the linearization point of $o'$.
\end{observation}

We now use the linearization points of all the operations in $H'$ to define the following linearization $L$ of $H'$:
\begin{definition}\label{hlp-tar-2}
	Let $L$ be a \emph{sequence} of operations such that:
	\begin{enumerate}
		\item\label{uno} An operation $o$ is in $L$ if and only if it is in $H'$.
		\item\label{due} An operation $o$ precedes an operation $o'$ in $L$ if and only if
     the linearization point of~$o$ is before the linearization point of $o'$ in $H'$.
	\end{enumerate}
\end{definition}

By Observation~\ref{precedes-tar-2} and Definition~\ref{hlp-tar-2}, 
    $L$ respects the precedence relation between the operations of $H'$.
More precisely:
\begin{observation}\label{precedes-l-tar-2}
    If an operation $o$ precedes an operation $o'$ in $H'$,
        then $o$ precedes $o'$ in~$L$.
\end{observation}

In the following lemmas and theorems we prove that the linearization $L$ of $H'$ conforms to the sequential specification of a SWMR \tar.

\begin{theorem}\label{rw-tar-2}
    In $L$,
        if a $\Test$ operation returns a value $v$, then:
        \begin{itemize}
          \item either some $\Set(v)$ operation precedes it and this $\Set(v)$ operation is the last $\Set(-)$ operation that precedes it,
          
            \item or $v = v_0$ (the initial value of the register) and no $\Set(-)$ operation {\precedes} it.
        \end{itemize}
       
\end{theorem}

\begin{proof}
    Suppose a $\Test$ operation returns a value $v$ in $L$.
    By Definition~\ref{hlp-tar-2}(1),
        this $\Test$ operation is also in $H'$.
    By Step 3 of Definition~\ref{steps-tar-2},
    $H'$ has a $\Set(v)$ operation such that the time $t$ of the response step of this $\Set(v)$ operation is before the time $t'$ of the response step of the $\Test$ operation and there are no steps between $t$ and $t'$.
    So by Definition~\ref{testsetlinearization-tar-2},
            the linearization point $t$ of this $\Set(v)$ operation is before the linearization point $t'$ of the $\Test$ operation, and 
            there is no other $\Set(-)$ operation with a linearization point between $t$ and $t'$.     
    Thus, by Definition~\ref{hlp-tar-2}(2),
        this $\Set(v)$ operation precedes the $\Test$ operation in $L$ and it is the last $\Set(-)$ operation that precedes the $\Test$ operation~in~$L$.
\end{proof}

\begin{lemma}\label{vs1}
        In $L$,
            if a $\valid(v)$ operation returns~$\false$,
		then there is no $\Set(v)$ operation that {\precedes} this $\valid(v)$ operation and~\mbox{$v\ne v_0$.}
\end{lemma}
\begin{proof}
	Assume for contradiction that, 
		in $L$,
            (1) a $\valid(v)$ operation returns $\false$ and
		(2) there is a $\Set(v)$ operation that {\precedes} this $\valid(v)$ operation or $v=v_0$.
        There are two cases:
        \begin{itemize}
            \item Case 1: $v=v_0$.
             By Definition~\ref{hlp-tar-2}(1), 
                the $\valid(v)$ operation is also in $H'$.
            Since $v=v_0$, 
                by Observation~\ref{v0alwaystrue2},
                this $\valid(v)$ operation returns $\true$ in $H'$.
                Therefore it also returns $\true$ in $L$ --- a contradiction to (1).

            \item Case 2: $v\ne v_0$.
            Then by (2), 
                there is a $\Set(v)$ operation that {\precedes} this $\valid(v)$ operation in $L$. 
            By Definition~\ref{hlp-tar-2}(1), 
            the $\Set(v)$ operation and the $\valid(v)$ operations are also in $H'$, 
            and the linearization point $t$ of the $\Set(v)$ operation is before the linearization point $t'$ of the $\valid(v)$ operation  in $H'$.
        So $t < t'$ ($\star$).
        
        Since this $\valid(v)$ operation returns $\false$:
            (a) by Definition~\ref{testsetlinearization-tar-2}, 
            $t'$ is the time of the invocation step of the $\valid(v)$ operation in $H'$,
            and
            (b)~by Observation~\ref{sameip-tar-2},
            $t' \le t^v_0$.
        By Definition~\ref{steps-tar-2},
            $t$ is the time of the response step of the $\Set(v)$ operation.
        By Observation~\ref{addwrite-tar-2} and Steps 2 and 3 of Definition~\ref{steps-tar-2},
            $t > t_0^v$.
        So $t' <t$ --- a contradiction to~($\star$).
        \end{itemize}
\end{proof}

\begin{lemma}\label{vs2}
	In $L$, 
	if a $\valid(v)$ operation returns $\true$,
	then there is a $\Set(v)$ operation that {\precedes} this $\valid(v)$ operation or $v=v_0$.
\end{lemma}

\begin{proof}
	Suppose in $L$, a $\valid(v)$ operation returns $\true$.
        By the Definition~\ref{hlp-tar-2}(1) of $L$,
            this $\valid(v)$ operation is also in $H'$.
         Since the $\valid(v)$ operation returns $\true$,
            by Definition~\ref{testsetlinearization-tar-2},
            its linearization point $t$ is the time of its response step in $H'$.
         By Observation~\ref{sameip-tar-2},
            $t \ge t^v_1$ ($\star$).

        Since the $\valid(v)$ operation returns $\true$ in $H'$,
            by Step 2 of Definition~\ref{steps-tar-2},
             $H'$ has a $\Set(v)$ operation such that the time of the response step of this $\Set(v)$ operation is within the interval $(t^v_0,t^v_1)$ in $H'$.
        By Definition~\ref{testsetlinearization-tar-2}, 
            the linearization point $t'$ of this $\Set(v)$ operation is in the interval $(t^v_0,t^v_1)$.
        So $t' < t^v_1$.
        By ($\star$),
            $t' < t$. 
        Thus, by the Definition~\ref{hlp-tar-2}(2) of $L$, this $\Set(v)$ operation {\precedes} the $\valid(v)$ operation in $L$.
\end{proof}

By Lemma~\ref{vs1} and Lemma~\ref{vs2}, we have the following theorem.
\begin{theorem}\label{sv-tar-2}
    In $L$, 
	a $\valid(v)$ operation returns $\true$ if and only if
	there is a $\Set(v)$ operation that {\precedes} this $\valid(v)$ operation or $v=v_0$.
\end{theorem}

\begin{theorem}\label{linearthe-tar-case2}
    If the writer is not correct in $H$,
    the history $H$ is Byzantine linearizable with respect to a SWMR {\tar}.
\end{theorem}

\begin{proof}
Assume that the writer $p_1$ is not correct in $H$.
    To prove that $H$ is
        Byzantine linearizable with respect to a SWMR {\tar} (Definition~\ref{def-tar}),
    we must prove that there is a history $H'$ such that:
    \begin{compactenum}[(1)]
        \item $H'|{\hct} = H|{\hct}$.
        \item $H'$ is linearizable with respect to a SWMR {\tar}. 
    \end{compactenum}
    
     Let $H'$ be the history given by Definition~\ref{steps-tar-2}.
    By Observation~\ref{cc-tar-2}, $H'$ satisfies (1).
    
    Let $L$ be the linearization of $H'$ given in Definition~\ref{hlp-tar-2}.
    \begin{compactenum}[(a)]
        \item By Observation~\ref{precedes-l-tar-2}, $L$ respects the precedence relation between the operations of $H'$.
        \item By Theorems~\ref{rw-tar-2} and~\ref{sv-tar-2}, $L$ conforms to the sequential specification of a SWMR {\tar}.
        \end{compactenum}
    Thus, $H'$ also satisfies~(2).
\end{proof}

By Theorems~\ref{linearthe-tar-case1} and~\ref{linearthe-tar-case2}, the arbitrary history \( H \) of Algorithm~\ref{code-tar} is Byzantine linearizable with respect to an {\tar} (independent of whether the writer $p_1$ is correct or not in $H$):

\begin{restatable}{theorem}{verifybl}\label{linearthe-tar}
    [\textsc{Byzantine linerizability}] The history $H$ is Byzantine linearizable with respect to a SWMR {\tar}.
\end{restatable}

By Theorem~\ref{termination-tar} and Theorem~\ref{linearthe-tar},
    we have the following:
    
\begin{restatable}{theorem}{tarcorrect}
In a system with $n > 3f$ processes,
    where $f$ processes can be Byzantine,
    Algorithm~\ref{code-tar} is a correct implementation of a SWMR {\tar}.
\end{restatable}

\section{Correctness of the {\src} Implementation}\label{a-sar}
In the following, we prove that Algorithm~\ref{code-sar} is a
    correct implementation of a SWMR {\sr} for systems 
    with $n >3f$ processes, $f$ of which can be Byzantine.
To do so, in the following, we consider an arbitrary (infinite) history $H$ of Algorithm~\ref{code-sar} (where $n>3f$), and prove:
    \begin{compactitem}
        \item \textsc{[Byzantine linearizability]} The history $H$ is Byzantine linearizable with respect to a SWMR {\sr} (Section~\ref{linear-sar}).
        \item \textsc{[Termination]} All processes that are correct in $H$ complete all their operations (Section \ref{wait-free-sar}).
    \end{compactitem}
    
\smallskip
\noindent
Recall that a sticky register is initialized to $\bot$,
    and the writer only writes non-$\bot$ values.

\noindent
\textbf{Notation.} For convenience, henceforth:
    \begin{compactitem}

    \item $v$ denotes a \emph{non-$\bot$ value}.
When we refer to the initial~value~$\bot$ of the sticky register, we do so explicitly. 

    \item When we say that a process $p$ is in $\setv$, $\setv$ \emph{contains} a process $p$,
        or $p \in \setv$, we mean $\langle -,p\rangle \in \setv$.

    \end{compactitem}

\subsection{Termination}\label{wait-free-sar}

\begin{observation}\label{correctonlyE-sar}
For every correct process $p_i\in \{p_1,\ldots,p_n\}$,
	if $E_i = v$ at time $t$, then $E_i = v$ at all times $t' \ge t$.
\end{observation}

\begin{observation}\label{correctonlywrites1-sar}

For every correct process $p_i\in \{p_1,\ldots,p_n\}$,
	if $R_i = v$ at time $t$, then $R_i = v$ at all times $t' \ge t$.
\end{observation}

\begin{observation}\label{monock-sar}
	For every correct reader $p_k\in \{p_2,\ldots,p_n \}$, 
		the value of $C_k$ is non-decreasing.
\end{observation}

\begin{observation}\label{monoset1-sar}
	Let $\Test$ be an operation by a correct reader $p_k$.
	In $\Test$,
	 $|\setv|$ of $p_k$ is non-decreasing.
\end{observation}

\begin{lemma}\label{firstv-sar}
	For any value $v$,
		if $p_i$ is the first correct process with $R_i = v$,
            then $p_i$ writes $v$ in $R_i$ at line~\ref{echosupport-sar}.
\end{lemma}
    \begin{proof}
    Assume for contradiction, $p_i$ is the first correct process with $R_i = v$ for some value $v$, and $p_i$ does not write $v$ in $R_i$ at line~\ref{echosupport-sar}.
    Then according to the code of the $\fresh()$ procedure, 
        $p_i$ writes $v$ in $R_i$ at line~\ref{followcondition-sar}, say at time $t$.
    This implies that $p_i$ finds that the condition $|\{ r_i~|~ r_i =v \}| \ge f+1$ holds in line~\ref{followcondition-sar} before time $t$.
    Since there are at most $f$ faulty processes,
        there is at least one correct process $p_j\ne p_i$ such that $p_i$ finds $r_j=v$.
    So $p_i$ reads $v$ from $R_j$ at line~\ref{sets-sar} before time $t$.
    Since $p_j$ is correct, 
        this contradicts the fact that $p_i$ is the first correct process with $R_i = v$.
    Therefore, $p_i$ writes $v$ in $R_i$ at line~\ref{echosupport-sar}.
    \end{proof}

\begin{lemma}\label{wbeforer-sar}
    Suppose the writer $p_1$ is correct.
    For any value $v$,
        if a correct process $p_i$ has $R_i=v$ at time $t$,
        then $p_1$ wrote $v$ in $E_1$ at line~\ref{setter1-sar} in a $\Set(v)$ operation before time $t$.
\end{lemma}
 \begin{proof}
        Suppose the writer $p_1$ is correct.
        Assume a correct process $p_i$ has $R_i=v$ at time $t$.
        Let $p_a$ be the first correct process with $R_a = v$, say at time $t_a\le t$.
        By Lemma~\ref{firstv-sar},
            $p_a$ writes $v$ in $R_a$ at line~\ref{echosupport-sar} at time $t_a\le t$.
         So there are at least $n-f$ processes $p_i$ such that $p_a$ reads $E_i=v$ at line~\ref{readecho-sar} before time $t_a$.
        Since there are at most $f$ faulty processes,
            there are at least $n-2f$ correct processes $p_i$ that have $E_i=v$ before time $t_a$.
        Since $n>3f$,
            $n-2f > f \ge 1$.
        So at least one correct process $p_i$ has $ E_i=v$ before time $t_a$.
        According to the code, 
        this implies that the correct process $p_i$ reads $E_1=v$ at line~\ref{reade1-sar} and then writes $v$ in $E_i$ at line~\ref{echo-sar} before time $t_a$.
        Since the writer $p_1$ is correct,
            $p_1$ wrote $v$ in $E_1$ at line~\ref{setter1-sar} in a $\Set(v)$ operation before time $t_a$.
        Since $t_a \le t$,
            $p_1$ wrote $v$ in $E_1$ at line~\ref{setter1-sar} in a $\Set(v)$ operation before time $t$.
    \end{proof}
	
\begin{lemma}\label{unqiuev-sar}
	For any correct processes $p_i$ and $p_j$,
		if $R_i = v$ at any time and $R_j = v'$ at any time,
            then $v=v'$.
\end{lemma}
\begin{proof}
    Assume for contradiction, a correct process $p_i$,
		has $R_i = v$ at some time and 
            a correct process $p_j$ has $R_j = v'$ at some time,
            and $v\ne v'$.
    Let $p_a$ be the \emph{first} correct process with $R_a= v$, say at time~$t_a$.
    Let $p_b$ be the \emph{first} correct process with $R_b = v'$, say at time $t_b$.
    Since $v\ne v'$, by Observation~\ref{correctonlywrites1-sar}, $p_a\ne p_b$.
    By Lemma~\ref{firstv-sar},
        $p_a$ writes $v$ in $R_a$ at line~\ref{echosupport-sar} at time $t_a$ 
            and $p_b$ writes $v'$ in $R_b$ also at line~\ref{echosupport-sar} at time $t_b$.
    So there is a set $S_v$ of at least $n-f$ processes $p_i$ such that $p_a$ reads $E_i = v$ at line~\ref{readecho-sar},
    and there is a set $S_{v'}$ of at least $n-f$ processes $p_i$ such that $p_b$ reads $E_i=v'$ at line~\ref{readecho-sar}.
    Since $n>3f$,
        this implies that $S_v \cap S_{v'}$ contains at least one correct process $p_j$,
        and $p_j$ has $E_j=v$ and $E_j = v'$ --- a contradiction to Observation~\ref{correctonlyE-sar}.
\end{proof}

\begin{lemma}\label{termination-sar-1}
	Every $\Set(-)$ operation by a {\ct} writer completes.
\end{lemma}

\begin{proof}
    Suppose the writer $p_1$ is {\ct}.
    Let $\Set(v)$ be an operation by $p_1$.
    \begin{compactitem}
        \item Case 1: $p_1$ returns at line~\ref{stick-sar}.
            Then $\Set(v)$ completes.
        \item Case 2: $p_1$ does not return at line~\ref{stick-sar}.
             According to the code of the $\Set()$ procedure,
              $p_1$ writes $v$ into $E_1$ at line~\ref{setter1-sar} in the $\Set(v)$, say at time $t$.
            By Observation~\ref{correctonlyE-sar},
                $E_1=v$ at all times $t'\ge t$.
            And by line~\ref{stick-sar},
                $E_1=\bot$ at all times $t'<t$.
            By lines~\ref{helpecho-sar}-\ref{echo-sar},
                eventually all the $n-f$ correct processes $p_j$ have $E_j = v$.
            By Observation~\ref{correctonlyE-sar} and line~\ref{echosupport-sar},
                all correct processes $p_j$ eventually have $R_j \ne \bot$.
            By Lemma~\ref{unqiuev-sar} and Observation~\ref{correctonlywrites1-sar},
                there is a time $t_v$ and a value $v'$ such that all the $n-f$ correct processes $p_j$ have $R_j = v'$ at all times $t' >t_v$.
            By Lemma~\ref{wbeforer-sar} and Observation~\ref{correctonlyE-sar},
                $v'=v$.
            So all the $n-f$ correct processes $p_j$ have $R_j = v$ at all times $t' >t_v$.
            Since $p_1$ executes line~\ref{repeat-sar}-\ref{until-sar} infinitely often,
                $p_1$ eventually finds that the condition $|\{ r_i ~|~ r_i = v\}| \ge n-f$ holds in line~\ref{vconfirm-sar} and so returns at line~\ref{wr-sar}. 
            So $\Set(v)$ completes.
    \end{compactitem}
\end{proof}

	\begin{lemma}\label{ri1beforeset1-sar}
		Let $\Test$ be an operation by a correct reader $p_k$.
		In this $\Test$, if $\langle v,p_i\rangle$ is in $\setv$ of $p_k$ at time $t$ and $p_i$ is a correct process,
			then $R_i=v$ at all times $t' \ge t$.
		\end{lemma}
		\begin{proof}		
			Let $\Test$ be an operation by a correct reader $p_k$.
			Consider any time $t$ such that $\langle v,p_i\rangle$ is in $\setv$ of $p_k$ at time $t$ and $p_i$ is a correct process.  
			Consider the iteration of the while loop at line~\ref{whileloop-sar} of $\Test$ in which 
				$p_k$ inserts $\langle v,p_i\rangle$ into $\setv$ at line~\ref{setv-sar},
				say at time $t_k^1 \le t$.
			By line~\ref{check1-sar}, 
				$p_k$ reads $\langle v, - \rangle$ with $v\ne\bot$ from $R_{ik}$ at line~\ref{readri-sar},
				say at time $t_k^0 < t_k^1$.
			Since $p_i$ is correct and $R_{ik}$ is initialized to $\langle \bot, - \rangle$, 
				$p_i$ wrote $\langle v, - \rangle$ into $R_{ik}$ at line~\ref{fresh1-sar} by time $t_k^0$.
			By lines~\ref{rj-sar}-\ref{fresh1-sar},
				$R_i=v$ by time $t_k^0$.
			Since $t_k^0 < t_k^1$ and $t_k^1 \le t$,
				$t_k^0 < t$.
			Since $p_i$ is correct,
				by Observation~\ref{correctonlywrites1-sar},
				$R_i=v$ at all times $t' \ge t$.
		\end{proof}
        
By Lemma~\ref{unqiuev-sar} and Lemma~\ref{ri1beforeset1-sar},
    we have the following corollary.
\begin{corollary}\label{unqiuesetv-sar}
        Let $\Test$ be an operation by a correct reader $p_k$.
	In this $\Test$, if $\langle v, p_i \rangle$ and $\langle v', p_j \rangle$ are in $\setv$ of $p_k$ in this $\Test$ and $p_i$ and $p_j$ are both correct,
            then $v=v'$.
\end{corollary}

	\begin{lemma}\label{asker1-sar}
        Suppose for some value $v$,
		there is a time $t$ such that
		at least $f+1$ processes $p_i$ have $R_i = v$ at all times $t' \ge t$.
		Consider any iteration of the while loop of the $\fresh()$ procedure executed by a correct process $p_j$.
		If $p_j$ inserts $p_k$ into $askers$ in line~\ref{askers-sar} at some time $t_a\ge t$, 
		then $p_j$ writes $\langle v , - \rangle$ into $R_{jk}$ at line~\ref{fresh1-sar}.
		\end{lemma} 
	
	\begin{proof}
         Suppose for some value $v$,
		there is a time $t$ such that
		at least $f+1$ processes $p_i$ have $R_i = v$ at all times $t' \ge t$.
		Consider any iteration of the while loop of the $\fresh()$ procedure executed by a correct process $p_j$.
		Suppose $p_j$ inserts $p_k$ into $askers$ in line~\ref{askers-sar} at some time $t_a\ge t$.
            Then when $p_j$ executes line~\ref{ifnotone-sar} at time $t_1 > t_a \ge t$,
                there are two cases:
            \begin{compactitem}
                \item Case 1: $R_j=\bot$.
                    Then in line~\ref{sets-sar}, $p_j$ reads every $R_i, 1\le i \le n$ after time $t$.
		Since at least $f+1$ processes $p_i$ have $R_i=v$ at all times $t' \ge t$,
			there are at least $f+1$ processes $p_i$ such that $p_j$ reads $R_i=v$ at line~\ref{sets-sar} after time $t$.
            So $p_j$ finds that the condition $|\{r_i|r_i = v\}|\ge f+1$ of line~\ref{followcondition-sar} holds for value $v$.
            \begin{claim}
                There is no other value $v'\ne v$ such that $p_j$ finds that the condition $|\{r_i|r_i = v'\}|\ge f+1$ of line~\ref{followcondition-sar} holds for value $v'$.
            \end{claim}
            \begin{proof}
                Since at least $f+1$ processes $p_i$ have $R_i=v$ at all times $t' \ge t$,
                    and there are at most $f$ faulty processes,
                    there is at least one correct process $p_i$ that has $R_i=v$ at all times $t' \ge t$.
                Then by Lemma~\ref{unqiuev-sar},
                    no correct process has $R_i=v'$ for any value $v'\ne v$ at any time.
                Thus, since there are at most $f$ faulty processes,
                    there is no value $v'\ne v$ such that $p_j$ reads $v'$ from $R_i$ for at least $f+1$ processes $p_i$ at line~\ref{sets-sar}.
                So there is no value $v'\ne v$ such that $p_j$ finds that the condition $|\{r_i|r_i = v'\}|\ge f+1$ of line~\ref{followcondition-sar} holds for value $v'$.
            \end{proof}
		 The above claim implies that $p_j$ finds that the condition $|\{r_i|r_i = v\}|\ge f+1$ of line~\ref{followcondition-sar} holds only for the value $v$ and so $p_j$ writes $v$ into $R_j$ at line~\ref{followcondition-sar}.
                \item Case 2: $R_j\ne \bot$.
                Since at least $f+1$ processes $p_i$ have $R_i=v$ at all times $t' \ge t$ and there are at most $f$ faulty processes,
                    at least one correct process $p_i$ has $R_i=v$ at all times $t' \ge t$.
                Since $R_j\ne \bot$ at time $t_1 > t$ and $p_j$ is correct, by Lemma~\ref{unqiuev-sar}, 
                    $R_j=v$ at time $t_1$.
            \end{compactitem}
		Therefore, in both cases, $p_j$ writes $\langle v , - \rangle$ into $R_{jk}$ at line~\ref{fresh1-sar}.
	\end{proof}

\begin{lemma}\label{onecorrectset0-sar}
	Let $\Test$ be an operation by a correct reader $p_k$.
	For each iteration of the while loop at line~\ref{whileloop-sar} of $\Test$, 
		the following loop invariants hold at line~\ref{whileloop-sar}:
	\begin{enumerate}[(1)]
            \item\label{inv0} $\setb$ and $\setv$ contain different processes. 
            \item\label{inv1} $\nexists v$ such that $|\{p_j|\langle v,p_j\rangle\in \setv\}| \ge n-f$, and $|\setb| \le f$.

		\item\label{inv2}  If there are at least $f+1$ correct processes in $\setv$,
		then there is no correct process in $\setb$.
	\end{enumerate}
\end{lemma}

\begin{proof}
	Let $\Test$ be an operation by a correct reader $p_k$.
    We now prove the invariants by induction on the number of iterations of the while loop at line~\ref{whileloop-sar} of $\Test$.

		\begin{itemize}
		\item Base Case: 
			Consider the first iteration of the while loop at line~\ref{whileloop-sar} of $\Test$.
			Since $p_k$ initializes $\setb$ and $\setv$ to $\emptyset$,
                    (\ref{inv0}) holds trivially.
        Moreover, $\forall v$ $|\{p_j|\langle v,p_j\rangle\in \setv\}|= 0$ and $|\setb| = 0$.
                
			Since $n \ge 3f+1$ and $f \ge 0$, $n-f > 0$.
                Thus (\ref{inv1}) holds.
			Since $\setb$ is empty,
			(\ref{inv2}) holds trivially.
	
		\item Inductive Case:
			Consider any iteration $I$ of the while loop at line~\ref{whileloop-sar} of $\Test$.
			Assume that (\ref{inv0}), (\ref{inv1}), and (\ref{inv2}) hold at the beginning of $I$.
			
			If $p_k$ does not find that the condition of line~\ref{until-sar} holds in $I$,
                    $p_k$ does not move to the next iteration of the while loop,
				and so (\ref{inv0}), (\ref{inv1}), and (\ref{inv2}) trivially hold
				at the start of the ``next iteration'' (since it does not exist).
			Furthermore, if $\Test$ returns at line~\ref{return1-sar-1} or line~\ref{return0-sar} in $I$,
				(\ref{inv0}), (\ref{inv1}), and (\ref{inv2}) trivially hold
				at the start of the ``next iteration'' (since it does not exist).

			We now consider any iteration $I$ of the while loop at line~\ref{whileloop-sar} of $\Test$
				in which $p_k$ finds that the condition of line~\ref{until-sar} holds 
				and $\Test$ does not return at line~\ref{return1-sar-1} or line~\ref{return0-sar}.
                We first show that (\ref{inv0}) remains true at the end of $I$.
                Since $p_k$ finds that the condition of line~\ref{until-sar} holds,
                     $p_k$ inserts a process, say $p_j$, into $\setb$ or inserts $\langle -,p_j\rangle$ into $\setv$ in $I$.
                By line~\ref{until-sar} and line~\ref{BS}, $p_j \not\in \setb$ and $\langle -,p_j\rangle \not\in \setv$.
                    Note that $p_j$ or $\langle -,p_j\rangle$ is the only process (or tuple) that $p_k$ inserts into $\setb$ or $\setv$ in $I$.
                So since (\ref{inv0}) holds at the beginning of $I$, and $p_j \not\in \setb$ and $\langle -,p_j\rangle \not\in \setv$,
                    (\ref{inv0}) holds at the end of $I$.
                    
			We now show that (\ref{inv1}) remains true at the end of $I$.
			Since $\Test$ does not return at line~\ref{return1-sar-1} or line~\ref{return0-sar},
				$p_k$ finds that the condition "$\exists v$ such that $|\{p_j|\langle v,p_j\rangle\in \setv\}|\ge n-f$" does \emph{not} hold at line~\ref{return1-sar} and
                the condition $|\setb| > f$ also does \emph{not} hold
                at line~\ref{return0-sar}.
                So the condition "$\nexists v$ such that $|\{p_j|\langle v,p_j\rangle\in \setv\}|\ge n-f$" holds at line~\ref{return1-sar} and the condition "$|\setb| \le f$" holds
                at line~\ref{return0-sar}.
               
			Since $p_k$ does not change $\setv$ after line~\ref{return1-sar} and does not change $\setb$ after line~\ref{return0-sar} in $I$,
            it is now clear that (\ref{inv1}) holds at the end of $I$.

			We now show that (\ref{inv2}) remains true at the end of  $I$.
			There are two cases:
			\begin{enumerate}
				\item Case 1: $p_k$ executes line~\ref{setv-sar} in $I$.
					Then $p_k$ changes $\setb$ to $\emptyset$ at line~\ref{empty0-sar}.
					Since $p_k$ does not change $\setb$ after line~\ref{empty0-sar} in $I$,
						$\setb$ remains $\emptyset$ at the end of $I$.
					So (\ref{inv2}) holds trivially at the end of $I$.

			\item Case 2: $p_k$ does not execute line~\ref{setv-sar} in $I$.
                    So $p_k$ does not change $\setv$ in $I$.
					Since $p_k$ finds that the condition of line~\ref{until-sar} holds in $I$,
						$p_k$ executes line~\ref{set0-sar} in $I$.
					Let $p_a$ be the process that $p_k$ inserts into $\setb$ in line~\ref{set0-sar} in $I$.
					There are two cases:
					\begin{itemize}
					\item Case 2.1: $p_a$ is faulty. 
						Then the number of correct processes in $\setb$ does not change in $I$.
                        So, since $\setv$ also does not change in $I$, (\ref{inv2}) remains true at the end of $I$.

					\item Case 2.2: $p_a$ is correct. 
					There are two cases:

					\begin{itemize}
					\item Case 2.2.1:  there are fewer than $f+1$ correct processes in $\setv$ at the beginning of $I$.
						Since $p_k$ does not change $\setv$ in $I$, 
						(\ref{inv2}) remains true at the end of $I$.

					\item Case 2.2.2:  there are at least $f+1$ correct processes in $\setv$ at the beginning of $I$.
						We now show that this case is impossible.
						
					Let $p_b$ be the last process that $p_k$ inserted into $\setv$.
						Note that this insertion occurred before $I$, 
							say at time $t_k^0$. 
						Since there are at least $f+1$ correct processes in $\setv$ at the beginning of $I$
							and $p_b$ is the last process that $p_k$ inserted into $\setv$ before $I$,
							there are at least $f+1$ correct processes $p_i$ in $\setv$ at time $t_k^0$.
                            By Corollary~\ref{unqiuesetv-sar},
                                there is a value $v\ne \bot$
                                such that for at least $f+1$ correct processes $p_i$,
                                $\langle v,p_i\rangle \in \setv$ at time $t_k^0$. 
						By Lemma~\ref{ri1beforeset1-sar},
							at least $f+1$ correct processes $p_i$ have $R_i =v$ at all times $t' \ge t_k^0$ ($\star$).

						Recall that $p_k$ inserts $p_a$ into $\setb$ in the iteration $I$.
						So in $I$,
							$p_k$ increments $C_k$ at line~\ref{ckplus-sar}, say at time $t_k^1$. 
                            Note that $t_k^1 > t_k^0$. 
						Let $c^*$ be the value of $C_k$ after $p_k$ increments $C_k$ at time~$t_k^1$.
						Since $C_k$ is initialized to 0,
							by Observation~\ref{monock-sar}, $c^* \ge 1$.
						By lines~\ref{readri-sar}, \ref{until-sar}, and~\ref{checkb-sar},
							$p_k$ reads $\langle \bot, c \rangle$ from $R_{ak}$ with $c \ge c^*$ at line~\ref{readri-sar} in $I$.
						Since $c\ge c^*\ge 1$ and $R_{ak}$ is initialized to $\langle -,0\rangle$,
							it must be that $p_a$ wrote $\langle \bot, c \rangle$ into $R_{ak}$ at line~\ref{fresh1-sar} ($\star\star$).

						Consider the iteration of the while loop of the $\fresh()$ procedure in which $p_a$ writes $\langle \bot, c \rangle$ into $R_{ak}$ at line~\ref{fresh1-sar}.
						Note that $c$ is the value that $p_a$ read from $C_k$ in line~\ref{collectck-sar} of this iteration;
							 say this read occurred at time $t_a^1$.
						Since $c \ge c^*$,
							by Observation~\ref{monock-sar},
							$t_a^1 \ge t_k^1$.
						Then $p_a$ inserts $p_k$ into $askers$ at line~\ref{askers-sar}, say at time $t_a^2 > t_a^1$.
						Since $t_a^1 \ge t_k^1$ and $t_k^1 > t_k^0$,
							$t_a^2 > t_k^0$.
						So $p_a$ inserts $p_k$ into $askers$ at line~\ref{askers-sar} after $t_k^0$.
						Thus, by ($\star$) and Lemma~\ref{asker1-sar},
							$p_a$ writes $\langle v , c \rangle$ into $R_{ak}$ at line~\ref{fresh1-sar} in this iteration.
					Since $v\ne\bot$,
                        this contradicts ($\star\star$).
						So this case is impossible.
					\end{itemize}
					\end{itemize}
				\end{enumerate}
				So in all the possible cases, we showed that (\ref{inv0}), (\ref{inv1}), and (\ref{inv2}) remain true at the end of the iteration.
			\end{itemize}
	\end{proof}

\begin{lemma}\label{correctoutside-sar}
    Let $\Test$ be an operation by a correct reader $p_k$.
    Every time when $p_k$ executes line~\ref{whileloop-sar} of this $\Test$ operation, 
	there is a correct process $p_i$ such that $p_i \not\in \setb$ and $\langle-,p_i\rangle\not\in\setv$.
\end{lemma}

			\begin{proof}
                    Let $\Test$ be an operation by a correct reader $p_k$.
				Suppose $p_k$ executes line~\ref{whileloop-sar} of this $\Test$ at time $t$.
				Consider $\setb$ and $\setv$ at time $t$.
				By Lemma~\ref{onecorrectset0-sar}(\ref{inv0}),
					$\setb$ and $\setv$ contain different processes.

				We now prove that $\setb \cup \setv$ contains fewer than $n-f$ correct processes; this immediately implies that there is a correct process $p_i$ such that $p_i \not\in \setb$ and $\langle-,p_i\rangle\not\in\setv$.
				There are two possible cases:
				
				\begin{enumerate}
				\item Case 1: $\setv$ contains at most $f$ correct processes.
				By Lemma~\ref{onecorrectset0-sar}(\ref{inv1}), $|\setb| \le f$.
				So $\setb \cup \setv$ contains at most $2f$ correct processes.
				Since $3f < n$, we have $2f < n-f$.

				\item Case 2: $\setv$ contains at least $f+1$ correct processes.
                By Lemma~\ref{onecorrectset0-sar}(\ref{inv2}), $\setb$ does \emph{not} contain any correct process.
                    Let $C$ be the set of all the correct processes in $\setb \cup \setv$.
                    Since $\setb$ does \emph{not} contain any correct process,
                    all the processes in $C$ are in $\setv$.
                    By Corollary~\ref{unqiuesetv-sar}, there must be a value $v$ such that
                    for all processes $p_i$ in $C$, $\langle v,p_i\rangle\in \setv$.
                   Thus, $\exists v$ such that $|\{p_j|\langle v,p_j\rangle\in \setv\}| \ge |C|$.
                   Since by Lemma~\ref{onecorrectset0-sar}(\ref{inv1}),  
                    $\nexists v$ such that $|\{p_j|\langle v,p_j\rangle\in \setv\}| \ge n-f$,
                    $|C| < n-f$.

    \end{enumerate}
				In both cases, $\setb \cup \setv$ contains fewer than $n-f$ correct processes.
				\end{proof}

\begin{lemma}\label{line7inf-sar}
	Let $\Test$ be an operation by a correct reader $p_k$.
	Every instance of the Repeat-Until loop at lines~\ref{repeat-sar}-\ref{until-sar} of this $\Test$ terminates.
\end{lemma}
\begin{proof}
	Let $\Test$ be an operation by a correct reader $p_k$.
	Assume for contradiction that
		there is an instance of the Repeat-Until loop at lines~\ref{repeat-sar}-\ref{until-sar} of $\Test$ that does not terminate.
        Let $t$ be the time when $p_k$ executes line~\ref{repeat-sar} for the first time in this instance of the Repeat-Until loop.
        Since this instance does not terminate, 
		$p_k$ never finds the condition of line~\ref{until-sar} holds after $t$ ($\star$). 
        
         Since lines~\ref{whileloop-sar}-\ref{repeat-sar} do not change $\setb$ and $\setv$,
        Lemma~\ref{correctoutside-sar} implies that 
		 there is a correct process $p_a$ such that $p_a \not\in \setb$ and $\langle-,p_a\rangle\not\in\setv$ at time $t$,
            i.e., $p_a \in S$ at time $t$.
         Let $c^*$ be the value of $C_k$ at time $t$;
    	by line~\ref{ckplus-sar},
    		$c^* \ge 1$.
        Since lines~\ref{repeat-sar}-\ref{until-sar} of the Repeat-Until loop do not change $S$ or $C_k$,
            and $p_k$ never exits this loop,
		 $p_a \in S$ and $C_k=c^* $ at all times $t' \ge t$.

	\begin{claim}\label{rakck-sar}
		There is a time $t'$ such that $R_{ak}$ contains $\langle -, c^* \rangle$ 
	   for all times after $t'$.
	\end{claim}
	\begin{proof}
	Let $\fresh()$ be the help procedure of $p_a$.
	Since $C_k = c^*$ for all times after $t$ and $p_a$ is correct,
		there is an iteration of the while loop of $\fresh()$
		in which $p_a$ reads $c^*$ from $C_k$ at line~\ref{collectck-sar}.
	Consider the \emph{first} iteration of the while loop of $\fresh()$ in which $p_a$ reads $c^*$ from $C_k$ at line~\ref{collectck-sar}.
	Since $p_k$ is correct, 
				by Observation~\ref{monock-sar},
				$p_a$ reads non-decreasing values from $C_k$,
				and so $p_a$ has $ c^* \ge prev\_c_k$ at line~\ref{askers-sar}.
	Since $c^* \ge 1$,
		there are two cases.
		\begin{itemize}
			\item Case 1: $c^*=1$.
			Then $p_a$ has $prev\_c_k \le 1$ at line~\ref{askers-sar}.
			Since $prev\_c_k$ is initialized to 0 (line~\ref{collectck-init-sar}) and $p_a$ reads $c^*=1$ from $C_k$ at line~\ref{collectck-sar} for the first time,
				$prev\_c_k = 0$ at line~\ref{askers-sar}. 
			So $p_a$ finds $c^*=1 > prev\_c_k=0$ holds and inserts $p_k$ into $askers$ at line~\ref{askers-sar}. 
			\item Case 2: $c^*>1$.
			Since $p_a$ reads $c^*$ from $C_k$ at line~\ref{collectck-sar} for the first time and $ c^* \ge prev\_c_k$,
				$c^* > prev\_c_k$ at line~\ref{askers-sar}. 
			So $p_a$ inserts $p_k$ into $askers$ at line~\ref{askers-sar}. 
		\end{itemize}
	So in both cases, $p_a$ inserts $p_k$ into $askers$ at line~\ref{askers-sar}. 
	Since $p_a$ is correct,
		in the same iteration of the while loop of $\fresh()$,
		$p_a$ writes $\langle -, c^* \rangle$ into $R_{ak}$ at line~\ref{fresh1-sar},
		say at time $t'$ 
		and then it sets $prev\_c_k$ to $c^*$ in line~\ref{setprev-sar} (i).
	Since $C_k = c^*$ for all times after $t$,
		by Observation~\ref{monock-sar},
		$C_k= c^*$ for all times after $t'$.
	Furthermore, by line~\ref{collectck-sar},
		$p_a$ has $c_k=c^*$ for all times after $t'$ (ii).
	From (i) and (ii),
		 by line~\ref{askers-sar},
		$p_a$ does not insert $p_k$ into $askers$ in any future iteration of the while loop of $\fresh()$.
	So by line~\ref{tellasker-sar},
		$p_a$ never writes to $R_{ak}$ after~$t'$,
		i.e., $R_{ak}$ contains $\langle -, c^* \rangle$ for all times after $t'$.
	\end{proof}

	Since $p_k$ executes lines~\ref{repeat-sar}-\ref{until-sar} infinitely often after $t$
		and  $p_a \in S$ for all times after $t' > t$,
		$p_k$ reads from $R_{ak}$ at line~\ref{readri-sar} infinitely often after $t$.
	By Claim~\ref{rakck-sar},
		eventually $p_k$ reads $\langle -, c^* \rangle$ from $R_{ak}$ after~$t$.
	Thus, since $p_k$ has $C_k =c^*$ at all times after $t$, 
		$p_k$ finds that the condition of line~\ref{until-sar} holds after~$t$
		--- a contradiction to ($\star$). 
\end{proof}

\begin{observation}\label{set1increment-sar}
	Let $\Test$ be an operation by a correct reader $p_k$.
	When $p_k$ executes line~\ref{setv-sar} of $\Test$, $p_k$ increments the size of $\setv$.
\end{observation}

\begin{observation}\label{set0increment-sar}
	Let $\Test$ be an operation by a correct reader $p_k$.
	When $p_k$ executes line~\ref{set0-sar} of $\Test$, $p_k$ increments the size of $\setb$.
\end{observation}

\begin{lemma}\label{termination-sar-2}
	Every $\Test$ operation by a {\ct} reader completes.
\end{lemma}

\begin{proof}
        Assume for contradiction, 
		there is a $\Test$ operation by a correct reader $p_k$ that does not complete,
		i.e.,
		$p_k$ takes an infinite number of steps in the $\Test()$ procedure.
	By Lemma~\ref{line7inf-sar},
		 $p_k$ must execute infinitely many iterations of the while loop at line~\ref{whileloop-sar} of $\Test()$.
        So $p_k$ executes line~\ref{setv-sar} or line~\ref{set0-sar} of $\Test()$
        infinitely often.	
	\begin{itemize}
		\item Case 1: $p_k$ executes line~\ref{setv-sar} infinitely often.	
				By Observations~\ref{monoset1-sar} and \ref{set1increment-sar},
					there is an iteration of the while loop at line~\ref{whileloop-sar} of this $\Test$,
					in which $p_k$ has $|\setv| = n$ at line~\ref{setv-sar}. 
                   Thus,
                        $\setv$ of $p_k$ contains at least $n-f$ correct processes at line~\ref{setv-sar} in that iteration.
                    So by Corollary~\ref{unqiuesetv-sar},
                        $\exists v$ such that
    $|\{p_i|\langle v,p_i\rangle\in \setv\}|\ge n-f$ at line~\ref{setv-sar} in that iteration.
				Since $p_k$ does not change $|\setv|$ between line~\ref{setv-sar} and line~\ref{return1-sar},
					in that iteration,
					$p_k$ finds that the condition $\exists v$ such that
    $|\{p_i|\langle v,p_i\rangle\in \setv \}|\ge n-f$ holds in line~\ref{return1-sar} of $\Test$.
				So $\Test$ returns the value $v$ at line~\ref{return1-sar}.
		\item Case 2: $p_k$ executes line~\ref{setv-sar} only a finite number of times.
			Then $p_k$ executes line~\ref{set0-sar} infinitely often.
			So there is a time $t$ such that 
				$p_k$ executes line~\ref{set0-sar} infinitely often after $t$
				while $p_k$ never executes line~\ref{setv-sar} after $t$.
			This implies $p_k$ never executes line~\ref{empty0-sar} after $t$.
			Since $|\setb|$ decreases only at line~\ref{empty0-sar},
				$|\setb|$ never decreases after $t$.
			Since $p_k$ executes line~\ref{set0-sar} infinitely often after $t$
				and $|\setb|$ never decreases after $t$,
				by Observations~\ref{set0increment-sar},
					there is an iteration of the while loop at line~\ref{whileloop-sar} of $\Test$,
					in which $p_k$ has $|\setb| > f$ at line~\ref{set0-sar} of $\Test$. 
			Since $p_k$ does not change $|\setb|$ between line~\ref{set0-sar} and line~\ref{return0-sar},
				in that iteration,
				$p_k$ finds that the condition $|\setb|>f$ holds in line~\ref{return0-sar}.
			So $\Test$ returns $\bot$ at line~\ref{return0-sar}.
	\end{itemize}
	In both cases,
		this $\Test$ returns, 
			a contradiction to the assumption that $\Test$ does not complete.
\end{proof}

By Lemma~\ref{termination-sar-1} and Lemma~\ref{termination-sar-2}, we have the following:

\begin{restatable}{theorem}{stickyt} \label{termination-sar}
    [\textsc{Termination}] Every $\Test$ and $\Set$ operation by a {\ct} process completes.
\end{restatable}

\subsection{Byzantine Linearizability}\label{linear-sar}

Recall that $H$ is an arbitrary history of the implementation given by Algorithm~\ref{code-sar}.
We now prove that $H$ is Byzantine linearizable with respect to a SWMR {\sr} (Definition~\ref{def-sar}).
We start by proving the following lemmas.

\begin{lemma}~\label{notreturn0-sar}
	Suppose for some value $v$, there is a time $t$ such that
		at least $f+1$ processes $p_i$ have $R_i = v$ at all times $t' \ge t$.
	If a {\ct} reader invokes a $\Test$ operation after time $t$, 
		then it does not insert any {\ct} process into $\setb$ in this $\Test$ operation.
\end{lemma}

\begin{proof}
	Assume for contradiction,
            for some value $v$,
		there is a time $t$ such that
		(1) at least $f+1$ processes $p_i$ have $R_i = v$ at all times $t' \ge t$, and
		(2) a {\ct} reader $p_k$ invokes a $\Test$ operation after time $t$,  
		but $p_k$ inserts a {\ct} process $p_a$ into $\setb$ in this $\Test$ operation.
 
	Consider the iteration of the while loop at line~\ref{whileloop-sar} of the $\Test$ in which $p_k$ inserts $p_a$ into $\setb$ at line~\ref{set0-sar} of the $\Test$. 
	Let $c^*$ be the value of $C_k$ after $p_k$ increments $C_k$ at line~\ref{ckplus-sar} in this iteration.
	Since $C_k$ is initialized to $0$, 
		$c^* \ge 1$ and $p_k$ writes $c^*$ into $C_k$ after time $t$.
	By Observation~\ref{monock-sar},
		for all times $t' \le t$, $C_k$ contains values that are less than $c^*$ ($\star$).
	Since $p_k$ inserts $p_a$ into $\setb$ at line~\ref{set0-sar},
		$p_k$ reads $\langle \bot, c \rangle$ from $R_{ak}$ with some $c \ge c^*$ at line~\ref{readri-sar}.
	Since $c^* \ge 1$, $c\ge 1$.
	Since $R_{ak}$ is initialized to $\langle \bot, 0 \rangle$ and $c\ge 1$,
		it must be that $p_a$ writes $\langle \bot, c \rangle$ into $R_{ak}$ 
		at line~\ref{fresh1-sar} in some iteration of the while loop of $\fresh()$ ($\star\star$).

	Consider the iteration of the while loop of $\fresh()$ procedure 
		in which {\ct} $p_a$ writes $\langle \bot, c \rangle$ into $R_{ak}$ at line~\ref{fresh1-sar}.
	In that iteration:	
		(a) $p_a$ reads $c$ from $C_k$ at line~\ref{collectck-sar}, say at time $t_a^1$ and then 
		(b) $p_a$ inserts $p_k$ into $askers$ at line~\ref{askers-sar}, say at time $t_a^2 > t_a^1$.
	Since $c \ge c^*$,
		by ($\star$),
		$p_a$ reads $c$ from $C_k$ at line~\ref{collectck-sar} after time $t$, i.e., $t_a^1 > t$.
	Since $t_a^2 > t_a^1$, $t_a^2 > t$, 
		i.e., $p_a$ inserts $p_k$ into $askers$ at line~\ref{askers-sar} at time $t_a^2 > t$.
	Thus, by (1) and Lemma~\ref{asker1-sar}, 
		$p_a$ writes $\langle v, c \rangle$ into $R_{ak}$ at line~\ref{fresh1-sar} --- a contradiction to~($\star\star$).
\end{proof}

\begin{lemma}\label{f+1line-sar}
    If a $\Test$ operation by a {\ct} reader returns $v$ at time $t$,
        then there are at least $f+1$ {\ct} processes $p_i$ that have $R_i = v$ at all times $t' \ge t$.
\end{lemma}
\begin{proof}
    Let $\Test$ be an operation by a {\ct} reader $p_k$ that returns $v$ at time $t$.
    Since $\Test$ returns $v$ at time $t$,
		$p_k$ finds the condition $|\{p_j|\langle v,p_j\rangle\in \setv\}|\ge n-f$ of line~\ref{return1-sar} holds in $\Test$, say at time $t_k < t$.
	Since $n \ge 3f+1$, $|\{p_j|\langle v,p_j\rangle\in \setv\}| \ge 2f+1$ at time $t_k$.
	Since there are at most $f$ faulty processes,
		for at least $f+1$ {\ct} processes $p_j$, $\langle v,p_j\rangle\in \setv$ of $p_k$ at time $t_k$.
	By Lemma~\ref{ri1beforeset1-sar},
		at least $f+1$ {\ct} processes $p_j$ have $R_j= v$ at all times $t' \ge t_k$.
	Since $t_k < t$,
		at least $f+1$ {\ct} processes $p_j$ have $R_j= v$ at all times $t' \ge t$.
\end{proof}

By Lemma~\ref{unqiuev-sar} and Lemma~\ref{f+1line-sar}, we have the following corollary.
\begin{corollary}\label{unique-sar}
    Let $\Test$ and $\Test'$ be operations by two {\ct} readers.
    If $\Test$ returns $v$ and $\Test'$ returns $v'$,
        then $v=v'$.
\end{corollary}

\begin{lemma}~\label{testreturn1-sar}
	Suppose there is a time $t$ such that at least $f+1$ processes $p_i$ have $R_i = v$ at all times $t' \ge t$.
	If a {\ct} reader invokes a $\Test$ operation after time $t$,  
		then $\Test$ returns $v$.
\end{lemma}

\begin{proof}
	Suppose that there is a time $t$ such that
		(1) at least $f+1$ processes $p_i$ have $R_i = v$ at all times $t' \ge t$, and 
		(2) a {\ct} reader $p_k$ invokes a $\Test$ operation after time $t$.
	By Lemma~\ref{notreturn0-sar},
            $p_k$ does not insert any {\ct} process into $\setb$ in this $\Test$ operation.
            Thus $\setb$ can contain only faulty processes, so 
		$|\setb| \le f$ in $\Test$.
	So $p_k$ never finds that $|\setb| > f$ holds in line~\ref{return0-sar} of $\Test$ and so $\Test$ never returns $\bot$ at line~\ref{return0-sar}.
	Thus,
		by Theorem~\ref{termination-sar} and the code of $\Test$ procedure,
		$\Test$ returns a value $v'$, say at time $t_r >t$.
        \begin{claim}
            $v=v'$.
        \end{claim}
        \begin{proof}
            Since $\Test$ returns a value $v'$,
                by Lemma~\ref{f+1line-sar},
                there are at least $f+1$ {\ct} processes $p_i$ that have $R_i = v'$ at all times $t' \ge t_r$.
            Since there are at most $f$ faulty processes,
                there is at least one {\ct} processes $p_i$ that has $R_i = v'$ at all times $t' \ge t_r$.
            Since there are at most $f$ faulty processes,
                by (1),
                there is at least one {\ct} processes $p_i$ that has $R_i = v$ at all times $t' \ge t$.
            Thus, by Lemma~\ref{unqiuev-sar},
                $v=v'$.
        \end{proof}
        Thus, by the above claim,
            $\Test$ returns $v$.
\end{proof}

\begin{theorem}\label{testatestb-sar}
	Let $\Test$ and $\Test'$ be operations by two {\ct} readers.
	If $\Test$ precedes $\Test'$ and $\Test$ returns $v$, then $\Test'$ returns $v$.
\end{theorem}

\begin{proof}
	Let $\Test$ and $\Test'$ be operations by two {\ct} readers $p_a$ and $p_b$ respectively.
	Suppose that $\Test$ precedes $\Test'$ and $\Test$ returns $v$.
	Let $t$ be the time when $\Test$ returns $v$.
	Since $\Test$ precedes $\Test'$,
		$p_b$ invokes $\Test'$ after time $t$ ($\star$).
	Since $\Test$ returns $v$ at time $t$,
		by Lemma~\ref{f+1line-sar},
		at least $f+1$ {\ct} processes $p_j$ have $R_j= v$ at all times $t' \ge t$.
	Thus, by ($\star$) and Lemma~\ref{testreturn1-sar}, $\Test'$ returns $v$.	
\end{proof}

\begin{definition}\label{t0t1-sar}
	\begin{itemize}
		\item Let $t_0$ be the max invocation time of any $\Test$ operation by a {\ct} reader that returns $\bot$;
		if no such $\Test$ operation exists, $t_0 = 0$.
		\item Let $t_1$ be the min response time of any $\Test$ operation by a {\ct} reader that returns any non-$\bot$ value;
		if no such $\Test$ operation operation exists, $t_1 = \infty$.
	\end{itemize}
\end{definition}

\begin{lemma}\label{iexists-sar}
      $t_1 > t_0$ and so the interval $(t_0, t_1)$ exists.
\end{lemma}
\begin{proof}
    There are four cases:
    \begin{itemize}
         \item Case 1: no $\Test$ operation by a {\ct} reader returns $\bot$,
            and no $\Test$ operation by a {\ct} reader returns a value $v$.
            Then $t_0 = 0$ and $t_1 = \infty$.
            So $t_1 > t_0$.
        \item Case 2: no $\Test$ operation by a {\ct} reader returns $\bot$ and some $\Test$ operation by a {\ct} reader returns a value $v$.
            Then $t_0 = 0$ and  $t_1 > 0$.
            So $t_1 > t_0$.
        \item Case 3: some $\Test$ operation by a {\ct} reader returns $\bot$
            and
            no $\Test$ operation by a {\ct} reader returns a value $v$.
            Then  $t_0 < \infty$ and $t_1 = \infty$.
            So $t_1 > t_0$.
            
        \item Case 4: some $\Test$ operation by a {\ct} reader returns $\bot$
        and
        some $\Test$ operation by a {\ct} reader returns a value $v$.
            Let $\Test'$ be the $\Test$ operation
            with the max invocation time of any $\Test$ operation by a {\ct} process that returns $\bot$. 
            By Definition~\ref{t0t1-sar}, the invocation time of $\Test'$ is $t_0$~($\star$).
            Let $\Test''$ be the $\Test$ operation
            with the min response time of any $\Test$ operation by a {\ct} process that returns some $v$.
            By Definition~\ref{t0t1-sar}, the response time of $\Test''$ is $t_1$ ($\star\star$).
                         
            Since $\Test''$ returns $v$ and $\Test'$ returns $\bot$, 
                by Theorem~\ref{testatestb-sar},
                $\Test''$ does not precede $\Test'$.
            So the response time of $\Test''$ is greater than the invocation time of $\Test'$.
            Thus, by ($\star$) and ($\star\star$),
                 $t_1 > t_0$.
    \end{itemize}
    Therefore in all cases,  $t_1 > t_0$ and so the interval $(t_0,t_1)$ exists.
\end{proof}

\begin{definition}
    Let $\hct$ be the set of processes that are {\ct} in the history $H$.
\end{definition}

\begin{definition}\label{hcstep-sar}
    Let $H|{\hct}$ be the history consisting of all the steps of all the {\ct} processes in $H$ (at the same times they occur in $H$).
\end{definition}

By Theorem~\ref{termination-sar}, all the processes that are correct in $H$ (i.e., all the processes in {\hct}) complete their operations, so:

\begin{observation}\label{allcomplete-sar}
    Every operation in $H|{\hct}$ is complete (i.e., it has both an invocation and a response).
\end{observation}

\begin{observation}\label{correctop-sar}
    An operation $o$ by a process $p$ is in $H|{\hct}$ if and only if $o$ is also in $H$ and $p \in {\hct}$.
\end{observation}

\begin{observation}\label{sameop-sar}
For all processes $p \in {\hct}$,
    $p$ has the same operations (i.e., the same invocation and response operation steps) in both $H$ and $H|{\hct}$.
    Furthermore, the $[invocation,response]$ intervals of these operations are the same in both $H$ and $H|{\hct}$.
\end{observation}

By the Definition~\ref{t0t1-sar} of $t_0$ and $t_1$, and Observation~\ref{sameop-sar}, we have the following.
\begin{observation}\label{sameip-sar-hc}
        In $H|{\hct}$,
	\begin{itemize}
		\item The max invocation time of any $\Test$ operation that returns $\bot$ is $t_0$.
		\item The min response time of any $\Test$ operation that returns any non-$\bot$ value is $t_1$.
	\end{itemize}
\end{observation}

To show that $H$ is Byzantine linearizable with respect to a SWMR \sr, 
    we must show that there is a history $H'$ such that:
    (a)~$H'|{\hct} = H|{\hct}$, and
    (b)~$H'$ is linearizable with respect to a SWMR {\sr}. 

There are two cases, depending on whether the writer is {\ct} in $H$.

\subsubsection*{Case 1: the writer $p_1$ is correct in $H$.}\label{case-writer-correct-sar}
Let $H' = H|{\hct}$.
We now show that the history $H'$ is linearizable with respect to a SWMR \sr.
To do so, 
    we first define the linearization points of operations in $H'$; 
and then we use these linearization points to define a linearization $L$ of $H'$ such that:
(a) $L$ respects the precedence relation between the operations of $H'$, and 
(b) $L$ conforms to the sequential specification of a SWMR {\sr} (which is given in Definition~\ref{def-sar}).

Since the writer $p_1$ is correct in $H$, we have the following observations and corollaries.
    \begin{observation}\label{onlywriteonce-sar}
In $H$, if the writer $p_1$ writes $v$ into $E_1$ at line~\ref{setter1-sar} in a $\Set(v)$ operation,
    then this $\Set(v)$ is the first\footnote{ When the writer is {\ct} in $H$, $\Set(-)$ operations by the writer are sequential in $H$. 
                So here the \emph{first} $\Set(-)$ operation is well-defined in $H$.} $\Set(-)$ operation in $H$.
\end{observation}

By Lemma~\ref{f+1line-sar} and Lemma~\ref{wbeforer-sar},
    we have the following corollary.
\begin{corollary}\label{nosettest-sar}
	In $H$, if a $\Test$ operation by a {\ct} reader returns $v$ at time $t$,
	   then $p_1$ writes $v$ into $E_1$ at line~\ref{setter1-sar} in a $\Set(v)$ operation before time $t$.
        Furthermore, this $\Set(v)$ is the first $\Set(-)$ in $H$.
\end{corollary}

Corollary~\ref{nosettest-sar} implies the following:
\begin{corollary}\label{nosettest-sar-1}
	In $H$, if a $\Test$ operation by a {\ct} reader precedes the first $\Set(-)$ operation by the writer (if any),
        then this $\Test$ returns $\bot$.
\end{corollary}

\begin{lemma}\label{valid-sar} 
    In $H$, if $\Set(v)$ is the first $\Set(-)$ operation and it precedes a $\Test$ operation by a {\ct} reader,
		then this $\Test$ returns $v$.
\end{lemma}
\begin{proof}
	Suppose $\Set(v)$ is the first $\Set(-)$ operation and it precedes a $\Test$ operation by a {\ct} reader in $H$.
    Since the writer $p_1$ is {\ct} and $\Set(v)$ is the first $\Set(-)$,
        $p_1$ finds that $|\{ r_i ~|~ r_i = v\}| \ge n-f$ holds in line~\ref{vconfirm-sar} of the $\Set(v)$,  say at time $t$,
        before it completes this $\Set(v)$.
    Then there are at least $n-f$ processes $p_i$ such that $p_1$ reads $R_i=v$ at line~\ref{writerreadsr-sar} by time $t$.
    Since there are at most $f$ faulty processes,
        there are at least $n-2f$ {\ct} processes $p_i$ that have $R_i=v$ before time $t$.
    Since $n>3f$,
        there are at least $f+1$ {\ct} processes $p_i$ that have $R_i=v$ before time $t$.
    By Observation~\ref{correctonlywrites1-sar},
        there are at least $f+1$ {\ct} processes $p_i$ that have $R_i=v$ at all times $t' \ge t$ ($\star$).
    Since the $\Set(v)$ precedes the $\Test$ operation by the {\ct} reader,
            the reader invokes this $\Test$ after time $t$.
    By $(\star)$ and Lemma~\ref{testreturn1-sar}, 
		this $\Test$ returns $v$.
\end{proof}

\begin{lemma}\label{tinip-sar}
	In $H$, if $\Set(v)$ is the first $\Set(-)$ operation,
        then there is a time $t$ in the $[invocation,response]$ interval of $\Set(v)$ that is also in the interval $(t_0,t_1)$.
\end{lemma}
\begin{proof}
    Suppose $\Set(v)$ is the first $\Set(-)$ operation in $H$.
    By Lemma~\ref{iexists-sar}, 
        the interval $(t_0,t_1)$ exists.
    Assume for contradiction that there is no time $t$ in the $[invocation,response]$ interval of $\Set(v)$ that is also in the interval $(t_0,t_1)$.
    There are two cases:
    \begin{compactitem}
        \item Case 1: the $[invocation,response]$ interval of $\Set(v)$ precedes the interval $(t_0,t_1)$. 
            This implies that the response time $t$ of $\Set(v)$ is less than $t_0$; so $t_0 \neq 0$.
            By Definition~\ref{t0t1-sar},
                $t_0$ is the invocation time of a $\Test$ operation by a correct reader that returns $\bot$.
            Since $t < t_0$,
                the $\Set(v)$ precedes this $\Test$.
            Thus, since $\Set(v)$ is the first $\Set(-)$ operation in $H$,
                by Corollary~\ref{valid-sar},
                the $\Test$ returns $v$ --- a contradiction.
        \item Case 2: the interval $(t_0,t_1)$ precedes the $[invocation,response]$ interval of $\Set(v)$.
             This implies that the invocation time $t$ of $\Set(v)$ is greater than $t_1$; so $t_1 \neq \infty$.
            By Definition~\ref{t0t1-sar},
                $t_1$ is the response time of a $\Test$ operation by a correct reader that returns a value $v$.
            Since $t > t_1$,
                this $\Test$ precedes the $\Set(v)$.
            Thus, since $\Set(v)$ is the first $\Set(-)$ operation in $H$,
                by Corollary~\ref{nosettest-sar-1},
                the $\Test$ returns $\bot$ --- a contradiction.
    \end{compactitem}
\end{proof}

The above Lemmas~\ref{nosettest-sar},~\ref{valid-sar}, and~\ref{tinip-sar},
    are about history $H$.
Since the writer is correct in $H$ and $H'=H|{\hct}$,
    by Observations~\ref{correctop-sar},~\ref{sameop-sar}, and~\ref{sameip-sar-hc},
    they also hold for the history $H'$, as~stated~below:

\begin{corollary}\label{nosettest-sar-c}
	In $H'$, if a $\Test$ operation by a {\ct} reader returns $v$ at time $t$,
	   then $p_1$ writes $v$ into $E_1$ at line~\ref{setter1-sar} in a $\Set(v)$ operation before time $t$.
        Furthermore, this $\Set(v)$ is the first $\Set(-)$ in $H'$.
\end{corollary}

\begin{lemma}\label{valid-sar-c}
    In $H'$, if $\Set(v)$ is the first $\Set(-)$ operation and precedes a $\Test$ operation by a {\ct} reader,
		then this $\Test$ returns $v$.
\end{lemma}

\begin{corollary}\label{tinip-sar-case1}
	In $H'$, if $\Set(v)$ is the first $\Set(-)$ operation,
        then there is a time $t$ in the $[invocation,response]$ interval of $\Set(v)$ that is also in the interval $(t_0,t_1)$.
\end{corollary}

We now prove that the history $H'$ is linearizable with respect to an \sr.
First, we define the \emph{linearization point} of each operation in $H'$ as follows.

\begin{definition}\label{linearpoints-sar}
Let $o$ be an operation in $H'$.
    \begin{itemize}
        \item If $o$ is a $\Test$ operation that returns $\bot$,
	   the linearization point of $o$ is the time of the invocation step of~$o$.
	
	\item If $o$ is a $\Test$ operation that returns $v$,
	   the linearization point of $o$ is the time of the response step of~$o$.

        \item If $o$ is the first $\Set(-)$ operation, 
            the linearization point of $o$ is any time $t$ in the $[invocation,response]$ interval of $o$ 
            that is also inside the interval $(t_0,t_1)$.
            Note that by Corollary~\ref{tinip-sar-case1},
                such a time $t$ exists.

        \item If $o$ is a $\Set(-)$ operation but not the first $\Set(-)$ operation, 
		the linearization point of $o$ is the time of the response step of~$o$.
\end{itemize}
\end{definition}

Note that in Definition~\ref{linearpoints-sar} the linearization point of every operation $o$ is between the invocation and response time of $o$. Thus, the following holds.

\begin{observation}\label{precedes-sar}
    If an operation $o$ precedes an operation $o'$ in $H'$,
        then the linearization point of~$o$ is before the linearization point of~$o'$.
\end{observation}

We now use the linearization points of all the operations in $H'$ to define the following linearization $L$ of $H'$:
\begin{definition}\label{hlp-sar}
	Let $L$ be a \emph{sequence} of operations such that:
	\begin{enumerate}
		\item\label{uno} An operation $o$ is in $L$ if and only if it is in $H'$.
		\item\label{due} An operation $o$ precedes an operation $o'$ in $L$ if and only if
     the linearization point of~$o$ is before the linearization point of $o'$ in $H'$.
	\end{enumerate}
\end{definition}

By Definition~\ref{hlp-sar}(\ref{uno}), $L$ is a linearization of $H'$.
By Observation~\ref{precedes-sar} and Definition~\ref{hlp-sar}(\ref{due}),
    $L$ respects the precedence relation between the operations of $H'$. More precisely:
\begin{observation}\label{precedes-l-sar}
    If an operation $o$ precedes an operation $o'$ in $H'$,
        then $o$ {\precedes} $o'$ in $L$.
\end{observation}

In the following lemmas and theorem we prove that the linearization $L$ of $H'$ conforms to the sequential specification of a SWMR \sr.

\begin{lemma}\label{mini}
    In $L$, 
	if a $\Test$ returns $\bot$,
	then no $\Set(-)$ {\precedes} the $\Test$.
\end{lemma}

\begin{proof}
	Assume for contradiction that, 
		in $L$,
		a $\Test$ returns $\bot$ but some $\Set(-)$ operation precedes the $\Test$.
        Let $\Set(v)$ be the first $\Set(-)$ operation in $L$.
        Then this $\Set(v)$ precedes the $\Test$ in $L$.
        By Definition~\ref{hlp-sar},
            these $\Set(v)$ and $\Test$ operations are also in~$H'$,
            and the $\Set(v)$ is the first $\Set(-)$ operation in $H'$.
        By Definition~\ref{linearpoints-sar},
            the linearization point of this $\Set(v)$ in $H'$ is a time $t \in (t_0,t_1)$.
        Since the $\Test$ returns $\bot$, by Definition~\ref{linearpoints-sar},
            the linearization point of this $\Test$ in $H'$ is the time $t'$ of its invocation step
            in $H'$.
        Since $H'=H|{\hct}$, 
            by Observation~\ref{sameip-sar-hc},
            $t' \le t_0$.
        Since $t \in (t_0,t_1)$ and $t' \le t_0$, $t' < t$.
        In other words, the linearization point of the $\Test$ is before the linearization point of the $\Set(v)$ in $H'$.
        So, by Definition~\ref{hlp-sar}(2),
            the $\Test$ {\precedes} the $\Set(v)$ in $L$ --- a contradiction.
\end{proof}

\begin{lemma}\label{mo}
	In $L$,
		  if a $\Test$ returns~$v$,
            then $\Set(v)$ is the first $\Set(-)$ operation and it {\precedes} the $\Test$.
\end{lemma}
\begin{proof}
	Suppose in $L$, a $\Test$ operation returns $v$.
	By Definition~\ref{hlp-sar}, this $\Test$ is also in $H'$.
        By Corollary~\ref{nosettest-sar-c},
             there is a $\Set(v)$ such that it is the first $\Set(-)$ operation in $H'$.
        Thus, by Definition~\ref{linearpoints-sar},
            the linearization point of this $\Set(v)$ in $H'$ is a time $t \in (t_0,t_1)$.
        Since the $\Test$ returns $v$, by Definition~\ref{linearpoints-sar},
            the linearization point of this $\Test$ in $H'$ is the time $t'$ of its response step in $H'$.
        Since $H'=H|{\hct}$,
            by Observation~\ref{sameip-sar-hc},
            $t' \ge t_1$.
        Since $t \in (t_0,t_1)$ and $t' \ge t_1$, $t < t'$.
        In other words, the linearization point of this $\Set(v)$ is before the linearization point of the $\Test$ in $H'$.
        So, by Definition~\ref{hlp-sar}(2),
            this $\Set(v)$ {\precedes} the $\Test$ in $L$. 
    Since this $\Set(v)$ is the first $\Set(-)$ operation in $H'$,
        by Observation~\ref{precedes-l-sar}, this $\Set(v)$ is also the first $\Set(-)$ operation in $L$.
        Thus, this $\Set(v)$ is the first $\Set(-)$ and it {\precedes} the $\Test$ in~$L$.
\end{proof}

By Lemmas~\ref{mini} and~\ref{mo}, $L$ conforms to the sequential specification of a SWMR {\sr}. More precisely:

\begin{theorem}\label{read-sar}
    In the linearization history $L$ of $H'$:
    \begin{compactitem}
        \item If a $\Test$ returns $v$ then $\Set(v)$ is the first $\Set(-)$ operation and it {\precedes} the $\Test$.
        \item If a $\Test$ returns $\bot$ then no $\Set(-)$ {\precedes} the $\Test$.
    \end{compactitem}
\end{theorem}

\begin{theorem}\label{linearthe-sar-case1}
    If the writer is correct in $H$, the history $H$ is Byzantine linearizable with respect to a SWMR {\sr}.
\end{theorem}

\begin{proof}
Assume that the writer $p_1$ is correct in $H$.
    To prove that $H$ is
        Byzantine linearizable with respect to a SWMR {\sr} (Definition~\ref{def-sar}),
    we must prove that there is a history $H'$ such that:
    
    \begin{compactenum}[(1)]
        \item $H'|{\hct} = H|{\hct}$.
        \item $H'$ is linearizable with respect to a SWMR {\sr}. 
    \end{compactenum}

    Let $H'= H|{\hct}$.
    Clearly, $H'$ satisfies~(1).
        
    Let $L$ be the linearization of $H'$ given in Definition~\ref{hlp-sar}.
    \begin{compactenum}[(a)]
        \item By Observation~\ref{precedes-l-sar}, $L$ respects the precedence relation between the operations of $H'$.
        \item By Theorem~\ref{read-sar}, $L$ conforms to the sequential specification of a SWMR {\sr}.
        \end{compactenum}
    Thus, $H'$ also satisfies~(2).
\end{proof}

\subsubsection*{Case 2: the writer $p_1$ is not correct in $H$.}\label{case-writer-not-correct-sar}

Recall that
$H$ is an arbitrary history of the implementation Algorithm~\ref{code-sar}, and
    {\hct} is the set of processes that are correct in $H$.
To show that $H$ is Byzantine linearizable with respect to a SWMR {\sr},
    we must prove that there is a history $H'$ such that (1)~$H'|{\hct} =H|{\hct}$ and 
(2)~$H'$ is linearizable with respect to a SWMR \sr.

We construct $H'$ from $H$ as follows.
We start from $H|{\hct}$.
    Since the writer is faulty in~$H$,
    $H|{\hct}$ \emph{does not contain any operation by the writer};
    in contrast, $H|{\hct}$ contains all the $\Test$ operations by the \emph{correct} processes in $H$.
So to ``justify'' the responses of these $\Test$ operations,
    we may have to add a $\Set(-)$ operation by the writer to $H|{\hct}$.
And we must add it such that
    resulting history conforms to the sequential specification of a {\sr}.
To do so, if there is a $\Test$ operation that returns $v$ in $H|{\hct}$,
      we add a $\Set(v)$ operation to $H|{\hct}$;
      we add it after the \emph{invocation} of the \emph{last} $\Test$ that returns $\bot$ and
        before the \emph{response} of the \emph{first} $\Test$ that returns $v$.
More precisely:

\begin{definition}\label{construct-sar}
    Let $H'$ be the history constructed from $H$ as follows.
    \begin{enumerate}[Step 1:]
	\item\label{s1} $H' = H|{\hct}$.
        \item\label{addset} If there is any $\Test$ operation that returns some value $v$ in $H'$,
		we add to $H'$ a single $\Set(v)$ operation (by the writer $p_1$) such that the $[invocation,response]$ interval of this $\Set(v)$ is inside the interval~$(t_0,t_1)$.
\end{enumerate}
\end{definition}

Since $p_1\not\in {\hct}$, we have the following:
\begin{observation}\label{hcorrectp-sar-2}
	$H'|{\hct} = H|{\hct}$.
\end{observation}	

To prove that $H$ is Byzantine linearizable, it now suffices to show that $H'$ is linearizable with respect to a SWMR {\sr} (Definition~\ref{def-sar}).

\begin{observation}\label{correctop-t-sar-2}
    A $\Test$ operation by a reader $p$ is in $H'$ if and only if this $\Test$ is in $H$ and $p \in {\ct}$.
\end{observation}

\begin{observation}\label{sameip-sar-2}
          In $H'$,
	\begin{itemize}
		\item The max invocation time of any $\Test$ operation that returns $\bot$ is $t_0$.
		\item The min response time of any $\Test$ operation that returns any non-$\bot$ value is $t_1$.
	\end{itemize}
\end{observation}

By Observation~\ref{correctop-t-sar-2} and Corollary~\ref{unique-sar}:
\begin{observation}\label{alltestv-sar-0-2}
     In $H'$, if $\Test$ returns $v$ and $\Test'$ returns $v'$,
	then $v = v'$.
\end{observation}

\begin{lemma}\label{setfromcorrect-sar}
     In $H'$, if a $\Test$ returns $v$,
        there is a $\Set(v)$ such that it is the first $\Set(-)$.
\end{lemma}
\begin{proof}
    Suppose in $H'$, a $\Test$ returns $v$.
    By Step 2 of Definition~\ref{construct-sar},
        there is a $\Set(v')$ for some value $v'$ in $H'$ and a $\Test'$ returns $v'$ in $H'$.
    By Observation~\ref{alltestv-sar-0-2},
        $v=v'$, i.e., there is a $\Set(v)$ in $H'$.
    Since the writer $p_1$ is not correct in $H$,
        no $\Set(-)$ operation is in $H|{\hct}$.
    Since at most one $\Set(-)$ operation is added in Step 2 of Definition~\ref{construct-sar},
         $H'$ has at most one $\Set(-)$ operation and so the $\Set(v)$ is the only $\Set(-)$ in $H'$.
    Thus, there is a $\Set(v)$ such that it is the first $\Set(-)$.
\end{proof}

\begin{observation}\label{allsetadd-sar}
    If a $\Set(-)$ operation is in $H'$,
        the $[invocation,response]$ interval of this $\Set(-)$ is in the interval $(t_0,t_1)$.
\end{observation}

We now prove that the constructed history $H'$ is linearizable with respect to an \sr.
First, we define the \emph{linearization point} of each operation in $H'$ as follows.
\begin{definition}\label{linearpoints-sar-2}
Let $o$ be an operation in $H'$.
    \begin{itemize}
        \item If $o$ is a $\Test$ operation that returns $\bot$,
	   the linearization point of $o$ is the time of the invocation step of~$o$.
	
	\item If $o$ is a $\Test$ operation that returns $v$,
	   the linearization point of $o$ is the time of the response step of~$o$.

        \item If $o$ is a $\Set(-)$ operation, 
            the linearization point of $o$ is any time $t$ in the $[invocation,response]$ interval of $o$.
\end{itemize}
\end{definition}

Note that in Definition~\ref{linearpoints-sar-2} the linearization point of every operation $o$ is between the invocation and response time of $o$. Thus, the following holds.

\begin{observation}\label{precedes-sar-2}
    If an operation $o$ precedes an operation $o'$ in $H'$,
        then the linearization point of $o$ is before the linearization point of $o'$.
\end{observation}

By Definition~\ref{linearpoints-sar-2} and Observation~\ref{allsetadd-sar}, we have the following:
\begin{observation}\label{setl-sar}
     If a $\Set(-)$ operation is in $H'$,
        the linearization point of this $\Set(-)$ in $H'$ is in the interval $(t_0,t_1)$.
\end{observation}

We now use the linearization points of all the operations in $H'$ to define the following linearization $L$ of $H'$:
\begin{definition}\label{hlp-sar-2}
	Let $L$ be a \emph{sequence} of operations such that:
	\begin{enumerate}
		\item\label{uno} An operation $o$ is in $L$ if and only if it is in $H'$.
		\item\label{due} An operation $o$ {\precedes} an operation $o'$ in $L$ if and only if
     the linearization point of~$o$ is before the linearization point of $o'$ in $H'$.
	\end{enumerate}
\end{definition}

By Observation~\ref{precedes-sar-2} and Definition~\ref{hlp-sar-2}(\ref{due}),
        $L$ respects the precedence relation between the operations of $H'$. More precisely:

\begin{observation}\label{precedes-l-sar-2}
    If an operation $o$ precedes an operation $o'$ in $H'$,
        then $o$ {\precedes} $o'$ in $L$.
\end{observation}

In the following lemmas and theorem we prove that the linearization $L$ of $H'$ conforms to the sequential specification of a SWMR \sr.

\begin{lemma}\label{Ouahed}
    In $L$, 
	if a $\Test$ operation returns $\bot$,
	then no $\Set(-)$ operation {\precedes} the $\Test$ operation.
\end{lemma}

\begin{proof}
	Assume for contradiction that, 
		in $L$,
		a $\Test$ operation returns $\bot$ but a $\Set(v)$ operation {\precedes} the $\Test$ operation.
        By Definition~\ref{hlp-sar-2}(1),
            these $\Set(v)$ and $\Test$ operations are also in $H'$.
        By Observation~\ref{setl-sar},
             the linearization point of this $\Set(v)$ operation in $H'$ is a time $t \in (t_0,t_1)$.
         Since the $\Test$ operation returns $\bot$, by Definition~\ref{linearpoints-sar-2},
            the linearization point of this $\Test$ operation in $H'$ is the time $t'$ of the invocation step of the $\Test$ operation in $H'$.
        By Observation~\ref{sameip-sar-2},
            $t' \le t_0$.
        Since $t \in (t_0,t_1)$ and $t' \le t_0$, $t' < t$.
        In other words, the linearization point of the $\Test$ operation is before the linearization point of the $\Set(v)$ operation in $H'$.
        So, by Definition~\ref{hlp-sar-2}(2),
            the $\Test$ operation {\precedes} the $\Set(v)$ operation in $L$ --- a contradiction.
\end{proof}

\begin{lemma}\label{Tnenen}
	In $L$,
		  if a $\Test$ operation returns~$v$,
            then $\Set(v)$ operation is the first $\Set(-)$ operation and it {\precedes} the $\Test$ operation.
\end{lemma}
\begin{proof}
	Suppose in $L$, a $\Test$ operation returns $v$.
	By Definition~\ref{hlp-sar-2}(1), this $\Test$ operation is also in $H'$.
        Since it returns $v$, by Definition~\ref{linearpoints-sar-2},
            the linearization point of this $\Test$ operation in $H'$ is the time $t'$ of its response step in $H'$.
        By Observation~\ref{sameip-sar-2},
            $t_1 \le t'$ ($\star$).
            
        Since this $\Test$ operation returns $v$ in $H'$,
            by Lemma~\ref{setfromcorrect-sar},
            there is a $\Set(v)$ operation such that it is the first $\Set(-)$ operation in $H'$.
        By Definition~\ref{hlp-sar-2}(\ref{uno}) and Observation~\ref{precedes-l-sar-2},
            this $\Set(v)$ operation is also the first $\Set(-)$ operation in $L$ ($\star\star$).
        By Observation~\ref{setl-sar},
             the linearization point of this $\Set(v)$ operation in $H'$ is a time $t \in (t_0,t_1)$,
             and so $t <t_1$.
        By ($\star$), $t < t'$.
        In other words, the linearization point of this $\Set(v)$ operation is before the linearization point of the $\Test$ operation in $H'$.
        So, by Definition~\ref{hlp-sar-2}(2),
            this $\Set(v)$ operation {\precedes} the $\Test$ operation in $L$.
        Thus, by ($\star\star$), this $\Set(v)$ operation is the first $\Set(-)$ operation and it {\precedes} the $\Test$ operation in $L$.
\end{proof}

By Lemmas~\ref{Ouahed} and~\ref{Tnenen}, $L$ conforms to the sequential specification of a SWMR {\sr}. More precisely:

\begin{theorem}\label{read-sar-2}
    In the linearization history $L$ of $H'$:
    \begin{compactitem}
        \item If a $\Test$ operation returns $v$ then $\Set(v)$ is the first $\Set(-)$ operation and it {\precedes} the $\Test$ operation.
        \item If a $\Test$ operation returns $\bot$ then no $\Set(-)$ operation {\precedes} the $\Test$ operation.
    \end{compactitem}
\end{theorem}

\begin{theorem}\label{linearthe-sar-case2}
    If the writer is not correct in $H$,
    the history $H$ is Byzantine linearizable with respect to a SWMR {\sr}.
\end{theorem}

\begin{proof}
Assume that the writer $p_1$ is not correct in $H$.
    To prove that $H$ is
        Byzantine linearizable with respect to a SWMR {\sr} (Definition~\ref{def-sar}),
    we must prove that there is a history $H'$ such that:
    \begin{compactenum}[(1)]
        \item $H'|{\hct} = H|{\hct}$.
        \item $H'$ is linearizable with respect to a SWMR {\sr}. 
    \end{compactenum}
    Let $H'$ be the history given by Definition~\ref{construct-sar}.
    By Observation~\ref{hcorrectp-sar-2}, $H'$ satisfies (1).
    
    Let $L$ be the linearization of $H'$ given in Definition~\ref{hlp-sar-2}.
    \begin{compactenum}[(a)]
        \item By Observation~\ref{precedes-l-sar-2}, $L$ respects the precedence relation between the operations of $H'$.
        \item By Theorem~\ref{read-sar-2}, $L$ conforms to the sequential specification of a SWMR {\sr}.
        \end{compactenum}
    Thus, $H'$ also satisfies~(2).
\end{proof}

By Theorems~\ref{linearthe-sar-case1} and~\ref{linearthe-sar-case2}, the arbitrary history $H$ of Algorithm~\ref{code-sar} is Byzantine linearizable with respect to an {\sr} (independent of whether the writer $p_1$ is correct or not in $H$):

\begin{restatable}{theorem}{verifybl}\label{linearthe-sar}
    [\textsc{Byzantine linerizability}] The history $H$ is Byzantine linearizable with respect to a SWMR {\sr}.
\end{restatable}

By Theorem~\ref{termination-sar} and Theorem~\ref{linearthe-sar},
    we have the following:

\stickycorrect*

\end{document}